\documentclass[pra,showpacs,aps,twocolumn]{revtex4}
\usepackage{amsmath}
\usepackage{graphicx}
\usepackage{dcolumn}
\usepackage{amssymb}
\usepackage{color}

\usepackage{graphicx,comment}
\usepackage{epstopdf}
\usepackage{bm}
\usepackage{enumerate}
\usepackage{fancyhdr,color}
\usepackage{verbatim}
\usepackage{amsmath,amsthm}
\usepackage{amsfonts}
\usepackage{hyperref}

\def\rD{{\rm d}}

\newtheorem{theorem}{Theorem}
\newtheorem{corollary}{Corollary}
\newtheorem{lemma}{Lemma}

\def\one{{\mathchoice {\rm 1\mskip-4mu l} {\rm 1\mskip-4mu l} {\rm
1\mskip-4.5mu l} {\rm 1\mskip-5mu l}}}

\newcommand{\ket}[1]{\left| #1\right\rangle}        
\newcommand{\sket}[1]{| #1\rangle}  
\newcommand{\bra}[1]{\left\langle #1\right|}        
\newcommand{\sbra}[1]{\langle #1|}  

\begin{document}
LA-UR-15-20573
\title{Quantum simulations of one dimensional quantum systems}


\author{Rolando D. Somma}
\affiliation{Los Alamos National Laboratory, Los Alamos, NM 87545, USA}

\begin{abstract}
We present  quantum algorithms
for the simulation of quantum systems in one spatial dimension, which result in quantum speedups
that range from superpolynomial to polynomial.
We first describe a method to simulate the evolution of the quantum harmonic oscillator (QHO)
based on a refined analysis of the Trotter-Suzuki formula that exploits the Lie algebra structure.  For total evolution time $t$ and precision $\epsilon>0$, the complexity of our method is 
$ O(\exp(\gamma \sqrt{ \log(N/\epsilon)}))$, where $\gamma>0$ is a constant and $N$ is the quantum number
associated with an ``energy cutoff'' of the initial state. Remarkably, this complexity is subpolynomial in $N/\epsilon$.
We also provide a  method to prepare 
discrete versions of the eigenstates of the QHO of complexity polynomial in $\log(N)/\epsilon$,
where $N$ is the dimension or number of points in the discretization. This method may be of independent interest as it provides a way to prepare, e.g., quantum states with Gaussian-like amplitudes. Next, we consider a system with a quartic potential.
Our numerical simulations suggest a method for simulating the evolution of sublinear complexity $\tilde O(N^{1/3+o(1)})$, for constant $t$ and $\epsilon$.
We also analyze  complex 
 one-dimensional systems and prove a complexity bound $\tilde O(N)$,
under fairly general assumptions. 
 Our  quantum algorithms may find applications in other problems. As an example,
 we discuss the fractional Fourier transform, a generalization of the Fourier transform that is useful for signal analysis and can be formulated in terms of the evolution  of the QHO.
 \end{abstract}

\maketitle

\section{Introduction}
One of the best known problems that a quantum computer  is expected to tackle more efficiently
than a classical one is quantum simulation (QS)~\cite{feynman_simulating_1982,lloyd_universal_1996,BMK2010}. QS provides insight to the behavior of a quantum system, hence it serves as a powerful 
method applicable to a variety of scientific areas including physics~(c.f., \cite{ortiz:qc2001a,somma_physics_2003}),
  quantum chemistry~(c.f. \cite{aspuru-guzik:qc2005a,kassal_chem_2011,wecker_chem_2013,PHW2014}), and computer science (c.f.,~\cite{DJRZ2006}).
One way to explore this problem is by solving 
 Schr\"{o}dinger's equation, so that the state of the system evolves under
the evolution operator, as determined by the  Hamiltonian of the system.
A quantum simulator \cite{GACZ2000,Fri2008} carries out this simulation task on a quantum 
system or computer by an approximate implementation of the evolution, either as a sequence of elementary gates (i.e.,
a digitial quantum simulator)
or by direct implementation of the Hamiltonians (i.e., an analog quantum simulator);
this paper is concerned  with digital quantum simulators.

Simulating a quantum system requires a certain amount of ``time and memory'', called resources, 
and the important question is how the resources scale with quantities such as the volume or size
of the system (or the dimension of the Hilbert space), and precision. Generally, when the simulation
is carried
on a classical computer, it would require an undesirably large 
amount of resources, which R.P. Feynman addressed as an ``exponential explosion'', and
 it rapidly becomes an intractable problem even for relatively small systems.  On the contrary,
 a quantum simulator is expected to carry this simulation efficiently,
 at least in some important cases~\cite{feynman_simulating_1982}.
 
A QS may be implemented for two goals:
 i- to compute time-dependent physical properties, such as scattering amplitudes, 
 as determined by the evolution of the quantum system, and ii- to compute spectral properties,
such as eigenvalues or expectation values of observables on different eigenstates of the system.  
In the first goal, we are mainly concerned with the  simulation of the evolution operator $U(t):=\exp \{-iHt \}$.
Results on the so-called Hamiltonian simulation problem provide efficient quantum algorithms to simulate $U(t)$ on quantum
computers~\cite{aharonov_adiabatic_2003,berry_efficient_2007,childs_thesis,
childs_efficient_2010,wiebe_product_2010}. Generally,
the complexity of such algorithms (i.e., the number of elementary two-qubit gates and queries) 
depends only polynomially
in $\tau$, $d$, and $1/\epsilon$, where $\tau=\| H  t \|$, $d$ is the sparseness 
of $H$, and $\epsilon>0$ is a precision parameter. The Trotter-Suzuki approximation~\cite{trotter_1959,suzuki_90}
plays an important role in these results. The complexity dependence on $1/\epsilon$ can be exponentially improved
using recent methods that implement a series approximation of $U(t)$~\cite{berry_sparse,BCC_Taylor_2015}.
In the second goal, we are mainly concerned with the preparation
of eigenstates~\cite{PW2009,schwarz_peps_2011,somma_gap_2013}.
Two commonly used techniques for these quantum algorithms 
are the simulation of quantum adiabatic evolutions~\cite{farhi_quantum_2000} 
and the implementation of the so-called phase estimation
algorithm~\cite{kitaev_quantum_1995}. But since approximating
the ground state energy of many-body systems is a computationally
hard problem~\cite{kitaev_computation_2002},   quantum algorithms
to prepare ground states are generally inefficient.

The quest for more efficient methods 
for QS is  ongoing,
particularly sparked for the need to understand the quantum advantages of relatively
small, potentially near-future, quantum computers~(c.f., \cite{MM2014}).
 Unfortunately, most quantum algorithms for QS
up to date concern the case of discrete, finite dimensional
quantum many-body systems where $\|H \| < \infty$. These algorithms cannot be directly
applied to speedup the simulation of continuous-variables (CVs) quantum systems: the effective energy scale
of the problem can increase polynomially with the dimension of the relevant Hilbert space.
While a few exceptions exist (e.g.,~\cite{papageorgiou_CV_2013,jordan_QFT_2012}), 
a detailed analysis of the power and limits of quantum computers
for QS of CV systems is lacking. Here, we incentivize such an analysis 
by studying a quantum-computer simulation of  simple one-dimensional quantum systems.

   We first consider the 
quantum harmonic oscillator (QHO). Remarkably, our first  result is a 
 quantum algorithm that simulates $U(t)$
  with subexponential complexity. That is, the complexity is $ O(\exp(\gamma \sqrt{\log(N/\epsilon)}))$, 
  where $N$ is the dimension
  (as determined by the discretization and effective energy of the problem) and $\epsilon$ is the precision. 
  This complexity does not depend on the evolution time
  because the evolution operator of the QHO is periodic. The first result is obtained via a refined analysis of the Trotter-Suzuki
  approximation that exploits the Lie algebra structure of this problem,
  allowing us to bypass the no fast-forwarding theorem of~\cite{berry_efficient_2007}, which would 
  indicate a complexity that is at least linear in $N$.
  Our results are in contrast with, for example, the results in~\cite{papageorgiou_CV_2013,jordan_QFT_2012},
  where the complexity is polynomial in $N$. 
  Our result suggests a quantum superpolynomial speedup over the 
  corresponding classical algorithms for this problem,
  showing a significant quantum advantage even for the simulation of this simple model.

  We then present quantum algorithms of complexity
  polynomial in $\log (N)/\epsilon$ that prepare approximations
  of the eigenstates of a discrete version of the QHO.  These
  states are prepared by simulating the evolution operator of a Hamiltonian
  that is a discrete version of the Jaynes-Cummings model,  using the Trotter-Suzuki approximation.
  In particular, for the preparation of the eigenstate of lowest eigenvalue, the complexity is polynomial in $\log(N/\epsilon)$.
  The computation
  of spectral properties of the QHO can then be done by computing
  expectation values on such states using a variety of known techniques (c.f.,~\cite{knill_expectation_2007}). 
  These algorithms may be of independent interest as they allow for a 
  very efficient way to prepare states whose amplitudes are,
  for example, Gaussian-like.
  
  In order to understand the complexity of simulating more complex quantum systems, 
  we perform several numerical simulations of a quantum system with a quartic potential.
  In contrast with the QHO, quantum algorithms to simulate $U(t)$ based on high-order Trotter-Suzuki approximations
  seem to be of complexity sublinear in $N$ for this case. In particular, our numerical simulations suggest that
  such a complexity is $\tilde O (|t|^{1+\eta} N^{1/3+4 \eta/3} /\epsilon^\eta)$ for arbitrarily small, but constant, $\eta>0$.
  The $\tilde O$ notation hides factors that are polynomial in $\log(N|t|/\epsilon)$.
  This complexity represents a polynomial quantum speedup, with respect to $N$, over the classical algorithm.
  The sublinear complexity in $N$ may be a result of the algebraic structure satisfied by the operators in the Hamiltonian
  (see the recent results in~\cite{Som_15} for more details).
  
  Finally, we use our results in~\cite{BCC_Taylor_2015} to present a generic quantum method to simulate the dynamics
  of one dimensional quantum systems of complexity $\tilde O(|t|N)$, under fairly general assumptions.
  In contrast with our previous results, the method in~\cite{BCC_Taylor_2015} implements a series approximation
  of the evolution operator and the complexity dependence on $1/\epsilon$ is polylogarithmic. Our main contribution
  in this case is to represent the evolution operator in the interaction picture, so that potentials of high degree
  do not increase the complexity significantly. We note that classical algorithms for this problem are expected to have
  complexity that is super-linear in $N$.

The remainder of the paper is organized as follows. 
In Sec.~\ref{sec:CVQHO} we revisit the QHO where we discuss
particular properties that are useful for our quantum algorithms, and also classical
algorithms for simulating this model. In Sec.~\ref{sec:DQHO} 
we define a discretization of the QHO as the starting point for a QS.
We provide a quantum algorithm to compute scattering amplitudes associated
with the QHO in Sec.~\ref{sec:DQHOexpval} and quantum algorithms to prepare
the eigenstates of the discrete QHO in Sec.~\ref{sec:ESprep}. We analyze the 
quantum system with a quartic potential in Sec.~\ref{sec:quartic} and present the upper bound
on the complexity of simulating general one-dimensional quantum systems in Sec.~\ref{sec:bound}.
In Sec.~\ref{sec:frft},  we discuss related work with particular emphasis on the so-called
fractional Fourier transform, which is a generalization of the Fourier transform
used in signal analysis~\cite{OZA01}. We state the conclusions in Sec.~\ref{sec:conc}.

\section{The QHO revisited}
\label{sec:CVQHO}
The QHO is   described by its Hamiltonian ($\hbar =1$)~\cite{AB95},
\begin{align}
 H = \frac 1 2 [x^2 + (-i \partial_x)^2] \;.
 \end{align}
Because $x \in (-\infty, \infty)$, we sometimes refer to $H$ as the continuous-variable
or CV QHO.  The eigenfunctions of $H$ are the so-called  Hermite functions, 
 \begin{align} \label{eq:psi_n}
 \psi_n (x) = \frac {1} {\sqrt{n! 2^n \sqrt \pi}} e^{ -\frac{x^2}{2}} H_n(x) \; ,
 \end{align}
 where each $H_n(x)$ is the  $n$-th (physicists') Hermite polynomial and $n = 0, \, 1, \, 2, ...$ are quantum numbers.
 Then,
 \begin{align}
 H \psi_n(x) = (n+1/2) \psi_n (x) \; ,
 \end{align}
where  $n+1/2$ are the corresponding eigenvalues.  
In standard bra-ket notation, we represent the eigenfunctions $\psi_n(x)$ by the eigenstates $\sket{\psi_n}$
and the QHO, in operator form, is 
\begin{align}
H=\frac 1 2 (\hat x^2 + \hat p^2) \;.
\end{align}
 Then, $H \sket{\psi_n}=(n+1/2)\sket{\psi_n}$, and the position and momentum operators satisfy 
 \begin{align}
 \sbra{\psi_m} \hat x \sket{\psi_n} & = \int dx \; \psi_m(x) (x \psi_n(x)) \; , \\
 \sbra{\psi_m} \hat p \sket{\psi_n} & =-i \int dx \; \psi_m(x) (\partial_x \psi_n(x) )\;,
 \end{align}
  respectively. The evolution operator of the QHO for time $t$ is $U(t)=\exp\{-iHt\}$. Since
  the evolution operator is periodic, i.e. $U(0)=U(4 \pi)$, we can assume $t=O(1)$.
  We also use $\ket x$ to denote the eigenstates of $\hat x$ of eigenvalue $x \in (-\infty,\infty)$.
  
 We can alternatively define the raising and lowering operators,
 \begin{align}
 a^\dagger  =( \hat x -i \hat p)/\sqrt 2 \; , \;
 a  = ( \hat x +i \hat p)/\sqrt 2 \; ,
 \end{align} 
and write $H=a^\dagger a +1/2$. The eigenstates satisfy $a^\dagger a \sket{\psi_n} = n \sket{\psi_n}$ and
$a^\dagger \sket{\psi_n} = \sqrt{n+1} \sket{\psi_{n+1}}$.  Then, other eigenstates of the QHO can be obtained
by repeated action of $a^\dagger$ on the vacuum (ground) state $\sket{\psi_0}$. Later, we will use this property
to devise a quantum algorithm that prepares eigenstates, up to some approximation error (Sec.~\ref{sec:ESprep}).

A well-known result states that $\hat x$ and $\hat p$ are related via the Fourier transform (FT), which transforms $\hat x \rightarrow -\hat p$ and
$\hat p \rightarrow  \hat x$ (i.e., a $\pi/2$ rotation in phase space). Similarly, the FT transforms $a^\dagger \rightarrow - i a^\dagger$ and $a \rightarrow i a$. One implication is that the eigenstates $\sket{\psi_n}$ of the QHO are also eigenstates of the FT, and direct 
computation shows that the eigenvalues of the FT are $(-i)^n$. We will use this property to prove some results regarding the 
quantum algorithm that simulates the evolution of the QHO (Appx.~A).

It is also important to remark that the operators in $H$ generate a Lie algebra $sp(2)$ of dimension 3.
This property will be useful to bound the errors when approximating the evolution operator in Appx.~B. The result
is that the effective
norm of nested commutators is significantly smaller than the product of the effective norms, allowing us
to perform a refined analysis of the errors in Hamiltonian simulation methods.

Of particular interest in a QS
is the computation of quantities such as scattering amplitudes, $\bra{\varphi'} \varphi(t) \rangle$. Here, 
$\sket{\varphi(t)}:=U(t) \sket{\varphi}$ is the evolved state and $\ket{\varphi}$, $\ket{\varphi'}$
are some other specific states of the system, such as  eigenstates of the position or momentum operator, or more general states.
Also important is the computation of expectation values or correlation functions in different eigenstates,
namely $\sbra{\psi_{m}} \hat A \sket{\psi_n}$, where $\hat A$ is some observable; e.g., $\hat A= \hat p^{l_1} \hat x^{l_2}$,
where $l_1,l_2 \ge 0$.
Our quantum algorithms are designed to compute such quantities, but they could also be used
for the computation of more general quantities, such as expectation values of other unitary operators, after minor and straightforward modifications.

With no loss of generality, $\sket{\varphi} = \sum_{n=0}^{N'} c_n \sket{\psi_n}$ and $\sket{\varphi'} = \sum_{n=0}^{N'} c_n' \sket{\psi_n}$ 
(with $\sum_n |c_n|^2=\sum_n |c_n'|^2=1$), so that
\begin{align}
\label{eq:scatamp}
\bra{\varphi'} \varphi(t) \rangle = \sum_{n=0}^{N'} (c'_n)^* c_n e^{-i(n+1/2)t} \; .
\end{align}
Then, a standard classical method to compute  scattering amplitudes involves
the spectral decomposition of   $\sket{\varphi}$ and $\sket{\varphi'}$ in terms of $\sket{\psi_n}$,
to obtain the amplitudes $c_n$ and $c'_n$. Because the sum in Eq.~\eqref{eq:scatamp} involves $O(N')$ terms,
the worst-case complexity of the classical method is of order polynomial in $N'$ if $N'<\infty$.
We also note that
\begin{align}
\label{eq:expecvalue}
\sbra{\psi_{m}} \hat A \sket{\psi_n} = \int \int dx' dx \; \psi_m(x') A(x',x) \psi_n(x) \; ,
\end{align}
where $A(x',x)=\bra {x'} \hat A \ket x$. Assuming $n,m \le N'$, classical methods to 
approximate the integral in Eq.~\eqref{eq:expecvalue} can also have a worst-case complexity 
of order polynomial in $N'$. This is because the functions $\psi_n(x)$ present oscillations
that become more significant as $n$ increases, so a good approximation to the integral
by a finite sum can only be obtained if the number of terms in the sum is polynomial in $N'$~\cite{papageorgiou_CV_2013}.
Nevertheless, it is well-known that many quantities associated with the QHO can be analytically obtained,
so our quantum algorithms will be more powerful than classical ones only in certain scenarios.

A natural quantum method to compute Eq.~\eqref{eq:scatamp} involves direct QS
of the QHO. In this case, if the states $\sket{\varphi}$ and $\sket{\varphi'}$ can be efficiently prepared, the quantum method is efficient or not
whether $U(t)$ can be simulated efficiently or not,  respectively.  In this paper
we are first interested in devising efficient quantum-computer simulations of the 
evolution induced by the QHO and then we will consider quantum algorithms for preparing eigenstates.

\section{A discrete quantum harmonic oscillator}
\label{sec:DQHO}
In analogy with the CV QHO, we define a discrete QHO by the Hamiltonian
  \begin{align}
  \label{eq:DQHO1}
  H^\rD=\frac 1 2 ((x^\rD)^2 +( p^\rD)^2)\; .
  \end{align}
 The Hilbert space dimension is $N$, where $N \ge 2$ is even for simplicity.
 $x^\rD$ is the discrete ``position'' operator given by the $N \times N$ diagonal matrix 
  \begin{align}
  x^\rD= \sqrt{\frac {2 \pi}N} \frac 1 2  \begin{pmatrix} -N & 0 & \cdots & 0 \cr 0 & (-N+2) & \cdots & 0 
  \cr \vdots & \vdots & \ddots & \vdots \cr 0 & 0 & \cdots & (N-2) \end{pmatrix} \; ,
  \end{align}
  and $p^\rD$ is the discrete ``momentum'' operator given by 
  \begin{align}
  p^\rD=  (F_{\rm c}^{\rD})^{-1}. x^\rD . F_{\rm c}^{\rD}  \;.
  \end{align}
  The $N \times  N$ unitary matrix $F_{\rm c}^{\rD}$ is the so-called centered discrete Fourier transform, 
  which is the standard discrete Fourier transform $F^{\rD}$ up to a simple (cyclic) permutation. Its matrix entries
  are
  \begin{align}
  \left [ F_{\rm c}^{\rD} \right]_{j,k}= \frac 1 {\sqrt{ N}} {\exp (i  \; 2 \pi jk/N)}\; ,
  \end{align}
  where $j,k \in \{-N/2, \ldots, N/2-1\}$ label the rows and columns, respectively.   Then, $F_{\rm c}^{\rD} = (X)^{N/2}. F^\rD. (X)^{-N/2}$, and
  \begin{align}
  \label{eq:cyclicP}
  X=\begin{pmatrix} 0 & 1 & 0 &\ldots & 0 \cr 0 & 0 & 1 & \ldots & 0  \cr \vdots & \vdots & \vdots & \ddots & \vdots \cr
  0 & 0 & 0 & \ldots & 1 \cr
  1 & 0 & 0 & \ldots & 0 \end{pmatrix} \; 
  \end{align}
  is the operation that performs a cyclic permutation, shifting the indices by one.
The relation between  $F_{\rm c}^{\rD}$ and $F^\rD$
will be useful to provide an efficient quantum circuit that implements $F_{\rm c}^{\rD}$ (Sec.~\ref{sec:TSdecomp}).

 We will generally assume that we are in the limit of large $N$, although some results do still apply when $N$ is fairly small.
 Then, in the following, the order notation to bound approximation errors assume the asymptotic limit; we refer to the corresponding appendices for more details.

\subsection{Spectral properties}
\label{sec:spectralprop}
In contrast to the CV QHO, the Hamiltonian $H^\rD$ may not be exactly solvable. However, some spectral properties of $H^\rD$ can be well approximated from those of the CV version. We write $E_n^\rD$ and $\sket{\phi_n^\rD}$ for the eigenvalues and 
eigenstates of $H^\rD$, respectively, and $n =0,1,\ldots,N-1$.  We also introduce the (unnormalized) quantum states
  \begin{align} \label{eq:psi_error}
  \sket{\psi_n^\rD} =\left(\frac{2 \pi}N\right)^{1/4}\sum_{j=-N/2}^{N/2-1} \psi_n(x_j) \ket j \; ,
  \end{align}
  which, as we will show, approximate the eigenstates of $H^\rD$. Here, $x_j=j \sqrt{2 \pi/N}$ and
  $\psi_n(x)$ is the $n$-th Hermite function, so that  $\sket{\psi_n^\rD}$
  represents a  ``discrete Hermite state''. In Appx.~A, Cor.~\ref{cor:maincor}, we show that
 there exists a constant $c$, $1> c>0$, such that
 \begin{align}
 \label{eq:approxeigen}
  \| H^{\rD}   \sket{\psi_n^\rD} -  (n+1/2)  \sket{\psi_n^\rD} \|^2 = \exp(-\Omega(N)) \; ,
  \end{align}
  for all $ n \le cN$. Here, for a matrix $A$, $\| A\|$  is the spectral norm and $\| \ket{\xi}\|$
  is the Euclidean norm of a state $\ket \xi$ .
  The notation $\exp(-\Omega(N))$ states that there is a constant $\beta>0$ such that
  the right hand side of Eq.~\eqref{eq:approxeigen} is at most $ e^{- \beta N}$, for  sufficiently large $N$.

Since $c<1$, we refer to the subspace spanned by $\sket{\psi_n^\rD}$, for all $n \le cN$, as the ``low-energy'' subspace.
Intuitively, in such a low-energy sector, the Hermite functions can be well approximated by piecewise constant 
functions $\tilde \psi_n(x)=\psi_n(x_j)$ if $x_j \le x < x_j +\sqrt{2 \pi/N}$. The integrals needed to compute
 properties of the QHO can be replaced, within high accuracy, by the sums that appear in the discrete case.
In contrast, for large values of $n$, the Hermite functions may present oscillations that will not be captured under such 
an approximation, and the approximation error gets large.

The states $\sket{\psi_n^\rD}$ form almost an orthonormal basis
of the low-energy subspace. In Appx.~A,
Cor.~\ref{cor:orthogonal}, we prove
\begin{align}
\label{eq:approxnorm}
|\sbra{\psi_n^\rD} \psi_{m}^\rD \rangle - \delta_{n,m}| = \exp(-\Omega(N)) \; ,
\end{align}
 for all $n,m \le cN$. 
Equations~\eqref{eq:approxeigen} and~\eqref{eq:approxnorm} imply that the  eigenvalues of $H^\rD$
in the low-energy sector satisfy
\begin{align}
\label{eq:approxev}
|E_n^\rD -(n+1/2) | = \exp(-\Omega(N)) \;.
\end{align}
This follows from noticing that $(H^\rD-(n+1/2))^2$ is nonnegative and its smallest eigenvalue is bounded from below by $0$
and bounded from above by $\sbra{\psi_n^\rD}(H^\rD-(n+1/2))^2 \sket{\psi_n^\rD}/\| \sket{\psi_n^\rD}\|^2=\exp(-\Omega(N))$.
The corresponding eigenstates of $H^\rD$ satisfy
\begin{align}
\label{eq:approxeigenstate}
\|  \sket{\phi_n^\rD} -  \sket{\psi_n^\rD}  \|  = \exp(-\Omega(N)) \; ,
\end{align}
when $n \le cN$.   Since the eigenvalues in the low energy subspace are gapped,
Eq.~\eqref{eq:approxeigenstate} can be shown by considering the projector $\sket{\phi_n^\rD}\sbra{\phi_n^\rD}=\lim_{u \rightarrow \infty}
e^{-u (H^\rD -E_n^\rD)^2}$. In particular, for precision $\exp(-\Omega(N))$, it suffices to choose $u=O(N)$. In this case,
$ e^{-u (H^\rD -E_n^\rD)^2} \sket{\psi_n^\rD} = ( \one - \int_0^u ds \; e^{-s (H^\rD -E_n^\rD)^2}(H^\rD -E_n^\rD)^2) \sket{\psi_n^\rD}$
and then $| \|e^{-u (H^\rD -E_n^\rD)^2} \sket{\psi_n^\rD} \| -1 |= \exp(-\Omega(N))$ or $| \| \sket{\phi_n^\rD}\sbra{\phi_n^\rD}  {\psi_n^\rD} \rangle \| -1 |= \exp(-\Omega(N))$.

The values of the constants hidden by the order notation, as well as $c$, may be estimated from
bounds on the tails of the Hermite functions and
the corresponding proofs in Appx.~A. The constant in the exponentials of Eqs.~\eqref{eq:approxnorm},~\eqref{eq:approxev},
and~\eqref{eq:approxeigenstate}
may depend on $n$ (i.e., it decreases as $n$ increases) but, for our results,
it suffices to claim that there is a constant such that the approximations of the eigenvalues
and eigenvectors are $\exp(-\Omega(N))$. Additionally, Lemma~\ref{lem:smallsupport} requires $c< \pi/16$
and the corresponding constant in the exponential for this lemma is $\pi/2$ when $n=0$.
Nevertheless, we can also perform a simple numerical analysis
to validate our results and give better estimates of    such constants and $c$.  As an example,
in Fig.~\ref{fig:schmidtmatrix} we show the absolute value of the overlap between $\sket{\psi_n^\rD}$ and 
the actual eigenstates of $H^\rD$; that is, $|\langle \phi_m^\rD \sket {\psi_n^\rD}|$, for $n,m \in \{0,1,\ldots,N-1\}$.
These numerical computations suggest that the states $\sket{\psi_n^\rD}$ are excellent approximations of the eigenstates
of $H^\rD$ for a fraction $c \ge 3/4$ of the whole spectrum.

\begin{figure}[htbp]
\includegraphics[width=8cm]{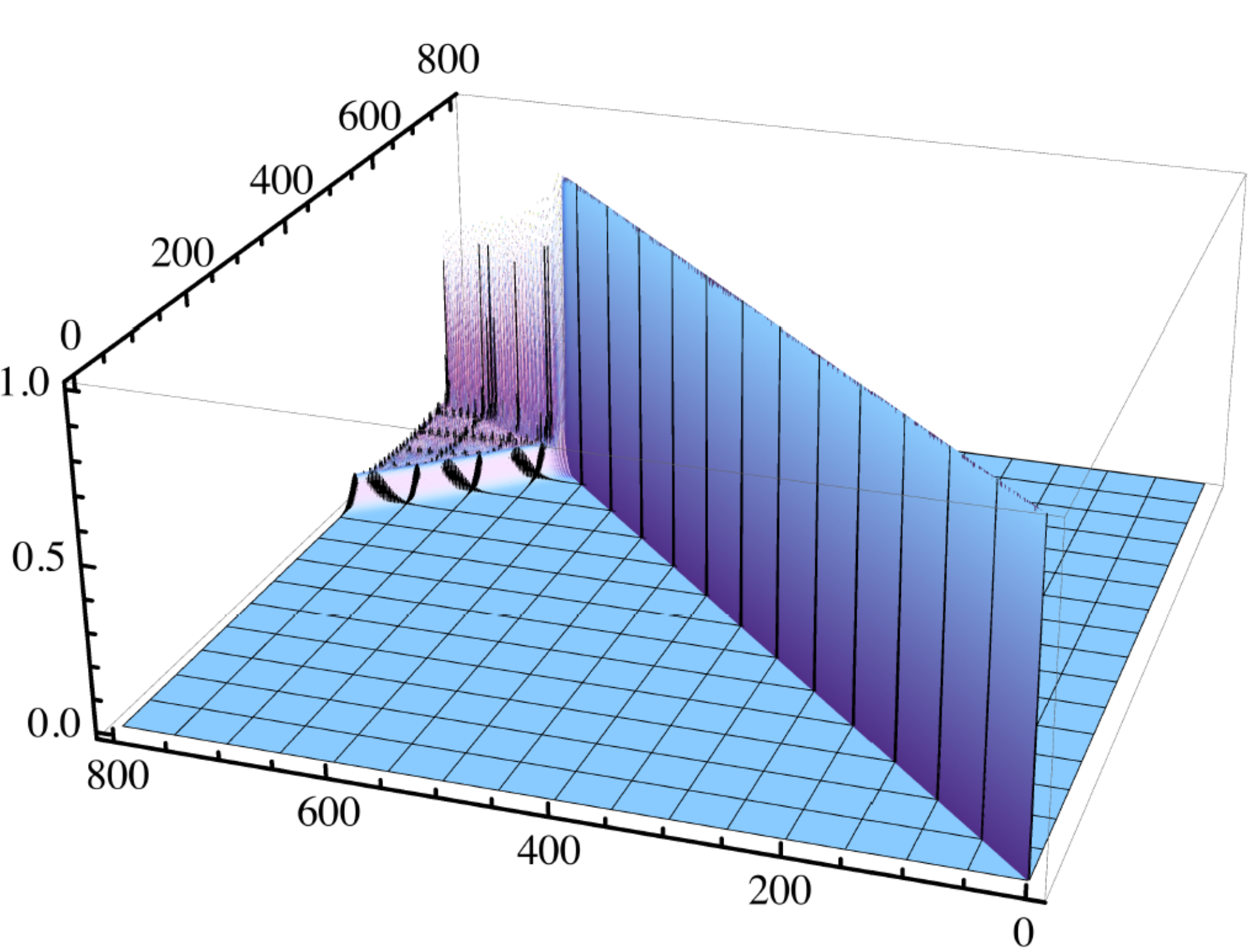}
\caption{\label{fig:schmidtmatrix} Overlap between the eigenstates of $H^\rD$, $\sket{\phi_m^\rD}$, and the 
discrete Hermite states, $\sket{\psi_n^\rD}$, for dimension $N=800$.}
\end{figure}

In Fig.~\ref{fig:erg_spectrum}  we show the eigenvalues of the CV QHO  and   $H^\rD$.  Our numerical results suggest that $|E_n^\rD - (n+1/2)|$ approaches zero as long as $n \le c N$,
also for some $c \ge 3/4$.
   \begin{figure}[hbtp]
\begin{center}
 \includegraphics[width=8cm]{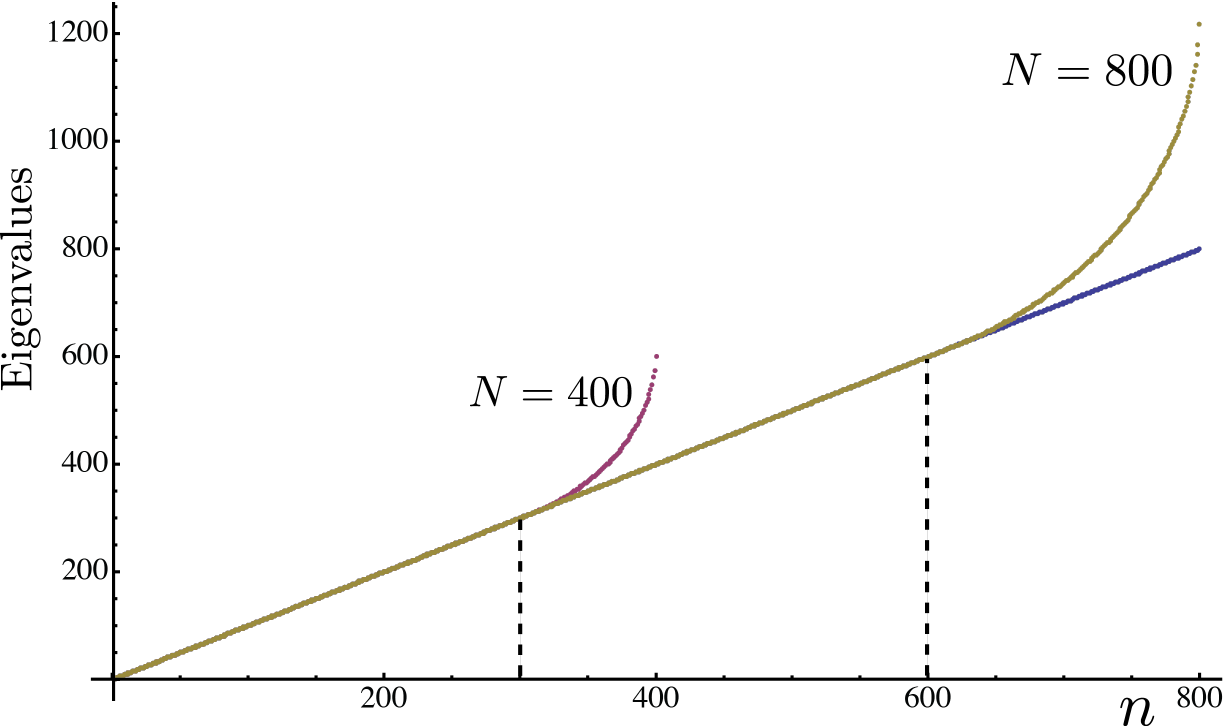}
\end{center}
\caption{\label{fig:erg_spectrum} Eigenvalues of $H^\rD$ for $N=400$ (purple) and $N=800$ (yellow), and comparison
with $n+1/2$, i.e. the eigenvalues of the CV QHO (blue). }
\end{figure}

In Fig.~\ref{fig:erg_error} we show the absolute difference between the $n$-th eigenvalue of $H^\rD$ and $n+1/2$, for $n=N/2$ and
$n=3N/4$.
The numerical results indicate that this difference decays exponentially with $N$, as determined in Eq.~\eqref{eq:approxev}, and 
where the constant in the exponential of Eq.~\eqref{eq:approxev} depends on $n$.
  \begin{figure}[hbtp]
\begin{center}
 \includegraphics[width=8cm]{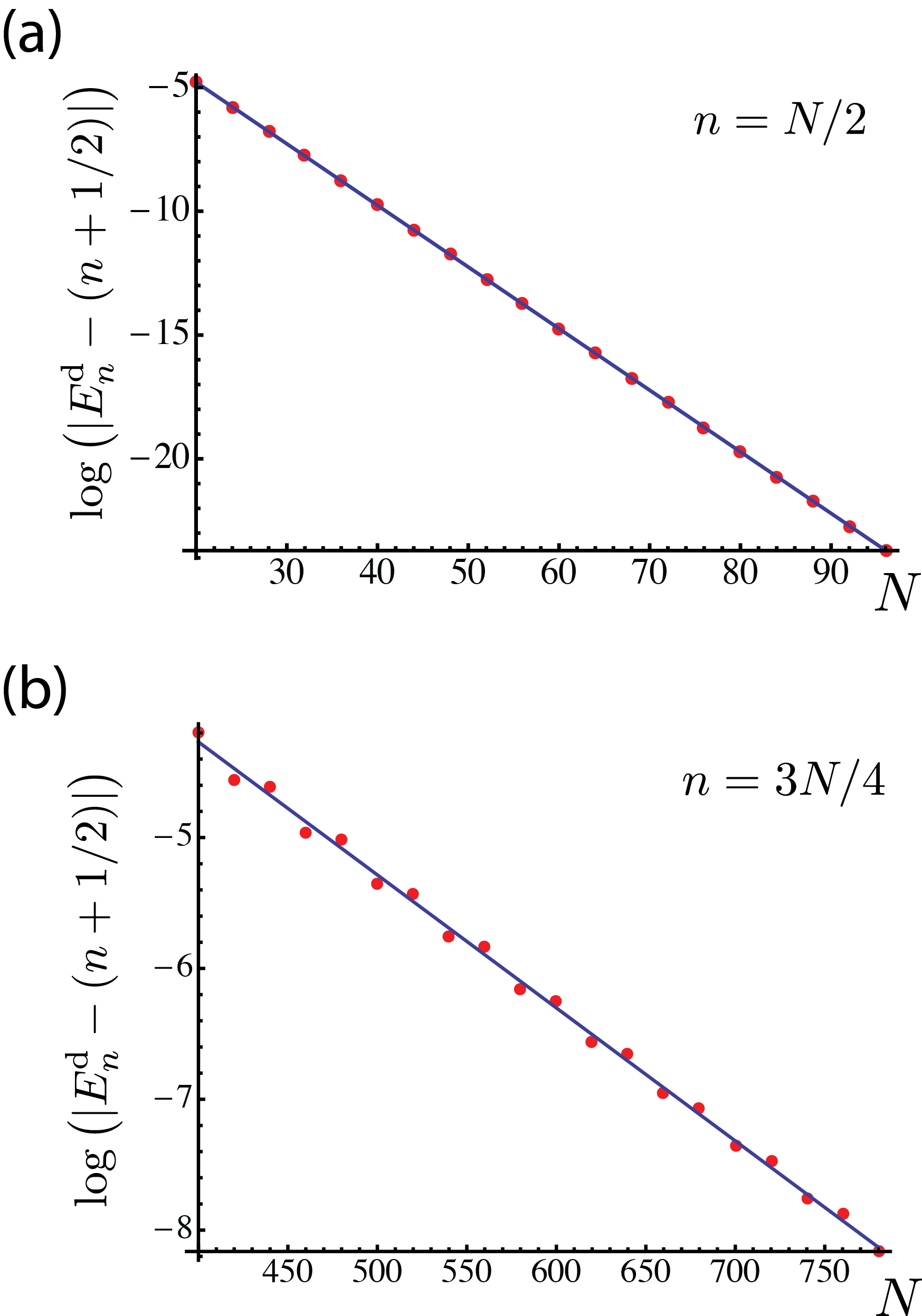}
\end{center}
\caption{\label{fig:erg_error} Plot of the natural logarithm of the absolute difference between the $n$-th eigenvalue of $H^\rD$ and $n+1/2$, 
as a function of the dimension (red dots).
The linear fit (blue line) indicates that such a difference decays exponentially with $N$. (a) $n=N/2$.
The approximate constant in the
exponential of Eq.~\eqref{eq:approxev} is $0.248$ in this case. (b) $n=3N/4$. The approximate constant in the
exponential of Eq.~\eqref{eq:approxev} is $0.010$ in this case.}
\end{figure}

\section{Scattering amplitudes}
\label{sec:DQHOexpval}

We address the first goal of a QS.
In particular, we seek  a quantum algorithm,
built upon a sequence of two-qubit gates, to compute scattering amplitudes
as determined by the evolution of the CV QHO.
 We let $U^\rD(t):=\exp\{-i H^\rD t\}$ be the unitary evolution operator of
 the   discrete QHO. Intuitively, $U^\rD(t)$ approximates $U(t)$, in some sense, as the dimension $N$
 grows larger.
 Our goal is to compute $\sbra{\varphi'} U(t) \sket{\varphi}$, where we assume
 $\sket{\varphi} = \sum_{n=0}^{N'} c_n \sket{\psi_n}$ and $N' < \infty$. We also assume that $\sket{\varphi'} = \sum_{n=0}^{N'} c'_n \sket{\psi_n}$
 or  $\sket{\varphi'} = (2 \pi/N)^{1/4} \sket{x_j}$, where  $\sket{x_j}$ is the eigenstate of the position operator $\hat x$ with
 eigenvalue $j \sqrt{2 \pi/N}$. (Note that the eigenstates of $\hat x$ cannot be written as a finite linear combination
 of $\sket{\psi_n}$.)
For the discrete case, we   define $\sket{\varphi^\rD} \propto \sum_{n=0}^{N'} c_n \sket{\psi^\rD_n}$ and $\sket{\varphi'^\rD} \propto \sum_{n=0}^{N'} c'_n \sket{\psi^\rD_n}$ or $\sket{\varphi'^\rD} = \sket j$, depending on the case.
With these definitions, $\bra {x_j} \varphi \rangle \propto \bra j \varphi^\rD \rangle$ and $\bra {x_j} \varphi' \rangle \propto \bra j \varphi'^\rD \rangle$,
so that the states $\sket{\varphi^\rD}$ and $\sket{\varphi^\rD}$ correspond to discretizations of $\sket{\varphi}$ and $\sket{\varphi'}$ in space, respectively.
 The proportionality constants are needed for normalization. The first result of this section is:
 \begin{theorem}
 \label{thm:MR1}
 Let $t$ be the evolution time and  $\epsilon >0$ be a precision parameter. Assume $|t|>1$ and $t=O(1)$ with no loss of generality.
 Then, there exists $N = O(\log (1/\epsilon) +N' )$ such that
 \begin{align}
  | \sbra{\varphi'^\rD} U^\rD(t) \sket{\varphi^\rD}  - \sbra{\varphi'} U(t) \sket{\varphi} | =O(\epsilon) \; .
 \end{align}
 \end{theorem}
 \begin{proof}
 The result follows from the subadditivity property of errors. First, we assume $\sket{\varphi'^\rD} \propto \sum_{n=0}^{N'} c'_n \sket{\psi^\rD_n}$ and let $\alpha$ and $\alpha'$ be the constants of proportionality,
 so that
 \begin{align}
 \sbra{\varphi'^\rD} U^\rD(t) \sket{\varphi^\rD} = \alpha' \alpha \sum_{n,n'=0}^{N'} (c'_{n'})^* \; c_n  \sbra{\psi_{n'}^\rD} U^\rD(t) \sket{\psi_n^\rD}  \; .
 \end{align}
 Because $\sum_{n=0}^{N'} |c_n|^2=\sum_{n=0}^{N'} |c_n'|^2=1$, Eq.~\eqref{eq:approxnorm}
 implies $|\alpha - 1 |= \exp(-\Omega(N))$ and $|\alpha' - 1 |=  \exp(-\Omega(N))$. This result assumes a   
 choice for the dimension of the discrete system of $N = O(N')$, so that $0 \le n \le N' \le cN$. The value of $n$ within the range represents the low-energy subspace, and we also
 used that $\sum_{n=0}^{N'} |c_n| = O(\sqrt N)$, $\sum_{n=0}^{N'} |c'_n| = O(\sqrt N)$.
 Furthermore,  Eqs.~\eqref{eq:approxnorm} and~\eqref{eq:approxev} 
 imply
 \begin{align}
 \label{eq:additional2}
\|  ( e^{-i H^\rD t} -e^{-i (n+1/2)t} )   \sket{\psi_n^\rD} \|= |t|\exp(- \Omega(N)) \; .
 \end{align}
 Note that we can assume $t=O(1)$ since $U(t)$ is periodic.
Then,
\begin{align}
\nonumber
 & | \sbra{\varphi'^\rD} U^\rD(t) \sket{\varphi^\rD}  -  \sum_{n,n'=0}^{N'} (c'_{n'})^* \; c_n e^{-i(n+1/2)t} \sbra{\psi_{n'}^\rD}   \psi_n^\rD\rangle | = \\
 \label{eq:additional1}
  &= \exp(- \Omega(N)) \; .
\end{align}
 From Eq.~\eqref{eq:approxnorm}, we know that there exists a constant $c$ such that, if $n,n' \le N'=cN <N$, then
 \begin{align}
 \nonumber
& \left | \sum_{n,n'=0}^{N'} (c'_{n'})^* \; c_n e^{-i(n+1/2)t} \sbra{\psi_{n'}^\rD}   \psi_n^\rD\rangle -  \right. \\
 &\left. - \sum_{n=0}^{N'} (c'_{n})^* \; c_n e^{-i(n+1/2)t} \right| =\exp(- \Omega(N)) \;.
 \end{align}
 Using Eq.~\eqref{eq:additional1}, it follows that
  \begin{align}
\nonumber
&| \sbra{\varphi'^\rD} U^\rD(t) \sket{\varphi^\rD} -\sbra{\varphi'} U (t) \sket{\varphi} | = \\
\nonumber
& =| \sbra{\varphi'^\rD} U^\rD(t) \sket{\varphi^\rD}  -  \sum_{n=0}^{N'} (c'_{n})^* \; c_n e^{-i(n+1/2)t}  |  
 \\
& =\exp(-\Omega(N)) \; .
 \end{align} 
 
We now consider the case where $\sket{\varphi'^\rD}=\ket j$. We seek to show that $ \sbra{j} U^\rD(t) \sket{\varphi^\rD}$ approximates
$(2 \pi/N)^{1/4} \sbra{x_j}U(t) \sket{\varphi}$. Equation~\eqref{eq:additional2} implies
\begin{align}
| \sbra{j} U^\rD(t) \sket{\varphi^\rD} - \sum_{n=0}^{N'} c_n e^{-i(n+1/2)t} \sbra j \psi_n^\rD \rangle | =  \exp(-\Omega(N)) \;.
\end{align}
Additionally, 
\begin{align}
\nonumber
& \sum_{n=0}^{N'} c_n e^{-i(n+1/2)t} \sbra j \psi_n^\rD \rangle= \\
\nonumber
&= (2 \pi/N)^{1/4} \sum_{n=0}^{N'} c_n e^{-i(n+1/2)t} \sbra{x_j} \psi_n\rangle \\
& = (2 \pi/N)^{1/4} \bra{x_j}U(t) \sket{\varphi} \;.
\end{align}
Then, 
  \begin{align}
| \sbra{\varphi'^\rD} U^\rD(t) \sket{\varphi^\rD} -\sbra{\varphi'} U (t) \sket{\varphi} | = \exp(-\Omega(N))
\end{align}
in this case as well.

If $|t|\ge 1$, it follows that there exists $N = O(\log(1/\epsilon))$ that  implies $  \exp(-\Omega(N))= O(\epsilon)$ for both cases of $\sket{\varphi'^\rD}$
or $\sket{\varphi'}$.
Then,   $N = O(\log(1/\epsilon) +N')$   suffices to provide an overall precision of order $\epsilon$
in the computation of $\sbra{\varphi'} U(t) \sket{\varphi}$.
 \end{proof} 
 \vspace{1cm}

 Theorem~\ref{thm:MR1} basically shows that scattering amplitudes
of the continuous-variable QHO can be well approximated by those of the discrete QHO as long as
the dimension $N$ scales  with the largest quantum number in the decomposition of $\sket{\varphi}$ (and  $\sket{\varphi'}$)
in terms of eigenvectors of
$H$. Also, $N$ is only logarithmic in $1/\epsilon$. In particular, if the  states $\sket{\varphi^\rD}$ and $\sket{\varphi'^\rD}$ 
can be efficiently prepared on a quantum computer, the complexity of  computing $\sbra{\varphi'^\rD} U^\rD(t) \sket{\varphi^\rD}$ 
 will be dominated by the complexity of implementing
the evolution operator $U^\rD(t)$ for the corresponding choice of $N$.

 
\subsection{Time evolutions: Trotter-Suzuki approximation}
\label{sec:TSdecomp}
To compute the propagator on a quantum computer,
we seek an implementation or simulation of $U^\rD(t)$, whose complexity
will be determined by the number of two-qubit gates required to approximate the operator.
Unfortunately, $H^\rD$ has large norm, so known results on Hamiltonian simulation~\cite{berry_efficient_2007,childs_efficient_2010,berry_sparse,BCC_Taylor_2015,cleve_query_2009,WBHS11}
are not very useful in the current case; other approaches are necessary.

Since $x^\rD$ is diagonal and each  entry can be efficiently computed,
a quantum computer simulation of  $\exp (-i (x^\rD)^2 t)$ can be 
done efficiently. To show this, let $q(\delta)=O({\rm polylog}(N/\delta))$ be the number of two-qubit gates
required to compute a diagonal entry of $(x^\rD)^2$ within precision $\delta >0$.
Then,  a diagonal entry of $(x^\rD)^2 t$ can be computed within precision
$\tilde \epsilon$ using $q(\tilde \epsilon/|t|) + O({\rm polylog} (N|t|/ \tilde \epsilon))$ two-qubit gates:
 we need $O({\rm log} (N|t|/\tilde \epsilon))$ bits (or qubits) to represent $(x^\rD)^2$
within precision $\tilde \epsilon/|t|$, and multiplication by $t$ can be done using additional
$O({\rm polylog} (N|t|/\tilde \epsilon))$ gates. Then, $\exp (-i (x^\rD)^2 |t|)$ can be simulated
on a quantum computer, within precision $\tilde \epsilon$, using $O({\rm polylog} (N|t|/\tilde \epsilon))$ two-qubit gates.
The complexity for the simulation of $\exp (-i (p^\rD)^2 t)$
is of the same order, since $x^\rD$ and $p^\rD$ are related by the centered Fourier transform and
$\exp (-i (p^\rD)^2 |t|) = (F^\rD_{\rm c}) ^{-1} . \exp (-i (x^\rD)^2 |t|) .  F^\rD_{\rm c} $. 
Additionally, $F^\rD_{\rm c}= X^{N/2}. F^\rD. (X)^{-N/2}$, where $X$ is a cyclic permutation [Eq.~\eqref{eq:cyclicP}].
(Note that $X^N= \one$ and $X^{N/2}=X^{-N/2}$.)
An efficient quantum circuit for $F^\rD$ that requires $O({\rm polylog}(N))$ two-qubit gates is known~\cite{kitaev_quantum_1995,nielsen_quantum_2000}.  $X^{N/2}$
can be implemented on a quantum computer using $O({\rm polylog}(N))$ two-qubit gates in a number of different
ways.
For example, $X^{N/2}$ can also be decomposed as $X^{N/2}=F^\rD . Z^{N/2} . (F^\rD)^{-1}$, where $Z$ is the diagonal 
unitary whose nonzero entries are the roots of unity, i.e., $\exp (i 2 \pi k/N)$ for $k \in \{0,\ldots,N-1 \}$. Then, the diagonal
entries of $Z^{N/2}$ are $\exp (i  \pi k ) = \pm 1$, resulting in a simple and efficient quantum-computer 
implementation of $Z^{N/2}$. We also note that when $N=2^a$ is associated with a system of $a$ qubits, $X^{N/2}$ and $Z^{N/2}$
are simply the Pauli operators $\sigma_x$ and $\sigma_z$, respectively, acting on the first qubit. 

The above results
suggest using the Trotter-Suzuki product formula to simulate $U^\rD(t)$~\cite{trotter_1959,Huyghebaert_1990,suzuki_90,suzuki_qmc_1998,berry_efficient_2007,wiebe_product_2010}.
Such a formula approximates
the evolution operator as a sequence or product of shorter time evolutions under $(x^\rD)^2$
and $(p^\rD)^2$. Known results provide an upper bound on the number of terms in the formula 
given by ${\cal M}=O(\exp(b/\eta) (\| H^{\rD} \| |t|)^{1+\eta}/\epsilon^\eta)$, where $\eta >0$ is arbitrarily small, $b>0$ is a constant, and 
$\epsilon>0$ is the precision~\cite{berry_efficient_2007,childs_thesis}. Since $\| H^\rD \|= O(N)$,
the results in~\cite{berry_efficient_2007,childs_thesis} would  imply that the number of two-qubit gates required to approximate $U^\rD(t)$ grows
faster than $N$, being undesirably large in the asymptotic limit.  Remarkably, 
an improved analysis of the complexity of high-order Trotter-Suzuki 
product formulas allows us to reduce the complexity substantially. The basic idea is to note that
the cost resulting from such  formulas actually depends on quantities (norms of nested commutators) such as $\| [( x^\rD)^2, ( p^\rD)^2 ] \|$ rather than
$\| H^{\rD} \|^{1+\eta}$. In particular, $\hat x^2$, $\hat p^2$, and $\{ \hat x , \hat p\}$
are a basis of a Lie algebra ${sp}(2)$ of dimension 3. Since our discrete QHO operators {\em approximate} the operators in the continuous-variable case,
it is possible to show that $( x^\rD)^2$, $( p^\rD)^2$, and $\{ x^\rD, p^\rD \}$ 
are  {\em almost} a basis of a Lie algebra of dimension 3, when projected onto the low-energy subspace.
This implies, for example,   $\| Q [( x^\rD)^2, ( p^\rD)^2 ] Q \| =O(N)$, where $Q$ is the projector
into the low energy subspace spanned by $\{ \sket{\psi_n^\rD} \}$, $n \le cN$.
In contrast, a simple analysis would have resulted in $\| Q [( x^\rD)^2, ( p^\rD)^2 ] Q \| =O(N^2)$.
 
 The (symmetric) Trotter-Suzuki approximations of the evolution operator over a course
 of evolution time $s$ are defined recursively as follows~\cite{suzuki_90,berry_efficient_2007,childs_thesis}:
 \begin{align}
 \label{eq:TSArecdef}
 U_{p+1}^\rD(s) := (U_p^\rD(s_p))^2 U_p^\rD(s-4s_p) (U_p^\rD(s_p))^2 \; ,
 \end{align}
 and 
$U_1^\rD(s) := e^{-i s (x^\rD)^2/4} e^{-i s (p^\rD)^2/2}e^{-i s (x^\rD)^2/4}$.
 Here, $p \ge 2$ is integer
 and $s_p = s/(4-4^{1/(2p+1)})$. For time $t$, we will approximate the evolution
 operator $U^\rD(t)$ by $(U^\rD_p(s))^k$, where $k = t/s$ and $p$ are chosen
 to minimize the number of exponentials of $(x^\rD)^2$ and $(p^\rD)^2$ in the product.

 In Appx.~B, Lemma~\ref{lem:TSA}, we show
 that there is a choice for the dimension of the Hilbert space where $N= \exp\{ O( \sqrt{\log(N' |t|/\epsilon)}) \}  +O(N')$
  a choice of $p=\Theta(\sqrt{\log(N' |t|/\epsilon)})$, and $|s| = \Theta (5^{-p})$, such that
 \begin{align}
 \label{eq:TSAdiscrete}
 \| [(U_p^\rD(s))^k - U^\rD(t)] \sket{\psi_n^\rD} \| =O(\epsilon) \; ,
 \end{align}
 for all $n \le N'$.  
 We will use Eq.~\eqref{eq:TSAdiscrete} to prove the main result of this section:
 \begin{theorem}
 \label{thm:MR2}
 Let $ \sket{\varphi} = \sum_{n=0}^{N'} c_n \sket{\psi_n}$ be the initial state
 of the CV QHO, $t=O(1)$ the evolution time ($|t| \ge 1$), and $\epsilon >0$. Then, there exists $N= \exp(O( \sqrt{\log(N' /\epsilon)})) + O(N')$,
  $p=\Theta(\sqrt{\log(N' /\epsilon)})$, and $|s| = \Theta(5^{-p})$, such that ($k=t/s$)
 \begin{align}
  \| [(U_p^\rD(s))^k - U^\rD(t)] \sket{\varphi^\rD} \| =O(\epsilon) \; ,
 \end{align}
 for all  $\sket{\varphi^\rD} \propto  \sum_{n=0}^{N'} c_n \sket{\psi_n^\rD}$.
 The number of exponentials of $(x^\rD)^2$ and $(p^\rD)^2$ in the product $(U_p^\rD(s))^k$ is
 ${\cal M}= O(\exp(\gamma \sqrt{\log (N' /\epsilon)}))$, where $\gamma >0$ is a constant.
   The number of
  two-qubit gates to simulate $(U_p^\rD(s))^k$ within precision $\epsilon$ is 
 $ \tilde {\cal M}= O( \exp(\tilde \gamma \sqrt{\log (N' /\epsilon)}))$, where $\tilde \gamma >0$ is a constant.
  \end{theorem}
 \begin{proof}
 The first result is a direct consequence of Eq.~\eqref{eq:TSAdiscrete}, for the same choices of $N$, $p$, and $s$
 as in Lemma~\ref{lem:TSA} -- see Eq.~\eqref{eq:spoptimal2} in Appx.~B. The only additional error is that coming from the fact that the states $\sket{\psi_n^\rD}$
 are not exactly an orthonormal basis for $n \le N'$. Such an error is order $\nu_1(N)$, and the choice of  
  $N$ implies
 that this error is negligible if $|t| \ge 1$. 
 
 The definition of $U^\rD_p$  implies that the total number of exponentials
 of $(x^\rD)^2$ and $(p^\rD)^2$ in $U^\rD_p(s)$ is bounded from above by $5^p$ [Eq.~\eqref{eq:TSArecdef}].
  Also, the choice of $s$ in Lemma~\ref{lem:TSA} satisfies $5^p |s| =\Theta(1)$. Then, the total number of exponentials in $(U^\rD_p(s))^k$      
 is ${\cal M}= O(5^p k) = O(|t| 5^{2p})$, where $k=t/s$. That is, ${\cal M}= O(|t| \exp(\gamma \sqrt{\log (N' |t|/\epsilon)}))$, for some
 constant $\gamma >0$.
 
 To achieve overall precision $\epsilon$, each exponential of $(x^\rD)^2$ or $(p^\rD)^2$ in $(U^\rD_p(s))^k$
 has to be simulated within precision $O(\epsilon/{\cal M})$. Then, the above discussion on the cost of implementing
 diagonal unitaries implies that the number 
 of two-qubit gates for each exponential is $O({\rm polylog}(N|s|{\cal M}/\epsilon))$
 or $O({\rm polylog}(N|t|5^p/\epsilon))$. [Note that $|s_p| = O(s)$ and $(s-4s_p)=O(s)$.]
 The choice of $N$ and $p$ imply that $\log (N|t|5^p/\epsilon) = O(\log (|t|/\epsilon)) + O( \sqrt{\log (N' |t|/\epsilon)}) + O(\log N')$,
 and we can safely bound this quantity by $O(\log(N' |t|/\epsilon))$. Then, the total number
 of two-qubit gates is $O({\cal M} \; {\rm polylog} (N' |t|/\epsilon) )$. That is, there exists a constant $\gamma >0$
 such that the   number of two-qubit gates to simulate $(U^\rD_p(s))^k$ within precision $\epsilon$
 is $\tilde {\cal M} = O(|t| \exp(\gamma \sqrt{\log (N' |t|/\epsilon)}) {\rm polylog}(N' |t|/\epsilon))$. 
 We let $x \ge 1$ and $k \ge 0$ a constant. Then, there exists a constant $\tilde \gamma$ such
 that $e^{\gamma \sqrt x} x^k \le e^{\tilde \gamma \sqrt x}$. It follows that $\tilde {\cal M}= O(|t| \exp(\tilde \gamma \sqrt{\log (N' |t|/\epsilon)}))$
 and the result is obtained setting $t=O(1)$.
 \end{proof}
 
 We illustrate the results of this section with several numerical simulations.
 In Fig.~\ref{fig:TS1} we show the error dependence of the Trotter-Suzuki approximation as a function of $n$ for fixed $s$ and $p$.
 A worst-case analysis indicates that this error would be  $O(n^{(2p+1)/2})$~\cite{berry_efficient_2007,berry_sparse}. 
  However, a linear scaling in $n$ is observed, for $n \le N/2$. This is a main reason for the quantum speedup; see Appx~B.
 In Fig.~\ref{fig:TS2} we choose $p$ and $s$ according to Thm.~\ref{thm:MR2} and show that the approximation error is indeed
 bounded by $\epsilon$.
 
 \begin{figure}[htbp]
\includegraphics[width=8cm]{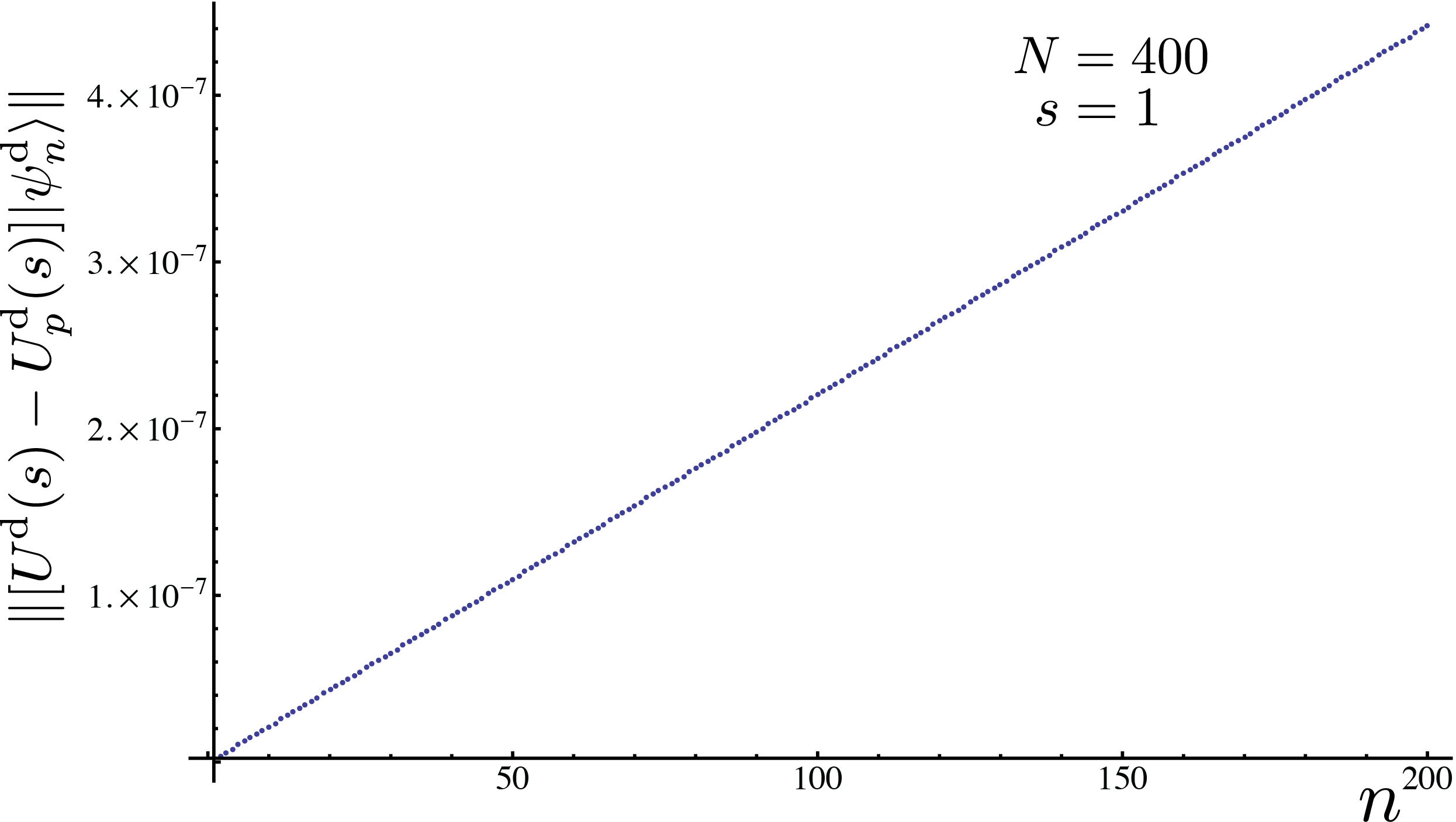}
\caption{ \label{fig:TS1} The difference between $U^\rD(s) \sket{\psi_n^\rD}$ and $U^\rD_p(s) \sket{\psi_n^\rD}$ as a function of $n$, for $N=400$, $p=4$, and $s=1$.
 The scaling is almost linear in $n$, as suggested by Lemma~\ref{lemma:TSA} in Appx.~B.
 Additional simulations also observe a linear scaling in $n$ for higher values of $p$.} 
\end{figure}

\begin{figure}[htbp]
\includegraphics[width=8cm]{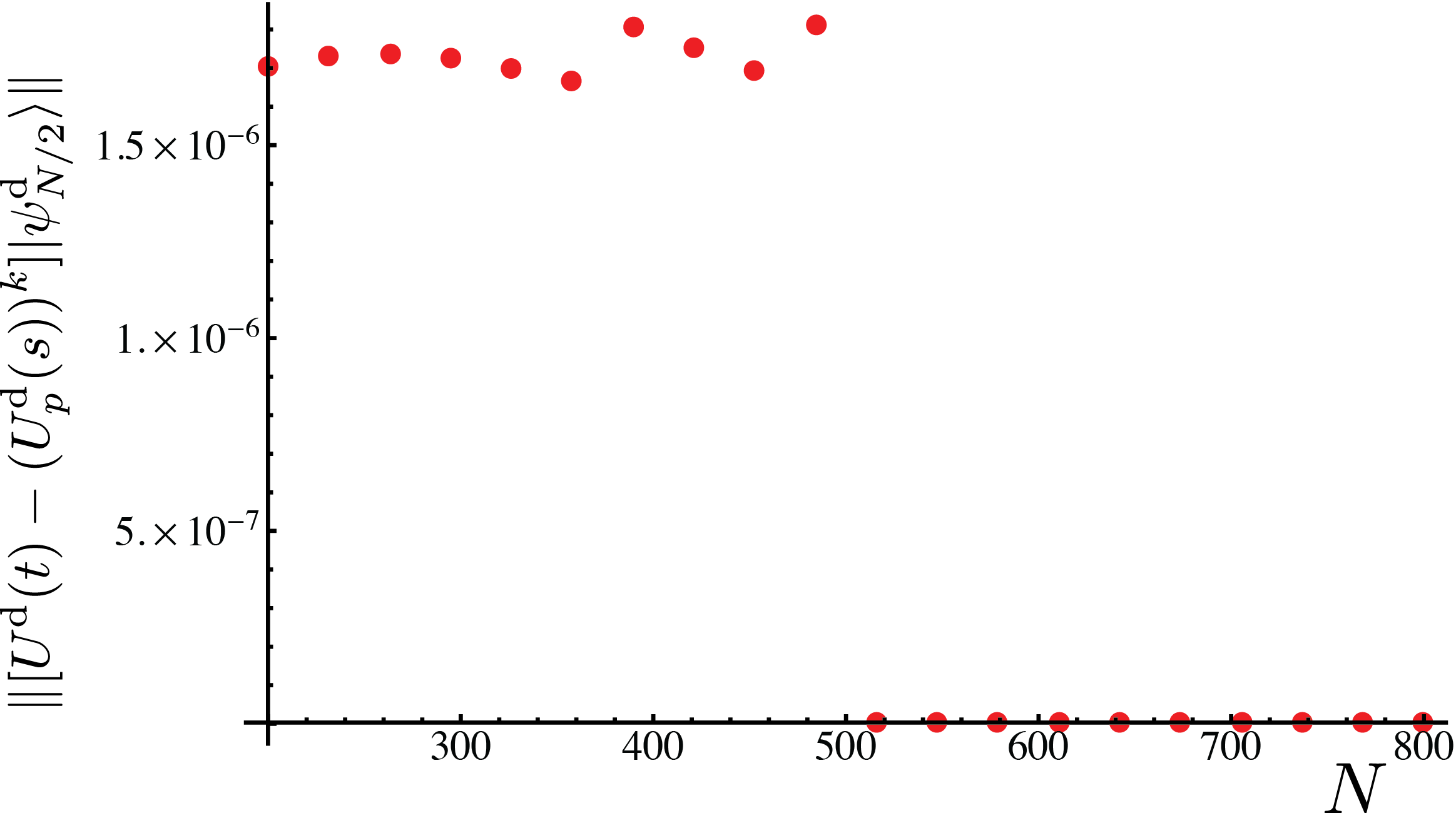}
\caption{ \label{fig:TS2}
The difference between $U^\rD(t) \sket{\psi_n^\rD}$ and $(U^\rD_p(s))^k \sket{\psi_n^\rD}$ as a function of the dimension,
for $t=\pi/2$ and $n=N/2$. The values of $p$ and $s$ where set as determined by Thm.~\ref{thm:MR2}. That is, according to Lemma~\ref{lemma:TSA2}
in Appx.~B, we chose
$p=\lceil \sqrt{\log((n+2)t/\epsilon)/(2 \log(5))} \rceil$, $k = \lceil t /(\epsilon/((n+2)t)^{1/(2p)}) \rceil$, and $s=t/k$, for $\epsilon=0.001$.
Note that with these choices, $s = \exp(-O(\sqrt{\log (N t/\epsilon)}))$ and the number of terms in the product $(U^\rD_p(s))^k$
is ${\cal M}= k5^p = O(t \exp(\gamma \sqrt{\log (N t/\epsilon)}))$, for some $\gamma>0$. The error induced 
by the Trotter-Suzuki approximation in this case is smaller than $\epsilon$. The ``jump'' at $N \approx 500$ corresponds to an increase of $p$ from 2 to 3,
due to the increase in $N$.
}
\end{figure}

\subsection{Subexponential-time quantum algorithm for computing scattering amplitudes}
\label{sec:results}
We can combine Thms.~\ref{thm:MR1} and~\ref{thm:MR2} to construct a quantum algorithm
to simulate the QHO and obtain the desired propagators $\sbra{\varphi'} U(t) \sket{\varphi}$. The quantum algorithm
has three basic steps: i- the preparation of an initial state, ii- the implementation or simulation of the evolution operator, 
and iii- a projective measurement. The desired expectation value can be obtained within arbitrary accuracy after repeated
executions of these steps. The simulation of the evolution operator induced by $H^\rD$ was discussed in Sec.~\ref{sec:TSdecomp}. 
We let $W^\rD_\varphi$ and $W^\rD_{\varphi'}$ be
the $N \times N$ unitary matrices that prepare the states $\sket{\varphi^\rD}$ and $\sket{\varphi'^\rD}$, from e.g. $\sket 0$,
respectively. We also let $V^\rD(t) = (W^\rD_{\varphi'})^\dagger U^\rD(t) W^\rD_\varphi$, so the propagator of interest
is $\sbra 0 V^\rD(t) \sket 0$.
To measure such expectation values at precision $\epsilon$, we can implement the quantum circuit in Fig.~\ref{fig:algorithm}, $O(1/\sqrt \epsilon)$
times~\cite{somma_physics_2002,somma_physics_2003}. Otherwise, we can implement the quantum methods described in~\cite{knill_expectation_2007} for optimal quantum measurements of overlaps,
where the number of repetitions can be improved to $O(1/\epsilon)$. 
\begin{figure}[htbp]
\includegraphics[width=8cm]{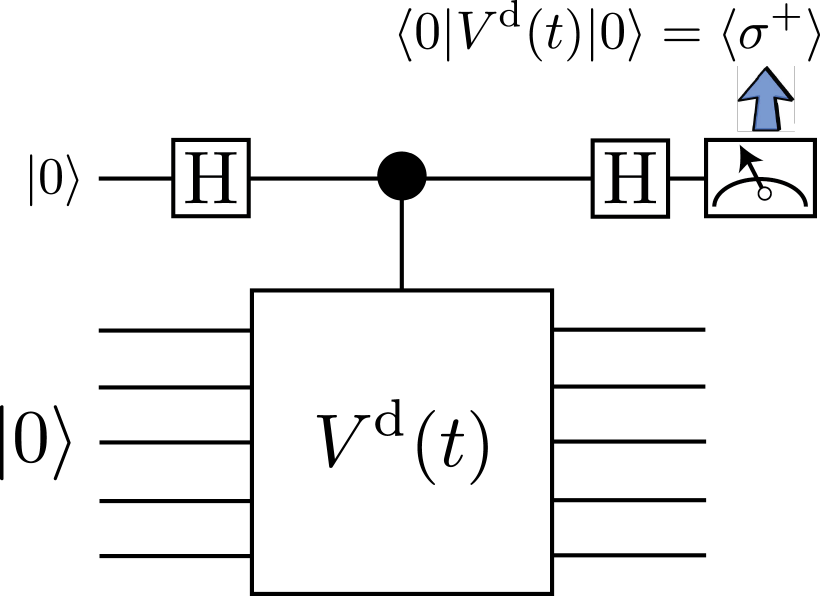}
\caption{\label{fig:algorithm}
Quantum circuit to compute the propagator $\sbra 0 V^\rD(t) \sket 0$, where $V^\rD(t)$ is unitary~\cite{somma_physics_2002,somma_physics_2003}.
The filled circle denotes the controlled operation on the state $\sket 1$ of the ancilla qubit; that is,  
the corresponding unitary operation is $\one \otimes \sket 0 \sbra 0 + V^\rD(t) \otimes \sket 1 \sbra 1$. H denotes
the Hadamard gate that maps the state of a qubit as ${\rm H} \ket 0 = (\ket 0 + \ket 1)/\sqrt 2$ and
${\rm H} \ket 1 = (\ket 0 - \ket 1)/\sqrt 2$. The single qubit operator $\sigma^+$ is $\sigma_x + i \sigma_y$, with
$\sigma_\alpha$ the Pauli operators. Since $\sigma^+$ is not Hermitian, the computation of its expectation
value can be done by repeated projective measurements of $\sigma_x$ and $\sigma_y$ independently (i.e., with
repeated executions of the  circuit).} 
\end{figure}

In the rest of this section we assume that the precision $\epsilon \ll 1$ is constant, i.e., $\epsilon =O(1)$.
We also assume that $W^\rD_\varphi$ and $W^\rD_{\varphi'}$ can be efficiently implemented
using $O({\rm polylog}(N))$ gates. Then, Thms.~\ref{thm:MR1} and~\ref{thm:MR2} imply
 \begin{theorem}
 \label{thm:MR3}
 Let $ \sket{\varphi} = \sum_{n=0}^{N'} c_n \sket{\psi_n}$, $ \sket{\varphi'} = \sum_{n=0}^{N'} c'_n \sket{\psi_n}$ or
  $ \sket{\varphi'} = (2 \pi/N)^{1/4}\sket{x_j}$ and $t$ the evolution time. Assume $|t| \ge 1$ and, with no loss of generality, $t = O(1)$.
  Then, there exists a quantum algorithm $\cal Q$ that outputs $\sbra{\varphi'} U(t)  \sket{\varphi}$
  at arbitrary accuracy.   $\cal Q$ requires ${\cal M}_{\cal Q}=O(\exp (\gamma ' \sqrt{\log (N' )}))$ two-qubit gates, where $\gamma '>0$ is a constant. 
   \end{theorem}
  \begin{proof}
  Let $N=\exp(O( \sqrt{\log(N' |t|)})) + O(N')$, $p=O(\sqrt{\log(N' |t|)})$, and $|s| = \exp(-O( \sqrt{\log(N' |t|)}))$ as in
  Lemma~\ref{lem:TSA} or the proof of Thm.~\ref{thm:MR2}, for $\epsilon = O(1)$. The
  subadditivity property of errors implies ($k=t/s$)
  \begin{align}
  \nonumber
  \sbra{\varphi'} U(t)  \sket{\varphi} & = \sbra{\varphi'^\rD} (U^\rD_p(s))^k  \sket{\varphi^\rD} + O(\epsilon) \\
& =  \sbra 0 (W_{\varphi'}^\rD)^\dagger (U^\rD_p(s))^k {W_\varphi^\rD} \sket{0} + O(\epsilon) \; ,
  \end{align}
  with $\epsilon \ll 1$.
  Then, we can use the circuit of Fig.~\ref{fig:algorithm}, a constant number of times, with a unitary $V^\rD(t) = (W_{\varphi'}^\rD)^\dagger (U^\rD_p(s))^k {W_\varphi^\rD}$
  to output the desired propagator. Since $W_{\varphi'}^\rD$ and $W_{\varphi}^\rD$ can be efficiently implemented,
  the complexity of the quantum algorithm is dominated by the number of two-qubit gates needed to implement $(U^\rD_p(s))^k$
  within the desired accuracy. Theorem~\ref{thm:MR2} implies that such a number is $O(|t| \exp (\gamma ' \sqrt{\log (N' |t|)}))$, for
  some constant $\gamma ' >0$. The result follows by assuming, with no loss of generality, $t=O(1)$.
  \end{proof}
  
 It is important to remark that the complexity of the quantum algorithm
  satisfies 
  \begin{align}
  \lim_{N' \rightarrow \infty} {\cal M}_{\cal Q}/N'^\eta=0 \; ,
  \end{align}
    for all $\eta >0$, and $\lim_{N' \rightarrow \infty} \log( N')/{\cal M}_{\cal Q}=0$.
  That is, the complexity of the quantum algorithm is subexponential in $\log (N')$. Since classical algorithms are expected to require
  ${\rm poly}(N')$ operations to compute the propagator in the worst case, our quantum algorithm provides a superpolynomial quantum speedup. 

\section{Eigenstate preparation}
\label{sec:ESprep}
 We now investigate ways of simulating and preparing low-energy
  eigenstates of $H^\rD$,
 via the action of unitary operations acting
 on simple initial states. In part, this section    addresses
 the second goal of a QS, namely the computation
 of expectation values on various eigenstates of the QHO,
 which can be obtained using the techniques presented in previous sections
 if we replace the initial state $\sket{\varphi}$ by the corresponding $\sket{\psi_n}$
 (or $\sket{\varphi^\rD}$ by $\sket{\phi_n^\rD}$).
 The results of this section may be of independent interest; e.g.,
  quantum algorithms to prepare states with Gaussian-like
 amplitudes are important in other cases~\cite{KW2008}.
 
%
 
 We first focus on the preparation of the ground state $\sket{\phi_0^\rD}$.
 In Appx.~C, Lemma~\ref{lemma:GSpreparation}, we prove in the large-$N$ limit, 
 \begin{align}
  \label{eq:gaussianpreparationerror}
 \| \sket{\psi_0^\rD} - e^{-i \alpha(t)} e^{i(x^\rD)^2 t'} e^{i(p^\rD)^2 t} \sket{\varphi^\rD} \| = O(\exp(-\Omega(\delta))) \; ,
 \end{align}
 where the initial state is 
 \begin{align}
(1/\sqrt \kappa) \sum_{j=-N/2}^{N/2} \exp(-j^2/(2 \delta)) \ket j \; ,
 \end{align}
 and $\delta >0$. The constant $\kappa$ is for normalization purporses and $\alpha(t)$
 is an irrelevant global phase that can be computed exactly. The evolution times
 satisfy $t= \sqrt{\sigma^2(2-4 \sigma^2)}/2$ and $t'= 1/(4t + 4 \sigma^2/t)$, and $\sigma^2 = \pi \delta/N$.
 This result was obtained by realizing that in CVs, the quantum state $\sket{\psi_0}$ can be obtained
 from an initial state with Gaussian-like amplitudes by evolving with the free-particle Hamiltonian (i.e., $-\hat p^2$).
 The result follows by approximating the CV case after a proper discretization.
 
 Lemma~\ref{lemma:GSpreparation} allows us to state the first result of this section.
 \begin{theorem}
 \label{thm:groundstatepreparation}
 Let $\epsilon>0$. Then, there is a unitary $W^\rD$  that satisfies
 \begin{align}
 \| \sket{\phi_0^\rD} - W^\rD \ket 0 \| = O(\epsilon)
 \end{align}
 in the large $N$ limit. $W^\rD$ can be implemented on a quantum 
 computer using a number of two-qubit gates that is polynomial in $\log(N/\epsilon)$.
 \end{theorem}
 
\begin{proof}
 First, we choose $\delta = O(\log(1/\epsilon))$ so that the right hand side of Eq.~\eqref{eq:gaussianpreparationerror}
 is $O(\epsilon)$.
 In the large $N$ limit, $\sket{\psi_0^\rD}$ can be safely replaced by $\sket{\phi_0^\rD}$
 in Eq.~\eqref{eq:gaussianpreparationerror}, as the error of this replacement is exponentially small in $N$ and thus negligible.
 Next, we note that we can approximate $\sket{\varphi^\rD}$ within error $O(\epsilon)$ by
 \begin{align}
 \label{eq:initialgaussianapprox}
\propto  \sum_{j=-j_0}^{j_0} e^{-j^2/(2\delta)} \ket j \;,
\end{align}
with $j_0 = O(\sqrt{\log(1/\epsilon)\delta })=O(\log(1/\epsilon))$.
The state of Eq.~\eqref{eq:initialgaussianapprox} can be prepared with complexity polynomial in 
$\log(1/\epsilon)$ using standard techniques. We write $V^\rD$ for the unitary that prepares such a state
and define $W^\rD = e^{-i \alpha(t)} e^{i(x^\rD)^2 t'} e^{i(p^\rD)^2 t} V^\rD$, with the choices of $t$
and $t'$ given above.
Because $t = O(1/\sqrt N)$ and $t' = O(\sqrt N)$ in the large $N$ limit,
the unitaries $e^{i(x^\rD)^2 t'}$ and  $e^{i(p^\rD)^2 t}$ can be implemented on a quantum computer
with complexity polynomial in $\log(1/\epsilon)$ and $\log( N)$ [i.e., polynomial in $\log(N/\epsilon)$] -- see Sec.~\ref{sec:TSdecomp}.
\end{proof}
 
In Fig.~\ref{fig:GroundStateError} we show the exponential decay of the error as a function of $\delta$, as stated by the theorem.
We note that quantum methods to prepare states with Gaussian-like amplitudes were also proposed in~\cite{KW2008,GR2002}.
 
 \begin{figure}[htbp]
\includegraphics[width=8cm]{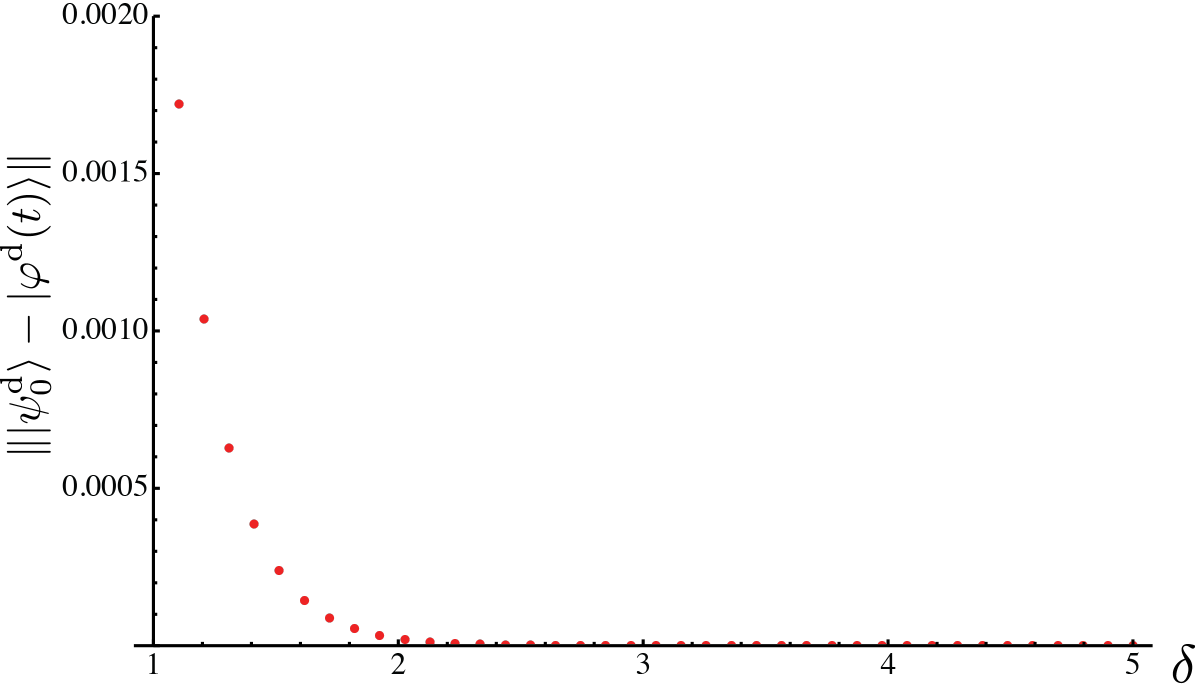}
\caption{\label{fig:GroundStateError} The norm of the difference between the state $\sket{\psi_0^\rD}$ and the evolved state $\sket{\varphi^\rD(t)}=
e^{-i \alpha(t)} e^{i (x^\rD)^2 t'} e^{-i (p^\rD)^2 t}\sket{\varphi^\rD}$, as a function of $\delta$ and for $N=800$.
Numerical simulations do not show significant changes for larger dimensions.}
\end{figure}

  To prepare the other eigenstates $\sket{\phi_n^\rD}$, with $n \ge 1$, we define
  a discrete version of the Jaynes-Cummings (JC) model:
  \begin{align}
  \label{eq:DJCH}
  H^\rD_{\rm JC} = (x^\rD \otimes \sigma_x - p^\rD \otimes \sigma_y)/\sqrt 2 \; ,
  \end{align}
 where $\sigma_\alpha$ are the corresponding Pauli operators acting on
 a the Hilbert space of an ancillary qubit. In CV, the evolution induced by the JC
 model eventually transforms the state $\sket{\psi_n}   \ket 0$ into $\sket{\psi_{n+1}}   \ket 1$,
 providing a unitary operation to prepare eigenstates of the QHO from $\sket{\psi_0}$.
 We will show that something similar occurs in the discrete case.
 
 For $n \ge 0$, we define the normalized states
 \begin{align}
 \sket{\gamma^\rD_{n,\pm}} = \frac 1 {\sqrt 2} [\sket{\phi_n^\rD} \ket 0 \pm \sket{\phi_{n+1}^\rD} \ket 1 ] \; .
 \end{align}
 These are approximations of the eigenstates of  $H^\rD_{\rm JC}$.
 In Appx.~D, Lemma~\ref{lem:DJCH1}, we show that if $n \le N' \le cN$, for some constant $c>0$,
 \begin{align}
 \label{eq:DJCHspectrum}
 \| (H^\rD_{\rm JC} \mp \sqrt{n+1} ) \sket{\gamma^\rD_{n,\pm}} \| = \nu_1(N)  \; .
 \end{align}
 We use Eq.~\eqref{eq:DJCHspectrum} to prove the second result of this section.
 \begin{theorem}
 \label{thm:eigenstatepreparation}
 Let $\epsilon >0$ and $t_n = \pi/(2 \sqrt{n+1})$. Then, there exists $N = O(\log(1/\epsilon)+ N') $
 such that
 \begin{align}
 \| e^{-i H_{\rm JC}^\rD {t_n}} \sket{\phi_n^\rD} \ket 0 + i  \sket{\phi_{n+1}^\rD} \ket 1 \| = O(\epsilon) \; ,
 \end{align}
 for all $n \le N'$.
 \end{theorem}
 \begin{proof}
 We let $N \ge N'/c$. Then, Eq.~\eqref{eq:DJCHspectrum} implies
 \begin{align}
 \| [ e^{-i H_{\rm JC}^\rD {t_n}} - e^{- i (\pm \sqrt{n+1} t_n)} ] \sket{\gamma^\rD_{n,\pm}} \| = |t_n| \nu_1(N)  \; ,
 \end{align}
 for all $n \le N'$. For $t_n=  \pi/(2 \sqrt{n+1}) $, this implies
 $ \| (e^{-i H_{\rm JC}^\rD t_n} \pm i ) \sket{\gamma^\rD_{n,\pm}} \| = \nu_1(N)$,
 and then
  \begin{align}
  \label{eq:DJCStateError}
 \| e^{-i H_{\rm JC}^\rD t_n} \sket{\phi_n^\rD} \ket 0 + i  \sket{\phi_{n+1}^\rD} \ket 1 \| =\nu_1(N)  \; .
 \end{align}
 Then, there is $N = O(\log(1/\epsilon))$ so that the right hand side of Eq.~\eqref{eq:DJCStateError} is $O(\epsilon)$.
\end{proof}
 
 We can combine Thms.~\ref{thm:groundstatepreparation} and~\ref{thm:eigenstatepreparation} to prepare approximations of other eigenstates $\sket{\phi_n^\rD}$, $n \ge 1$,
 by a sequential action of $(\one \otimes  \sigma_x ) . e^{-i H_{\rm JC}^\rD t_n} $ on the initial state $\sket{\phi_0^\rD} \ket 0$.
 The Pauli operator $\sigma_x$ is necessary  to transform $\ket 1 \rightarrow \ket 0$ for the state of the ancilla qubit at each step.
Then,
 \begin{align}
 \label{eq:eigenstateprep1}
\sket{\phi_n^\rD} \ket 0 \approx \prod_{n'=0}^{n-1} (\one \otimes  \sigma_x ) e^{-i H_{\rm JC}^\rD t_{n'}} \sket{\phi_0^\rD} \ket 0\; .
\end{align}
 
 We now seek a quantum algorithm to prepare
 the eigenstates $\sket{\phi_n^\rD}$, up to a given approximation error.
 This requires giving a quantum circuit to approximate each $e^{-i H_{\rm JC}^\rD t_{n'}}$ in Eq.~\eqref{eq:eigenstateprep1}.
 Since the unitaries $e^{-i (x^\rD \otimes \sigma_x) t}$ and $e^{-i (p^\rD \otimes \sigma_y) t}$ can be simulated within precision $\tilde \epsilon$ using
 a number of two-qubit gates that is ${\rm polylog}(N |t|/\tilde \epsilon)$ (Sec.~\ref{sec:TSdecomp}), we will use the Trotter-Suzuki
 approximation. In Appx.~D, Lemma~\ref{lem:DJCevol}, we show that if $s = O(\epsilon / \sqrt{n+1})$ for some $n \le N' \le cN$, then
 \begin{align}
 \| [e^{-i (x^\rD \otimes \sigma_x) s/\sqrt 2} e^{i (p^\rD \otimes \sigma_y) s/\sqrt 2} - e^{-i H_{\rm JC}^\rD s}]\sket{\gamma^\rD_{n,\pm}} \| = O(\epsilon^2) \; .
 \end{align} 
 The proof uses a simple Trotter-Suzuki approximation and the scaling with $\epsilon$ can be improved using higher order approximations.
 We can use this result to prove:
 \begin{theorem}
 Let $\epsilon >0$. Then, there is a quantum circuit $W^\rD$ that satisfies
 \begin{align}
 \| W^\rD \sket{\phi_0^\rD} \ket 0 - \sket{\phi_n^\rD} \ket 0 \| = O(\epsilon) \; 
 \end{align}
 for any given $n \le N' \le cN$, where $N$ is the dimension of the Hilbert space and $c>0$ is a constant. $W^\rD$ can be implemented using a number of two-qubit gates
 that is $O((n^2/\epsilon) {\rm polylog}(N/\epsilon))$.
 \end{theorem}
\begin{proof}
 The quantum circuit is
 \begin{align}
W^\rD =
 \prod_{n'=0}^{n-1} (\one \otimes  \sigma_x ) \left [e^{-i (x^\rD \otimes \sigma_x) t_{n'}/(m\sqrt 2)} e^{i (p^\rD \otimes \sigma_y) t_{n'}/(m\sqrt 2)} \right]^m .
 \end{align}
 Here, $m = O(n/\epsilon)$ so that each term in $W^\rD$ introduces an error $O(\epsilon/n)$ as implied by Lemma~\ref{lem:DJCevol}.
 Because we work in the asymptotic limit, approximation errors that are exponentially small in $N$ or $\sqrt N$ are negligible. Then, 
 \begin{align}
  \| W^\rD \sket{\phi_0^\rD} \ket 0 - \sket{\phi_n^\rD} \ket 0 \| = O((\epsilon/n)n) = O(\epsilon)\; .
 \end{align}
 The number of terms in the product is $O(n^2/\epsilon)$. The number of two-qubit gates is then
 \begin{align}
 O((n^2/\epsilon) {\rm polylog}(N/\epsilon)) \; .
 \end{align}
\end{proof}

\section{The quartic potential}
\label{sec:quartic}
We now analyze quantum algorithms to simulate the evolution operator of a
 quantum system with Hamiltonian
$H = \frac 1 2 (\hat p^2 + \hat x^4)$.
We will use the same discretization as that for the DQHO, where $x^\rD$ and $p^\rD$ where defined in Sec.~\ref{sec:DQHO}. Then,
\begin{align}
\label{eq:DQuartic}
H^\rD = \frac 1 2 ((p^\rD)^2 + (x^\rD)^4) \; .
\end{align} 
In contrast to previous sections, we only conduct a numerical analysis here
and state some observations from numerical results. In part, this is due to not
having an exact solution in this case. Our simulations suggest a polynomial
speedup for the computation of scattering amplitudes.

In Fig.~\ref{fig:x4Eigenvalues} we plot the eigenvalues $E_n^\rD(N)$ of $H^\rD$ for different dimensions $N$ and
as a function of $n=0,1,\ldots,N-1$.
Taking dimension $4000$ as a reference, in Table~\ref{table:x4Eigenvalues} we look for the maximum value of $n$ such that 
$|E_n^\rD (N) - E_n^\rD (4000)|$ is below some threshold $\epsilon^i \ll 1$.
The values of the ratios $r^i$ suggest that the eigenvalues of a large sector of the low-energy subspace of the CV
system can be well approximated by those of the discrete system. While $r^i$ does seem to decrease in $N$,
the scaling does not seem to be of the form $1/N^\chi$, for some $\chi>0$, but rather of the form $1/\log N$.
If this is the case, then, to approximate up to the $n$-th eigenvalue of $H$, it suffices to choose $N = O(n \log n)$
for the discrete system. Nevertheless, both the classical and quantum algorithms to simulate this system
will have complexity that depend on the same value of $N$, and the dependence of $r^i$ on $N$ is unimportant to demonstrate
a quantum speedup.

 \begin{figure}[htbp]
\includegraphics[width=8cm]{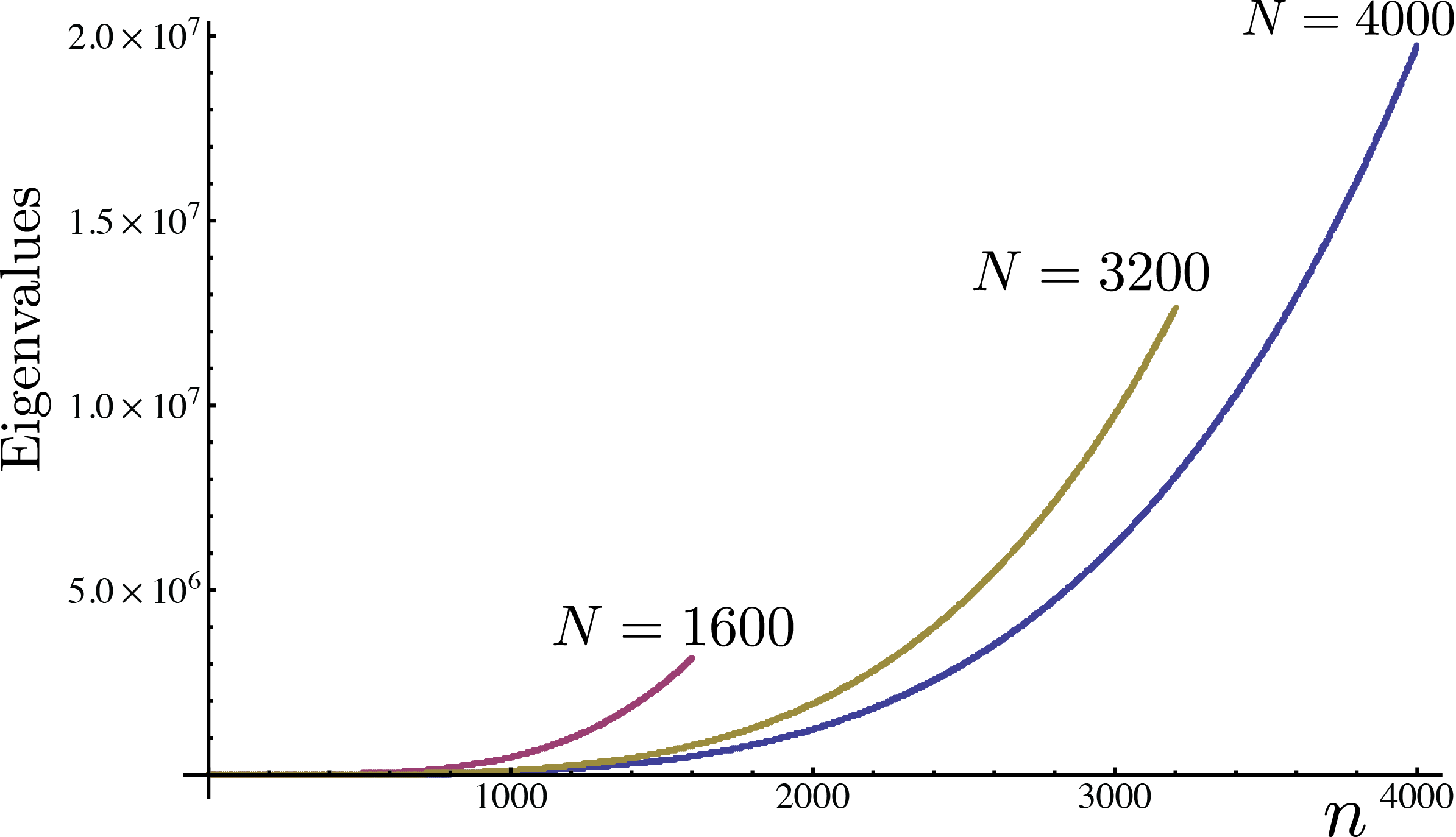}
\caption{\label{fig:x4Eigenvalues} Eigenvalues of $H^\rD$ for  $N=4000$ (blue), $N=3200$ (yellow), and $N=1600$ (purple).}
\end{figure}

\begin{table}[htb]
\centering
\begin{tabular} {c | c | c | c | c}
 $N $& $n^{1}$ & $r^{1}$ & $n^{2}$  & $r^{2}$   \\
 \hline
 100 & 16 &   0.16 & 12 & 0.12 \\
 200 & 30 &  0.15  & 27 & 0.135 \\
 400 & 57 & 0.1425   & 54 & 0.135 \\
 800 & 105 &   0.1312 & 99 & 0.125 \\
 1600 & 181 & 0.1131   & 177 & 0.1106 \\
 2000 & 216 & 0.108   & 212 & 0.106 \\
 3200 & 314 & 0.0981   & 310 & 0.0968\\

 \end{tabular}
 \caption{ \label{table:x4Eigenvalues}
$n^i$, maximum value of $n$ such that  $|E_n^\rD (N) - E_n^\rD (4000)|\le \epsilon^i$, for $\epsilon^1=10^{-5}$
 and $\epsilon^2=10^{-7}$. The ratios are $r^i = n^i/N$.}
\end{table} 

We can use the Trotter-Suzuki approximation to simulate the evolution operator $U^\rD(t) = \exp \{- i H^\rD t \}$.
This approximation splits the evolution operator into a product of exponentials or short time evolutions under $(p^\rD)^2$ and $(x^\rD)^4$ and,
using the result of Sec.~\ref{sec:TSdecomp}, each of these exponentials can be simulated efficiently. In contrast to the QHO, the operators $\hat x^4$
and $\hat p^2$ do not form a small dimensional Lie algebra and tight error bounds from high-order Trotter-Suzuki approximations
may be difficult to obtain. The recursive definition for $U_p^\rD(s)$ is given in Eq.~\eqref{eq:TSArecdef}, assuming
\begin{align}
 U_1^\rD(s):= e^{-i s (x^\rD)^4/4 }
 e^{-i s (p^\rD)^2/2 }  e^{-i s (x^\rD)^4/4 } \; .
\end{align}

A worst-case analysis of the Trotter-Suzuki formula results in an approximation error bounded by $\epsilon_p(s)=O(|s|^{2p+1} \| H^\rD \|^{2p+1}) = O((|s|N^2)^{2p+1})$~\cite{trotter_1959,Huyghebaert_1990,suzuki_90,suzuki_qmc_1998,berry_efficient_2007,wiebe_product_2010}.
However, as in the case of the DQHO, 
we would expect that the error for the current case is significantly smaller than that for the worst case.
This is because the operators in $H$, while they do not form a finite dimensional Lie algebra, posses an algebraic structure
that results in an effective norm for nested commutators that is significantly smaller than the product of the effective norms~\cite{Som_15}.
In Fig.~\ref{fig:TSErrorx4p} we plot the error $\| (U_p^\rD(s) - U^\rD(s)) \sket{\phi_n^\rD} \|$ as a function of $n$, for $|s|=O(1)$, and $p=1,2,3$.
Here,   $\sket{\phi_n^\rD}$ are the eigenstates of $H^\rD$ in Eq.~\eqref{eq:DQuartic} and cannot be approximated
by the discrete Hermite states $\sket{\psi_n^\rD}$. The results suggest $\| (U_p^\rD(s) - U^\rD(s)) \ket{\phi_n^\rD} \| = O(|s|^{2p+1} n^{(2p+4)/3})$.
The order dependence in $s$ follows from the analysis of the high-order Trotter-Suzuki approximations and was verified
by additional numerical simulations, to assure that we are in a region of convergence.

 \begin{figure}[htbp]
\includegraphics[width=8cm]{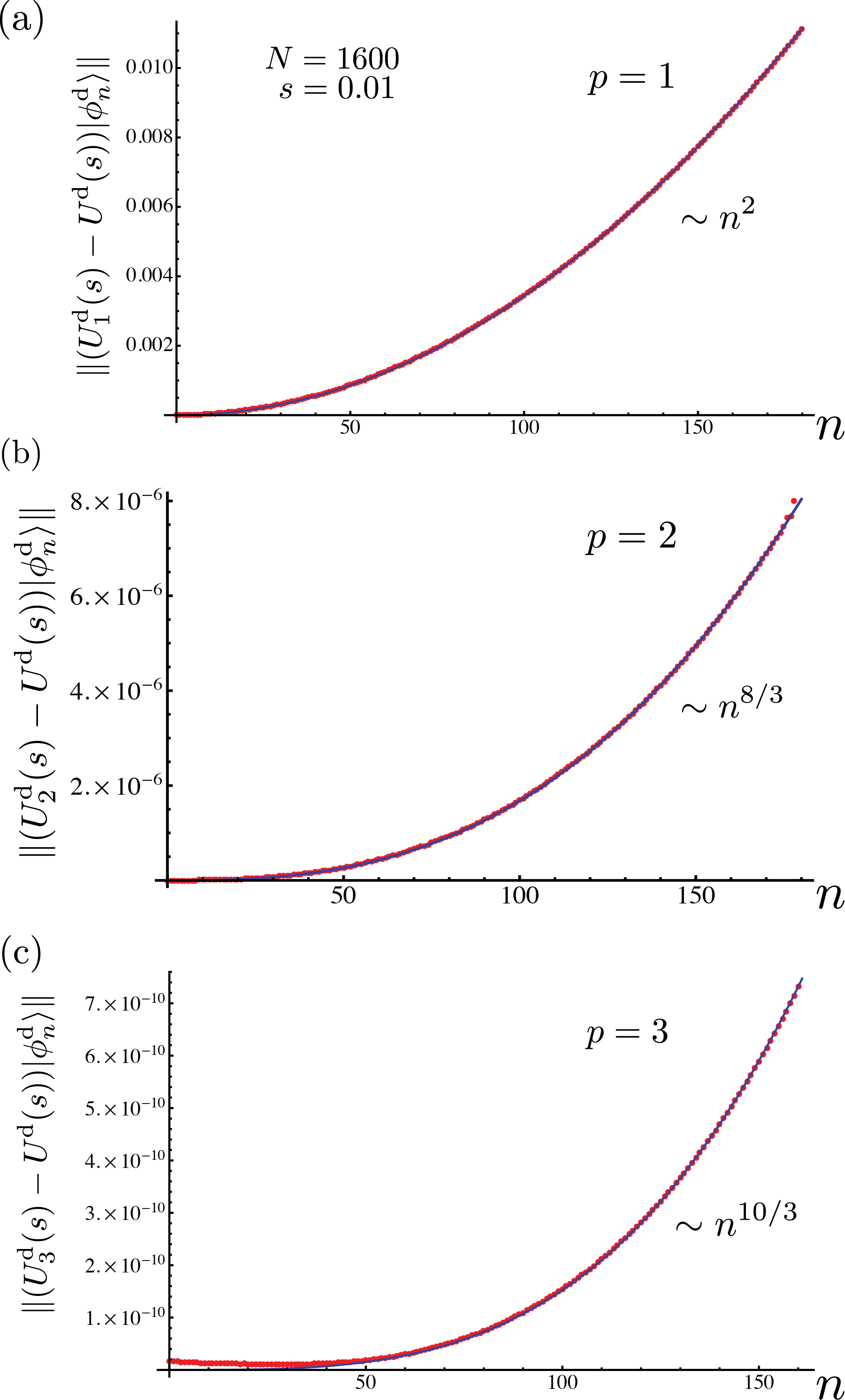}
\caption{\label{fig:TSErrorx4p}
The error from high-order Trotter-Suzuki approximations (red dots) as a function of $n$. 
The results are for dimension $N=1600$, $s=0.01$, and for $n=0,1,...,181$, according to Table~\ref{table:x4Eigenvalues}.
Our simulations suggest that the error is  $ O ( s^{2p+1} n^{p'})$, where $p'$ depends on $p$.
(a) $p=1$. The blue line indicates a fit with the function $f_1= 3.44.10^{-7} \; n^2$.  (b) $p=2$. The blue line indicates a fit with the function $f_1= 7.8.10^{-12}n^{8/3}$. (c)
$p=3$. The blue line indicates a fit with the function $f_3= 3.3.10^{-17} \; n^{10/3}$. }
\end{figure}

Under these numerical observations, 
we can analyze the complexity of a quantum algorithm
that computes scattering amplitudes
\begin{align}
\label{eq:x4propagator}
\bra{\varphi'} U(t) \ket{\varphi} \;,
\end{align}
within some  precision $\epsilon$. For simplicity, we assume that the initial and final states, $\sket{\varphi}$
and $\sket{\varphi'}$, can be written as linear superpositions of the eigenstates $\sket{\phi_n}$, with $n \le N'$.
This suggests that we can approximate Eq.~\eqref{eq:x4propagator} by
$
\sbra{\varphi'^{\rD}} U^\rD(t) \sket{\varphi^\rD}$ ,
where the dimension is $N = O(N' \log (N'))$, and $\sket{\varphi^\rD}$ and $\sket{\varphi'^\rD}$ are quantum states
obtained by replacing $\sket{\phi_n}$ by $\sket{\phi_n^\rD}$ in the spectral decompositions of 
$\sket{\varphi}$ and $\sket{\varphi'}$, respectively.
 We assume that $\sket{\varphi^\rD}$ and $\sket{\varphi'^\rD}$
can be efficiently prepared so that the main cost of the algorithm is that from simulating $U^\rD(t)$.
We split the evolution time into $k$ parts of size $s=t/k$. Using the $(2p+1)$-th order Trotter-Suzuki approximation,
our numerical simulations suggest that the error is bounded by
$k \omega (\omega' \; |s|)^{2p+1} N^{(2p+4)/3}$. $\omega$ and
 $\omega' > 1$ are constants. Replacing $s$ by $t/k$, the number of terms
in the product formula is
\begin{align}
k 5^p = O\left(   \frac{ |t|^{1+1/2p} N^{(1+2/p)/3}}{\epsilon^{1/2p}} 5^p \right) \; .
\end{align}
As $p$ grows large, the number of terms can be made $O(|t|^{1+\eta} N^{1/3 + 4 \eta/3} /\epsilon^\eta)$, for arbitrary small $\eta$.
Then,  the quantum circuit that approximates the evolution operator $U^\rD(t)$,
needed to compute Eq.~\eqref{eq:x4propagator},
can be implemented using a number of two-qubit gates that is
\begin{align}
{ \cal M_Q}= O \left (\frac{|t|^{1+\eta} N^{1/3 + 4 \eta/3} }{\epsilon^\eta} {\rm polylog}(Nt/\epsilon)\right) \; .
\end{align}
This complexity represents a polynomial quantum speedup, with respect to $N$, over the quantum algorithm
that approximates scattering amplitudes for the quartic potential. Classical algorithms,
in the worst-case, may require computing and obtaining the spectral properties of $U^\rD(t)$, which can be done 
in complexity $O(N^\sigma)$, for $\sigma>2$.

\section{One-dimensional quantum systems: An upper bound}
\label{sec:bound}
We now present an upper bound on the complexity of simulating 
the evolution operator of one-dimensional quantum systems described
by
$H = \frac 1 2 \hat p^2 + V(\hat x)$, where $V(\hat x)$ is the potential, i.e., some operator that depends on $\hat x$.
If $\ket{\phi_n}$, $n=0,1,\ldots$, denote the eigenstates of $H$ and $\ket{\varphi}$ is the initial state, we will assume
$\ket{\varphi} = \sum_{n=0}^{N'} c_n \ket{\phi_n}$ and $\| V(\hat x) \ket{\phi_n} \| = O({\rm poly}(N'))$, for all $n \le N'$.

We also assume that there exists $N<\infty$, the dimension of the Hilbert space, such that, if
\begin{align}
H^\rD = \frac 1 2  (p^\rD)^2 + V( x^\rD) \; ,
\end{align}
the scattering amplitudes of the CV system can be well approximated by those of the discrete system. 
$N$ will depend on $N'$ and the precision parameter $\epsilon$.
In particular, for constant precision, we will assume that $N=O({\rm poly}(N'))$,
an assumption that is satisfied by a large class of one-dimensional quantum systems 
such as the QHO or the quartic potential.
\begin{theorem}
Let $\epsilon>0$ and $t$ be the evolution time, $|t|=\Omega(1)$. Then, there is a quantum circuit $W^\rD$ that satisfies
\begin{align}
\| (W^\rD - U^\rD(t)) \ket{\varphi^\rD} \| \le \epsilon \; ,
\end{align}
for all $\ket{\varphi^\rD} \propto \sum_{n=0}^{N'} c_n \sket{\phi^\rD_n}$. Here, $\sket{\phi^\rD_n}$
are the eigenstates of $H^\rD$ and $U^\rD(t)=\exp \{-i H^\rD t \}$.   $W^\rD$ can be implemented
using $ O(N|t| {\rm polylog}(N|t|/\epsilon))$ two-qubit gates.
\end{theorem}
\begin{proof}
We use our results in~\cite{BCC_Taylor_2015} for Hamiltonian simulation
and do some modifications to obtain the desired result. 
In~\cite{BCC_Taylor_2015} we showed that for a $d$-sparse time dependent Hamiltonian $A(t)$ acting on $q$ qubits,
the evolution operator can be approximated within precision $\epsilon$ using
\begin{align}
\label{eq:gatecostHS}
O(q d^2 \|A\|_{\max} |t| \frac{\log((d^2 \|A\|_{\max} + \| \dot A \|) |t| /\epsilon)}{\log\log(d^2 \|A\|_{\max}| t|/\epsilon)}
\end{align}
two-qubit gates, where $\dot A = \partial_t A$ and the norms are the maximum norms 
in the time interval $[0,t]$. The gate cost results from a decomposition
of $A(t)$ in terms of $O(d^2 \|A\|_{\max})$ unitary operators.
 Since $H^\rD$ is not sparse and its norm can be larger than $O(N)$,
Eq.~\eqref{eq:gatecostHS} would result in a large gate complexity in this case.
To overcome this difficulty, we analyze the evolution operator in the interaction
picture. We then define the time-dependent interaction Hamiltonian
\begin{align}
\nonumber
H^\rD_I(s) &= \frac 1 2 e^{i V(x^\rD) s} (p^\rD)^2 e^{-i V(x^\rD) s} \\
& =  \frac 1 2 e^{i V(x^\rD) s} (F_{\rm c}^\rD)^{-1}(x^\rD)^2 F_{\rm c}^\rD e^{-i V(x^\rD) s}  \; ,
\end{align}
and denote $U^\rD_I(t)$ for the corresponding evolution. Since $x^\rD$ is 1-sparse
and $\|x^\rD \|=O(\sqrt N)$, we obtain $\| H^\rD_I \|_{\max} = O(N)$ and $H^\rD_I(s)$ can be decomposed
as a sum of $O(N)$ unitary operators. In addition, $\| \dot H_I^\rD \| = O({\rm poly}(N))$ under the assumptions and $q=O(\log N)$.
This implies that the gate complexity to simulate $U_I^\rD$ as given by Eq.~\eqref{eq:gatecostHS} is
\begin{align}
O(N   |t|  \log^2(N |t| /\epsilon) ) = \tilde O(N|t|)
\end{align}
if $|t|=\Omega(1)$.
We note that $U^\rD(t) =e^{-i V(x^\rD) t}  U_I^\rD(t) $
and that $ e^{-i V(x^\rD) t}$ can be efficiently simulated with complexity ${\rm polylog}(N|t|/\epsilon)$
as explained in Sec.~\ref{sec:TSdecomp}.
\end{proof}

\section{The fractional Fourier transform and Related work}
\label{sec:frft}
\label{sec:frft}
The evolution induced by the QHO results in a transformation
referred to as the fractional Fourier transform (frFT), which corresponds to an
arbitrary rotation in phase space. The frFT has been proven useful in signal analysis~\cite{Alm94,LZ03,Jou04},
in noise filtering in particular, when the noise does not have a well defined frequency spectrum.
The evolution operator $U^\rD(t)$ can then be interpreted as an approximate version to a ``discrete''
frFT, and our results may prove useful in the design of classical or quantum algorithms for discrete signal analysis.
This would require efficient methods for encoding and decoding of signals, which
may exist under some assumptions such as sparsity.

The split-step (Fourier) method is a well known technique to solve the nonlinear Schr\"odinger equation
in quantum mechanics and in fiber optics (c.f.,~\cite{TA84}). As in our case, the idea is to evolve the initial state according
to small step evolutions under the corresponding operators. This also requires a discretization
of the continuous-variable coordinates. Our results on Trotter-Suzuki
approximations can then be used and generalized to bound the errors induced by the split-step method.

The development of quantum simulation methods for continuous-variable quantum systems, including quantum chemistry, 
is an active area of research (c.f.,~\cite{kassal_chem_2011,yung_chem_2012,wecker_chem_2013} and references therein). Commonly, the complexity of such methods is polynomial
in the energy of the system. As an example, in~\cite{papageorgiou_CV_2013}, a quantum algorithm for approximating the ground state energy
of a continuous-variable quantum system is provided. The algorithm works under some assumptions on the energy potential and its complexity
is proportional to $d$, the number of state variables (or particles). Another example is~\cite{jordan_QFT_2012}, which presents a quantum algorithm to compute scattering probabilities
in a certain quantum field theory ($\phi^4$ theory). The complexity of such algorithm is also polynomial in the energy. In contrast, our quantum algorithm to simulate the time evolution operator of the QHO has complexity that is subexponential in $\log N'$, where $N'$ 
denotes the relevant energy scale of the problem. 
The QHO provides a basis 
for the quantization of the electromagnetic field and quantum field theories, and our algorithms
are expected to find wide applications in quantum simulation.

\section{Conclusions}
\label{sec:conc}
We provided a quantum algorithm to approximate the propagator of the QHO within arbitrary accuracy.
For  precision $\epsilon>0$, the complexity of the algorithm is ${\cal M_Q}=O(\exp (\tilde \gamma \sqrt{\log (N/\epsilon)}))$, where 
the evolution time can be assumed to be constant, and $N$ is the relevant energy scale of the simulation. Asymptotically,  ${\cal M_Q} / (N)^\eta \rightarrow 0$, for any $\eta>0$,
so the complexity of the algorithm is subexponential in $\log( N)$. 
Remarkably, this represents a superpolynomial speedup over 
the corresponding classical algorithm to compute the propagator, whose complexity is $O(N)$.
Our results consider a refined analysis of the error of high-order Trotter-Suzuki approximations. This analysis works in this case because
the operators under consideration form a Lie algebra of dimension 3 [i.e., $sp(2)$]. We can then use properties of commutators
to show that the error induced by a high-order Trotter-Suzuki approximation is significantly smaller than that obtained considering the worst-case scenario;
recent results consider this problem more generally~\cite{Som_15}. 
Our quantum algorithm considers a discrete version of the QHO whose low-energy spectrum can be shown
to reproduce the properties of the QHO with very high accuracy (i.e., the approximation errors decay exponentially with $N$).

We also provided quantum algorithms to compute spectral properties by preparing approximations of the eigenstates of the QHO.
These algorithms have complexity polynomial in $\log(N)/\epsilon$ and may be of independent interest (e.g., the ground state
has Gaussian-like amplitudes). Here, $N$ is the number of points in the discretization or dimension of the Hilbert space.
To prepare such states, our quantum algorithms simulate the evolution induced by a version of the Jaynes-Cummings
model that is  used in quantum optics. This evolution is also approximated using a high-order Trotter-Suzuki approximation.

Last, we presented quantum algorithms to simulate more complex one-dimensional quantum systems. For the case
of a quartic potential, we presented numerical evidence for the existence of a method that simulates the evolution operator
with complexity ${\cal M_Q}=\tilde O(N^{1/3 + o(1)})$, if $t$ and $\epsilon$ are constant.
This method is also based on high-order Trotter-Suzuki approximations.
Our result represents a polynomial quantum speedup over the classical method in this case. The quantum advantage may be a result
of the algebraic structure satisfied by the operators in the Hamiltonian~\cite{Som_15}.
We also showed how general quantum systems can be tackled on the basis of recent results in~\cite{BCC_Taylor_2015}
that simulate the evolution operator by implementing its Taylor series decomposition.  We proved
an $\tilde O(N)$  bound for the complexity of simulating the evolution operator under fairly general assumptions.

We conjecture that some of our results can be generalized to provide subexponential time quantum algorithms to simulate
the evolution   of other continuous-variable quantum systems.  


\acknowledgements
We thank R. Cleve, S. Gharibian, N. Nguyen and A. Chowdhury for discussions. 
 We acknowledge support from the NSF through the CCF program and the Laboratory Directed Research and Development
 program at LANL.



\appendix

\section{Properties of the discrete QHO}
 \label{appendix:spectralproperties}
 Most of our results can be obtained from approximations of integrals appearing
 in the continuous-variable case as finite sums appearing in the discrete variable case.
 For completeness, the Hermite functions are
  \begin{align}
  \psi_n(x) = \frac 1 {\sqrt{2^n n! \sqrt \pi}} e^{-x^2/2} H_n(x) \; ,
  \end{align}
  where $H_n(x)$ is the (physicists') $n$-th Hermite polynomial, $n \ge 0$,
  \begin{align}
  H_n(x)=(-1)^n e^{x^2} \frac{ \partial^n (e^{-x^2})}{\partial x^n }\; .
  \end{align}
  The orthogonality (or normalization) condition is $\int dx \; \psi_m(x) \psi_n(x)=\delta_{n,m}$, and
  a useful property is $x \psi_n(x) = \sqrt{n/2}\psi_{n-1}(x)+\sqrt{(n+1)/2} \psi_{n+1}(x)$.
  In Dirac's bra-ket notation, these are $\sbra{\psi_m} \psi_n \rangle = \delta_{n,m}$ and $\hat x
  \sket{\psi_n} = \sqrt{n/2}\sket{\psi_{n-1}}+\sqrt{(n+1)/2}\sket{ \psi_{n+1}}$, respectively.
  The Hermite functions are eigenfunctions of the Fourier transform with eigenvalues $(-i)^n$, $n \ge 0$.
  Then, we can write $\hat F$ for an operator that applies the Fourier transform and $\hat F \sket{\psi_n} = (-i)^n \sket{\psi_n}$.
 
  Unless a region of integration is given explicitly, we assume $x \in (-\infty,\infty)$. Similarly, sums
  are infinite unless otherwise stated. In all cases,  $l \ge 0$, $k$, and $n \ge 0$ are  integer numbers.
  $N \ge 4$ denotes the dimension of the Hilbert space for the discrete QHO.
  
 We will show that our main results follow from statements about the ``tails'' of the Hermite
 functions for $x \notin [-kT/2,kT/2]$, where $T=\sqrt{2 \pi N}$ and $k \ge 1$. 
 For even $n$, the Hermite polynomials satisfy
  \begin{align}
  |H_n(x)| \le   \frac{n!}{(n/2)!} e^{x \sqrt{2 n}} \; ,
  \end{align}
and there is a similar upper bound for odd $n$ (see inequality 8.954 of ~\cite{IntTable}).
  Stirling's approximation implies $c_1 \sqrt{n} (n/e)^n \le n! \le c_2 \sqrt{n} (n/e)^n$,
  with $0< c_1=\sqrt {2\pi} < c_2=e$. Then
  \begin{align}
  \label{eq:HB1}
  | \psi_n(x)| \le c_3 e^{-x^2/2} e^{x \sqrt{2 n}}
  \end{align}
  for some constant $c_3 \approx 0.7$ and $x \ge 0$.  In the event that $x/2-\sqrt{2n} \ge \tilde \beta \sqrt N$, for some $\tilde \beta>0$,
  Eq.~\eqref{eq:HB1} implies $\psi_n(x) = O(\exp(-x \tilde \beta \sqrt N))$. A similar result is obtained for odd $n$.
  Also (\cite{IntTable}),
  \begin{align}
  |H_n(x)|< c_4 \sqrt{n!}2^{n/2}e^{x^2/2} \;,
  \end{align}
 which implies
 \begin{align}
 \label{eq:HB2}
 |\psi_n(x)| < 1 \;,
 \end{align}
  where $c_4 \approx 1.0864$.

 The results described in Secs.~\ref{sec:spectralprop} and~\ref{sec:DQHOexpval}  will follow from the following lemmas, which suffice to provide analytical proofs. For simplicity, we use $\nu_1(N)$ to denote the order of a function that decays exponentially in $N$.
 That is, $\nu_1(N)= \exp(-\Omega (N))$, and there exists a constant $\beta >0$ such that, if $f(N) =\nu_1(N)$, then $f(N)  \le \exp(-\beta N)$.
 Note that $(N^\alpha \exp(-\Omega (N)))$  and  $\nu_1(N)^\alpha$ are also $\exp(-\Omega (N))$
 for any constant $\alpha>0$ and sufficiently large $N$.
Then, if $f(N)=N^\alpha \nu_1(N)$ or $f(N)= \nu_1(N)^\alpha$, we have $f(N)=\nu_1(N)$.
 The value of the constants appearing in these lemmas, such as $\beta$, can be improved with more detailed analyses.
  \begin{lemma}
 \label{lem:smallsupport}
Given $l \ge 0$, there exists a constant   $c >0$ such that, for all $ n \le cN$ and all $k \ge 1$,  
 \begin{align}
 \int_{kT/2}^\infty dx \; (\psi_n(x))^2 x^l  =  (\nu_1(N))^k\; .
 \end{align}
 \end{lemma}
  \begin{proof}
  We will choose $c>0$ so that  $x/2-\sqrt{2n} \ge T/4-\sqrt{2n} \ge \tilde \beta \sqrt N$, for some $\tilde \beta>0$,
  in the integration region.
  For example, we can choose $c < \pi/16$ and $\tilde \beta$ to satisfy $\tilde \beta = (\sqrt{\pi/8}-\sqrt{2c}) >0$.
  Then, Eq.~\eqref{eq:HB1} implies $\psi_n(x) = O(\exp(-x \tilde \beta \sqrt N))$ for $x \ge T/2$.
  If $I_{k}$ is the integral of interest, explicit calculation  
  gives
   \begin{align}
    \label{eq:IB1}
  I_{k} &= O(\exp(-kT \tilde \beta \sqrt N)(k \sqrt N)^l) \;.
    \end{align}
    Since $l=O(1)$, $T = \Omega(\sqrt N)$, and $(k \sqrt N)^l) = \exp(l \log(k \sqrt N))$, we obtain $I_k = (\nu_1(N))^k$.
    
 
\end{proof}

We note that the value of $\tilde \beta$ in the exponential of Eq.~\eqref{eq:IB1}
depends on $c$ (or $n$). For example, if we are interested in the case where $n=0$, we can choose
$\tilde \beta = \sqrt{\pi/8}$ so that $I_k= O(\exp(-k (\pi/2)N)(k \sqrt N)^l)$. 

 
 \begin{lemma}
 \label{lem:D-CVposition}
Given $l \ge 0$, there exists a constant  $c>0$ such that, for all $ n,m \le c N$,  
\begin{align}
 \left | \sbra{\psi_m^\rD} (x^\rD)^l \sket{\psi_n^\rD} - \sbra{\psi_m } \hat x ^l \sket{\psi_n }  \right| =  \nu_1(N)  \; .
\end{align}
 \end{lemma}
 \begin{proof}
 From the property of the Hermite functions, we note
 \begin{align}
 (\hat x)^l \sket{\psi_n}  =\sum_{n' \le n+l} c_{n'} \sket{\psi_{n'}} \Rightarrow
  ( x^\rD)^l \sket{\psi_n^\rD} =\sum_{n' \le n+l} c_{n'} \sket{\psi^\rD_{n'}}  \; ,
 \end{align}
 where $|c_{n'}|=O(n^{l/2})$.
 Since $l=O(1)$, the lemma follows from showing that $\langle \psi^\rD_m \sket{\psi^\rD_{n}}$ approximates
 $\delta_{n,m}$ within precision that is exponentially small in $N$.
 
 The remainder of the proof has two parts. First, we will show that 
  \begin{align}
 \label{eq:sum1}  
  \left({2 \pi}/N\right)^{1/2}  \sum_j \psi_m(x_j)\psi_n(x_j)   \; ,
  \end{align}
  with $x_j = j \sqrt{2 \pi/N}$, approximates $\langle{\psi_m^\rD}   \sket{\psi_n^\rD}$ at the 
  desired order. The Cauchy-Schwarz inequality implies
\begin{align}
\nonumber
& |  \sum_{j=N/2}^\infty  \psi_m(x_j) \psi_n(x_j)   | ^2 \le \\
& \le
 \left( \sum_{j=N/2}^\infty  \psi_m^2(x_j)   \right)\left(  \sum_{j=N/2}^\infty  \psi_n^2(x_j)    \right) \; .
\end{align}
In this case, $x_j \ge \sqrt{\pi N/2}$. We will choose $c$ so that $x_j/2 - \sqrt{2n} \ge T/4 - \sqrt{2n} \ge \tilde \beta \sqrt N$,
for some $\tilde \beta >0$. That is, $c$ and $\tilde \beta$ can be those of Lemma~\ref{lem:smallsupport}, and $n \le cN$. Then, Eq.~\eqref{eq:HB1} implies $\psi_n(x_j) = O(\exp(-x_j \tilde \beta \sqrt N))$ for $x_j \ge T/2$ or, equivalently, for $j \ge N/2$. Explicit calculation of the sum implies  
\begin{align}
(2\pi/N)^{1/2} \! \! \! \sum_{j=N/2}^\infty   \psi_n^2(x_j)   =  \nu_1(N)   \; ,
\end{align}
for sufficiently large $N$.
This coincides with the result of Lemma~\ref{lem:smallsupport} when $k=1$.  
The same result can be obtained if we replace $n \rightarrow m$, with the   assumption that $m \le cN$.
Then, from the Cauchy-Schwarz inequality,
 \begin{align}
\left |\sqrt{ \frac{2 \pi} {N} } \! \!\sum_{j=N/2}^\infty  \psi_m(x_j) \psi_n(x_j)   \right | = \nu_1(N)   \; ,
  \end{align}  
  implying that $\langle{\psi_m^\rD}   \sket{\psi_n^\rD}$
  can be approximated by Eq.~\eqref{eq:sum1} within precision that is exponentially small in $N$.
  
  In the second part of the proof, we will show that Eq.~\eqref{eq:sum1} also approximates
  the expectation $\langle{\psi_m }   \sket{\psi_n }=\delta_{n,m}$.
   We use the identity $\sum_j \delta(x-x_j)=(N/(2\pi))^{1/2} \sum_{k}e^{-ikTx}$, where $\delta(x)$ is the Dirac delta,
  to write Eq.~\eqref{eq:sum1} as
  \begin{align}
  \label{eq:sum2}
  \sum_k \int dx \; \psi_m(x) \psi_n(x)  e^{-ikTx} =\sum_k \sbra{\psi_m} e^{-ikTx} \sket{\psi_n} \; .
  \end{align}
The term with $k=0$ is $\sbra{\psi_m } {\psi_n } \rangle$, so we need to show that the 
sum of the terms with $k \ne 0$ is small and satisfies the desired bound. 
 The Hermite functions are also eigenstates of the Fourier transform and $ \sbra{\psi_m}e^{-ikT \hat x} \ket{\psi_{n}}
 = (-i)^{n-m}\sbra{\psi_m}e^{-ikT \hat p} \ket{\psi_{n}}$, where we used $\hat F \hat x (\hat F)^\dagger = \hat p$.
 Note that $e^{-ikT \hat p}$ is the space translation operator. 
   Then, each term of Eq.~\eqref{eq:sum2} can be written as
\begin{align}
  \label{eq:sum4}
 i^{m-n}    \int dx \; \psi_m(x) \psi_{n}(x+kT)   \; ,
  \end{align}
  because  $\exp\{-kT\partial_x\}\psi_{n}(x)=\psi_{n}(x+kT)$.

   We are then interested in showing that  $\int dx \; \psi_m(x) \psi_{n}(x+kT)$ is small when $k \ne 0$.
  From symmetry arguments, it suffices to analyze the case $k \ge 1$ only. We write $\psi_m(x) = \tilde \psi_m(x) +
  \phi_m(x)$, where $\tilde \psi_m(x) = \psi_m(x)$ if $x \le kT/2$ and $\tilde \psi_m(x)=0$ otherwise.
  We use a similar decomposition for $\psi_n(x)=\hat \psi_n(x) + \phi'_n(x)$, where $\hat \psi_n(x)=\psi_n(x)$ if $x > kT/2$
  and $\hat \psi_n(x)=0$ otherwise. It follows that $\hat \psi_n(x) \tilde \psi_m(x)=\phi_m(x)\phi'_n(x)=0$ because these functions are supported in disjoint regions.
  Also, Eqs.~\eqref{eq:HB1} and~\eqref{eq:HB2} imply $\hat \psi_n(x) \phi_m(x) = \nu_1(N)^k$ and $\phi'_n(x) \tilde \psi_m(x)=\nu_1(N)^k$
  so that Eq.~\eqref{eq:sum4} is also of order $\nu_1(N)^k$. Summing over $k\ne 1$ implies that 
  Eq.~\eqref{eq:sum1} can be approximated by $\sbra{\psi_m}   {\psi_n} \rangle =\delta_{n,m}$,
  the term with $k=0$, or by $\sbra{\psi_m^\rD}   {\psi_n^\rD} \rangle$, within precision $\nu_1(N)$. 
  \end{proof}

The previous analysis implies:
\begin{corollary}
\label{cor:orthogonal}
There exists a constant $c>0$  such that,  for all $n,m \le c N$,
the discrete Hermite states are almost orthonormal:
\begin{align}
\left | \sbra{\psi_m^\rD} \psi_n^\rD \rangle - \delta_{n,m} \right | = \nu_1(N) \; .
\end{align}
\end{corollary}
  \begin{proof} It is a direct consequence of Lemma~\ref{lem:D-CVposition} for $l=0$.
  \end{proof}

 \vspace{0.8cm}

  For all $n' \ge 0$ integer, we define the states
  \begin{align}
    \sket{\bar \psi_{n'}^\rD}=\left(\frac{2 \pi}N\right)^{1/4} \sum_{j=-N/2}^{N/2-1} \bar \psi_{n'}(x_j) \ket j \; ,
  \end{align}
  with
  \begin{align}
  \label{eq:CFTeigenstates}
  \bar \psi_{n'}(x_j) = \sum_{k=-\infty}^\infty \psi_{n'}(x_j +kT) \; .
  \end{align}
  Such states will be useful to prove the following lemmas.
   Remarkably, the $\sket{\bar \psi_{n'}^\rD}$ are eigenvectors of the centered Fourier transform with eigenvalues $(-i)^{n'}$.
   To show this, we work in the bra-ket notation and let $\sket{C_j} = \sum_{k} \sket{x_j + kT}$ be the states
   that represent the corresponding Dirac combs (sums of Dirac deltas): $\sbra{C_j} \psi_{n'}\rangle =   \bar \psi_{n'}(x_j)$.
   Then, 
   \begin{align}
     \sket{\bar \psi_{n'}^\rD}=\left(\frac{2 \pi}N\right)^{1/4} \sum_{j=-N/2}^{N/2-1} \ket j \sbra{C_j} \psi_{n'}\rangle \; .
   \end{align}
   The properties of the Fourier transform when acting on the Dirac comb implies
   \begin{align}
    \hat F \sket{C_j}=(1/\sqrt N)\sum_{j'=-N/2}^{N/2-1}
   e^{i 2 \pi j j'/N} \sket{C_{j'}} \; .
   \end{align}
    The centered Fourier transform has a similar action on $\ket j$:
    \begin{align}
    F^\rD_{\rm c} \sket{j}=(1/\sqrt N)\sum_{j'=-N/2}^{N/2-1}
   e^{i 2 \pi j j'/N} \sket{{j'}} \; .
   \end{align}
    Then,
   \begin{align}
   F^\rD_{\rm c}  \sket{\bar \psi_{n'}^\rD}& = \left(\frac{2 \pi}N\right)^{1/4}   \sum_{j=-N/2}^{N/2-1} \ket j \sbra{C_j} \hat F \sket{\psi_{n'}} 
     = (-i)^{n'}  \sket{\bar \psi_{n'}^\rD} \;,
   \end{align}
   where we used $\hat F \sket{\psi_{n'}} =(-i)^{n'}  \sket{\psi_{n'}} $.
   
   We note that
   \begin{align}
  \nonumber
  \| \sket{\psi_{n}^\rD}-\sket{\bar \psi_{n}^\rD} \|^2 &=(2\pi/N)^{1/2} \sum_{j=-N/2}^{N/2-1} |\psi_{n}(x_j)-\bar \psi_{n}(x_j)|^2 \\ &
 =(2\pi/N)^{1/2} \sum_{j=-N/2}^{N/2-1} |\sum_{k \ne 0}\psi_{n}(x_j+kT) |^2 \; .
  \end{align}
  For $k \ne 0$, we obtain $|x_j + kT| \ge T/2$.  Then, if $c$ and $\tilde \beta$
  are the constants in Lemma~\ref{lem:smallsupport} and Lemma~\ref{lem:D-CVposition}, and $n \le cN$,
  Eq.~\eqref{eq:HB1} implies $| \psi_{n}(x_j +kT)| = O(\exp(-x \tilde \beta \sqrt N))$, with $x=|x_j + kT|$.
  Consider, for example, $k \ge 1$. In that case, $|x_j + kT| \ge kT/2$ and explicit calculation
  gives $\sum_{k \ge 1} |\psi_{n}(x_j+kT)| = O(\nu_1(N))$ for sufficiently large $N$.
  The result for $k \le  -1$ is similar. It follows that
  \begin{align}
  \label{eq:convapprox}
   \| \sket{\psi_{n}^\rD}-\sket{\bar \psi_{n}^\rD} \| =  \nu_1(N) \; .
  \end{align}

  \begin{lemma}
  \label{lem:D-CVmomentum}
  Given $l \ge 0$, there exists a constant  $c>0$   such that, for all $n,m \le c N$,   
  \begin{align}
  \left | \sbra{\psi_m^\rD} (p^\rD)^l \sket{\psi_n^\rD} - \sbra{\psi_m } \hat p ^l \sket{\psi_n }  \right| = \nu_1(N)   \; .
  \end{align}
   \end{lemma}
  \begin{proof}
 Because $\sket{\psi_{n}^\rD}$ approximates $\sket{\bar \psi_{n}^\rD}$ for $n \le cN$, where $c$ is the constant of Lemma~\ref{lem:smallsupport},
 we can transform the operators $\hat p$ and $p^\rD$ with the corresponding Fourier transformations, and use Lemma~\ref{lem:D-CVposition}
  to obtain the desired result. Since $\| p^\rD \| =\| x^\rD \| = O(\sqrt N)$,   the properties of the norm and Eq.~\eqref{eq:convapprox} imply
  \begin{align}
 \label{eq:momentumapprox1}
  | \sbra{\psi_m^\rD} (p^\rD)^l \sket{\psi_n^\rD} - \sbra{\bar \psi_m^\rD} (p^\rD)^l \sket{\bar \psi_n^\rD} |=
   \nu_1(N)   \; ,
  \end{align}
  for   $n,m \le c N$.
  Additionally, conjugation by the corresponding Fourier transforms gives
  \begin{align}
  \nonumber
& | \sbra{\bar \psi_m^\rD} (p^\rD)^l \sket{\bar \psi_n^\rD} - \sbra{  \psi_m} \hat p ^l \sket{  \psi_n }  | = \\
  & =| \sbra{\bar \psi_m^\rD} (x^\rD)^l \sket{\bar \psi_n^\rD} - \sbra{  \psi_m} \hat x ^l \sket{  \psi_n }  | \; ,
  \end{align}
  and Lemma~\ref{lem:D-CVposition} implies
  \begin{align}
  \label{eq:momentumapprox2}
   | \sbra{\bar \psi_m^\rD} (p^\rD)^l \sket{\bar \psi_n^\rD} - \sbra{  \psi_m} \hat p ^l \sket{  \psi_n }  | =
   \nu_1(N)   \; .
  \end{align}
  Applying the triangle inequality to Eqs.~\eqref{eq:momentumapprox1} and \eqref{eq:momentumapprox2} gives the desired result.
    \end{proof}

  \begin{lemma}
  \label{lem:D-CVpos-mom}
  Given $l_1,l_2 \ge 0$, there exists a constant   $c>0$   such that, for all $n,m \le c N$,   
  \begin{align}
  \left | \sbra{\psi_m^\rD}  (p^\rD)^{l_1} (x^\rD)^{l_2}  \sket{\psi_n^\rD} - \sbra{\psi_m } \hat p ^{l_1} \hat x ^{l_2} \sket{\psi_n }  \right| = \nu_1(N)   \; .
  \end{align}
   \end{lemma}
  \begin{proof}
  The property of the Hermite functions  $ x \psi_n(x) = \sqrt{(n+1)/2} \psi_{n+1}(x)+ \sqrt{n/2} \psi_{n-1}(x)$
  immediately implies
  \begin{align}
  \nonumber
    x^\rD \sket{\psi^\rD_n} &= (2\pi/N)^{1/4}\sum_{j=-N/2}^{N/2-1} x_j \psi_n(x_j) \sket j \\
    \label{eq:DPosAction}
  & = \sqrt{(n+1)/2} \sket{\psi_{n+1}^\rD}+ \sqrt{n/2} \sket{\psi_{n-1}^\rD} \; ,
  \end{align}
  for all $n \ge 0$. Then, if $x^{l_2} \psi_n(x) = \sum_{l'=-l_2}^{l_2} c_{l'} \psi_{n+l'}(x)$, 
  we obtain 
  \begin{align}
   (x^\rD )^{l_2} \sket{\psi^\rD_n} =\sum_{l'=-l_2}^{l_2} c_{l'}\sket{\psi^\rD_{n+l'}}\; ,
   \end{align}  
  and
  \begin{align}
  \label{eq:discreteposmom}
  \sbra{\psi_m^\rD}  (p^\rD)^{l_1} (x^\rD)^{l_2}  \sket{\psi_n^\rD} =   \sum_{l'=-l_2}^{l_2} c_{l'} \sbra{\psi_m^\rD}  (p^\rD)^{l_1} \sket{\psi^\rD_{n+l'}} \; .
  \end{align}
  Because $c_{l'} = O(N^{l_2/2})$ and $l_1$ and $l_2$ are constants,   Lemma~\ref{lem:D-CVmomentum} implies that   Eq.~\eqref{eq:discreteposmom}
  can be approximated by
  \begin{align}
  \sum_{l'=-l_2}^{l_2} c_{l'} \sbra{\psi_m}  \hat p^{l_1} \sket{\psi_{n+l'}} =  \sbra{\psi_m}  \hat p^{l_1}  \hat x^{l_2} \sket{\psi_{n}}
  \end{align}
  within precision exponentially small in $N$, as long as $n,m \le cN$. 
  The constants $c>0$ is as in Lemma~\ref{lem:smallsupport}. The constant $\beta >0$ (used for the lower bound of $\nu_1(N)$) is as in Lemma~\ref{lem:D-CVmomentum}.  
    \end{proof}
  
  A similar result is obtained if we swap the order of $x^\rD$ and $p^\rD$, and $\hat x$ and $\hat p$.
  This can be shown by acting with the corresponding Fourier transforms.
  
 

  \begin{corollary}
  \label{cor:maincor}
  There exists a constant $c>0$ such that, for all $n \le c N$,
  \begin{align}
\| ( H_\rD -(n+1/2) ) \sket{\psi_n^\rD} \|^2 = \nu_1(N)  \; .
  \end{align}
  \end{corollary}
  \begin{proof}
  Alternatively, we can show that 
  \begin{align}
  \label{eq:normsquared}
  \left |\sbra{\psi_n^\rD} (H_\rD)^2 -2 (n+1/2) H_\rD + (n+1/2)^2 \sket{\psi_n^\rD} \right | 
  \end{align}
  is exponentially small in $N$.  Corollary~\ref{cor:orthogonal} implies that there exists $c>0$  (as in Lemma~\ref{lem:smallsupport}) such that,
  if $n \le cN$,
  \begin{align}
 | \sbra{\psi_n^\rD} (n+1/2)^2 \sket{\psi_n^\rD} - (n+1/2)^2 | = \nu_1(N) \; .
  \end{align}
Also, since $H_\rD=((x^\rD)^2 + (p^\rD)^2)/2$, Lemmas~\ref{lem:D-CVposition} and \ref{lem:D-CVmomentum} imply ($l=2$)
\begin{align}
\nonumber
 | \sbra{\psi_n^\rD} H_\rD \sket{\psi_n^\rD} -  \sbra{\psi_n } H \sket{\psi_n } | &=  | \sbra{\psi_n^\rD} H_\rD \sket{\psi_n^\rD} -  (n+1/2) |  
 \\ & =\nu_1(N) \; ,
\end{align}
 and then
 \begin{align}
| \sbra{\psi_n^\rD} (n+1/2) H_\rD \sket{\psi_n^\rD} - (n+1/2)^2 | = \nu_1(N) \; .
 \end{align} 
 Since $( H_\rD)^2=[ ((x^\rD)^4+ (p^\rD)^4) + (x^\rD)^2(p^\rD)^2 + (p^\rD)^2 (x^\rD)^2]/4$, Lemma~\ref{lem:D-CVpos-mom} ($l=2$) implies
   \begin{align}
| \sbra{\psi_n^\rD} ( H_\rD)^2 \sket{\psi_n^\rD} - (n+1/2)^2 | = \nu_1(N)  \; .
 \end{align} 
  It follows that Eq.~\eqref{eq:normsquared} can be approximated by $(n+1/2)^2-2(n+1/2)^2+(n+1/2)^2=0$
  within precision that is exponentially small in $N$. 
 
  \end{proof}

\section{High-order Trotter-Suzuki formula for the discrete QHO} 
\label{appendix:Trotter-Suzuki}
We  first prove our results for the continuous
variable QHO and then find approximations in the discrete case.
Since $[\hat x, \hat p]=i$, the operators of the QHO form the Lie algebra ${ sp} (2)$ and satisfy the following commutation relations:
  \begin{align}
       [\hat x^2, \hat p^2] =2i\{\hat x, \hat p\} \; , 
       [\hat x^2,\{\hat x, \hat p\}] =4i \hat x^2 \; , 
         \label{eq:commrel}
       [ \hat p^2,\{ \hat x, \hat p\}]  =-4i \hat p^2\; .
       \end{align}
For $s \in {\bf R}$, Eqs.~\eqref{eq:commrel} imply
   \begin{align} 
   \nonumber
   e^{-i s \hat x^2} \hat p^2 e^{i s \hat x^2}& =\hat p^2 + 2 s \{\hat x,\hat p\} + 8 s^2 \hat x^2 \; , \\
   \nonumber
   e^{-i s \hat p^2} \hat x^2 e^{i s \hat p^2}&=\hat x^2 - 2 s \{\hat x,\hat p\} + 8 s^2 \hat p^2 \; , \\
   \nonumber
   e^{-i s \hat x^2} \{\hat x,\hat p\} e^{i s \hat x^2}&=\{\hat x,\hat p\} + 4 s \hat x^2 \; , \\
   e^{-i s \hat p^2} \{\hat x,\hat p\} e^{i s \hat p^2}&=\{\hat x,\hat p\} - 4 s \hat p^2 \; .
      \label{eq:QHOtransformations}
     \end{align}

We let $U(t) =\exp (-i Ht)$ be the evolution operator of the QHO for time $t \in {\bf R}$. The second order Trotter-Suzuki (symmetric) approximation  over a course of evolution time $s$ is
 \begin{align} \label{eq:U1}
 U_1(s)=e^{-i s \hat x^2/4}  e^{-i s \hat p^2/2}e^{-i s \hat x^2/4} \; .
 \end{align}
While such an approximation is typically defined in a finite dimensional Hilbert space, 
here we use it in Hilbert spaces of infinite dimension. We also construct  higher order 
Trotter-Suzuki approximations using the recurrence relation
\begin{align}\label{eq:Upplus}
U_{p+1}(s)=\left(U_p(s_p)\right)^2  U_p(s-4s_p) \left(U_p(s_p)\right)^2 \; ,
\end{align}
with $s_p=s/(4-4^{1/(2p+1)})$ and $p=2,3,\ldots$~\cite{suzuki_qmc_1998,childs_thesis,berry_efficient_2007}. The operator $\epsilon_p (s) =U_p(s) U(-s) - \one $,
where $\one$ is the identity operation,
can be used to quantify the error made in the approximation.  
\begin{lemma}
\label{lemma:TSA}
There exists a constant $d \ge 1$ such that, for all $s$ satisfying $|s|<1/d$ and all $n\ge 0 $, $p \ge 1$,  
\begin{align}
\| \epsilon_p(s) \ket{\psi_n} \| = O((n+2) |s|^{2p+1}) \; .
\end{align}
\end{lemma}


\begin{proof}
With no loss of generality, we write
\begin{align}
\epsilon_p(s)&= \int_0^s ds' \; \partial_{s'}  \epsilon_p(s')) \;.
\end{align}
From the definition of $\epsilon_p(s)$, we obtain
\begin{align}
\label{eq:ep}
 \partial_{s'}  \epsilon_p(s')= U_p(s')  \hat f_p(s') U(-s') \; ,
\end{align}
where $\hat f_p$ is an operator that depends on the approximation order $p$.
Furthermore, we can use Eqs.~\eqref{eq:QHOtransformations} to show  
\begin{align}
\hat f_p(s) = a_{l(p)}(s) \hat x^2 + b_{l(p)}(s) \hat p^2 + c_{l(p)}(s) \{ \hat x, \hat p \} \; ,
\end{align}
where $a_{l(p)}(s)$, $b_{l(p)}(s)$, and $c_{l(p)}(s)$  are polynomials in $s$ of lowest degree ${l(p)}$, and ${l(p)}$ is a positive integer that depends on $p$.
For example, if $p=1$, explicit calculation of Eq.~\eqref{eq:ep} 
results in
\begin{align}
\nonumber
 & \hat f_{1}(s)  = e^{is \hat x^2/4}e^{is \hat p^2/2} (-i \hat x^2/4)e^{-is \hat p^2/2}e^{-is \hat x^2/4} + \\
 &+
e^{is \hat x^2/4} (-i \hat p^2/2) e^{-is \hat x^2/4} -i \hat x^2/4 +i H \; ,
\end{align}
with $H= (\hat x^2 + \hat p^2)/2$.
Equations~\eqref{eq:QHOtransformations} 
imply $\hat f_{1}(s) = -i (s^4/4) \hat x^2 -i
(s^2/2) \hat p^2 -i (s^3/4) \{ \hat x, \hat p \}$ and then $l(p=1)=2$.
The first goal is to obtain upper bounds on $|a_{l(p)}(s)|$, $|b_{l(p)}(s)|$, and $|c_{l(p)}(s)|$; we will obtain
such bounds from the corresponding series expansions in $s$.

In general, we will show by induction that $l(p)=2p$. This result also follows from~\cite{suzuki_90}.
Since $U_p(s)=(\one+ \epsilon_p(s))U(s)$, we can rewrite 
Eq.~(\ref{eq:Upplus})
as
\begin{align}
\nonumber
U_{p+1}(s) = &((\one+ \epsilon_p(s_p))U(s_p ))^2 (\one+ \epsilon_p(s-4s_p))    \\
& \ \ U(s-4s_p ) ((\one+ \epsilon_p(s_p))U(s_p ))^2 \; .
\end{align}
This is a sum of $2^5$ terms; the term without $\epsilon_p$ corresponds to $(U(s_p))^2 U(s-4s_p)(U(s_p))^2 =U(s)$, and the 
remaining terms sum up to $U_{p+1}(s) - U(s)$. The sum of the terms containing a single $\epsilon_p$
is 
\begin{align}
\nonumber
 E_p(s) = & \epsilon_p (s_p) U(s) + U (s_p)\epsilon_p (s_p ) U(s-s_p) + \\
 \nonumber & + U (2s_p) 
\epsilon_p (s-4s_p)U (s-2s_p) + \\ 
\nonumber
& +U (s-2s_p)\epsilon_p (s_p)U (2 s_p)+ \\
&+  U (s-s_p)\epsilon_p (s_p) U(s_p)\; .
\end{align}

In the induction step we assume that $l(p)=2p$, for some $p \ge 1$. 
A Taylor series expansion of $\epsilon_p(s)$ can be obtained
by Taylor expanding $U_p(s')$ and $U(-s')$ in Eq.~\eqref{eq:ep}.
Because the lowest degree of $\hat f_p(s')$ is assumed to be $2p$,
integration in $s'$ implies that the lowest degree in the series expansion of
$\epsilon_p(s)$ will be $2p+1$. The lowest degree in the Taylor series
of $E_p(s)$ will be determined by the lowest degree of $4 \epsilon_p(s_p) + \epsilon_p (s-4s_p)$,
which is the operator obtained from $E_p(s)$ if we replace $U$ by $\one$ (i.e., the lowest degree term
in the expansion of $U$).
Then, since $4(s_p)^{2p+1} + (s-4s_p)^{2p+1}=0$,
the lowest degree in the series of $E_p(s)$ is, at least, $2p+2$.
Also, the lowest degree in the  series  of $U_{p+1}(s) - U(s)$  (and $\epsilon_{p+1}(s)$) is determined
by that of $E_p(s)$ (i.e., the term of lowest order in $\epsilon_p$), and   is also bounded from below by $2p+2$.
It follows that $U_{p+1}(s) -U(s)= \hat V s^{2p+2} + O(s^{2p+3})$, where $\hat V$ is some operator
that depends on (powers of) $\hat x^2$, $\hat p^2$ and $\{ \hat x, \hat p \}$.
The Trotter-Suzuki approximations determined by Eq.~\eqref{eq:Upplus} are symmetric and imply
\begin{align}
\nonumber
&\one = U_{p+1}(s) U_{p+1}(-s)  \\
\nonumber
&= (U(s) + \hat V s^{2p+2} + O(s^{2p+3})) \times \\ \nonumber & \ \ \ \times
 (U(-s) + \hat V s^{2p+2} + O(s^{2p+3})) \\
\nonumber
& = \one + U(s) \hat V s^{2p+2} +\hat V s^{2p+2} U(-s) + O(s^{2p+3})   \\
& = \one +  2 \hat V s^{2p+2} + O(s^{2p+3}) \; ,
\end{align}
which can only be satisfied if $\hat V=0$. The last equality follows 
from the expansion of $U(s)$, whose lowest-degree term is $\one$.
Then, the lowest degree in the Taylor series of $U_{p+1}(s) -U(s)$, or $\epsilon_{p+1}(s)$, is 
$2p+3 = 2(p+1)+1$, implying that $l_{p+1} = 2(p+1)$, and proving the induction step.

From Eq.~\eqref{eq:Upplus}, $U_p(s)$ is a product of $O(5^p)$ exponentials of $\hat x^2$ and $\hat p^2$.
The relations in Eqs.~\eqref{eq:QHOtransformations} and Eq.~\eqref{eq:ep} imply that the highest degree in $\hat f_p$
is smaller than $2 \times 5^p$. Since we already proved that $l(p)=2p$, we obtain
\begin{align}
a_{l(p)}(s) = \sum_{l=2p}^{2 \times 5^p} u_l \; s^l \; ,
\end{align}
and there is a similar expression for $b_{l(p)}(s)$ and $c_{l(p)}(s)$. In order to show that the dominant term
in $\epsilon_p(s)$ is that of degree $2p+1$, as in the lemma, we still need to show that the coefficients
$u_l$ are bounded. In Eqs.~\eqref{eq:QHOtransformations} we showed that the corresponding unitary 
transformations of $\hat x^2$, $\hat p^2$ and $\{ \hat x,
\hat p \}$,   needed to obtain $\hat f_p$, are linear combinations of the same operators, and the largest prefactor
in such combinations is a constant (8 in this case). The term of degree $l$ in $\hat f_p$ is obtained from $O(l)$
unitary transformations of the operators. Since $|s_p| \le |s|$ and $|s-4s_p| \le |s|$,   there exists a constant
$d>0$ such that $|u_l| \le d^l$. The value of $d$ can be determined from Eqs.~\eqref{eq:QHOtransformations}.
If $s$ is such that $|s| < 1/d$, the series for $a_{l(p)}(s)$ is convergent and $|a_{l(p)}(s)| = O((|s|/d)^{2p})$.
That is, $|a_{l(p)}(s)| = O(|s|^{2p})$, and a similar result can be obtained for $b_{l(p)}(s)$ and $c_{l(p)}(s)$.

The properties of the Hermite functions imply  $\| \hat x^2 \sket {\psi_n} \|= \| \hat p^2 \sket {\psi_n} \|=O(n+2)$,
and also $\| \{ \hat x , \hat p \}  \sket {\psi_n} \|= O(n+2)$. From the triangle inequality and   the previous results, 
we obtain
$\| \hat f_p(s')  \sket {\psi_n} \|= O((n+2) |s'|^{2p})$. 
In addition, since $U(-s')  \sket {\psi_n} =e^{is'(n+1/2)}  \sket {\psi_n}$ 
and $\|U_p(s')\|=1$,   Eq.~\eqref{eq:ep} implies
\begin{align}
\| \epsilon_p (s) \sket {\psi_n} \| & \le \int_{0}^{|s|} ds' \;  \| \hat f_p(s')  \sket {\psi_n} \|  
 =O  (  (n+2) |s|^{2p+1} )\; ,
\end{align}
which is the desired result.

\end{proof}

Lemma~\ref{lemma:TSA} basically demonstrates that $U_p(s)$ is a  product formula
approximation of $U(s)$ of order $2p+1$ in $s$.
The approximation is better for smaller values of $n$, i.e., for the low energy states. It is important to note that 
the dependence of the approximation error in $n$ is only linear.

\begin{lemma} \label{lemma:TSA2}
 The number of exponentials of $\hat x^2$ and $\hat p^2$ needed to 
 prepare $U(t) \sket{\psi_n}$, for $|t| \ge 1$, and within precision $\epsilon >0$, is
 \begin{align}
 {\cal M} = \Theta \left (|t| \exp(\gamma \sqrt{\log ((n+2) |t| /\epsilon)}) \right) \; ,
 \end{align}
 where $\gamma >0$ is a constant.
\end{lemma}
\begin{proof}
We first divide $t$ into $k$ intervals of size $s =t/k$, i.e., $U(t) = (U(s))^k$.   We will
approximate each $U(s)$ by $U_p(s)$, for some $p \ge 1$, and then choose $p$ and $k$
that minimize the number of exponentials in the product $(U_p(s))^k$ to obtain $\cal M$.
If $|s| < 1/d$, for some constant $d>0$, subadditivity of errors and Lemma~\ref{lemma:TSA} imply
\begin{align}
\label{eq:TSAerror}
k ((n+2) |s|^{2p+1}) = O(\epsilon) \; .
\end{align}
Each $U_p(s)$ can be implemented with less than $5^p$ exponentials of $\hat x^2$ and $\hat p^2$.
Then, the number of exponentials in  $(U_p(s))^k$
is ${\cal M} \le k 5^p$.
Equation~\eqref{eq:TSAerror} implies that there exists
\begin{align}
k = \Theta \left( \frac {(n+2)^{1/2p} |t|^{1+1/2p}} {\epsilon^{1/2p}}\right)
\end{align}
to satisfy the desired error bound. This implies that the number of exponentials is, for any $p=1,2,\ldots$,
\begin{align}
{\cal M }(p)= \Theta \left( \frac {5^p (n+2)^{1/2p} |t|^{1+1/2p}} {\epsilon^{1/2p}}\right) \;.
\end{align}
It is simple to obtain the optimal value of $p$, defined by $p = \arg \min_{p' \ge 1} {\cal M}(p')$. The result is
\begin{align}
 &\lceil  \sqrt{\log((n+2)|t|/\epsilon) /  (2\log 5)} \rceil \ge p \; ,\\
 & p \ge \lfloor  \sqrt{\log((n+2)|t|/\epsilon) /  (2\log 5)} \rfloor \; ,
\end{align}
implying that there is a constant $\gamma >0$ such that
\begin{align}
{\cal M} 
 = \Theta \left ( |t| \exp(\gamma \sqrt{\log((n+2)|t|/\epsilon)}) \right) \; .
\end{align}
The constant $\gamma$ can be obtained from the value of $p$
and is approximately $\sqrt{2 \log 5}$. The idea of computing an optimal value of $p$
was also considered in~\cite{berry_efficient_2007}.

Note that $|s|$ decreases with $|t|$ so the assumption $|s| < 1/d$ is valid, with no loss
of generality. In particular, we can always reduce $s$
by a constant factor, at a constant increase in the cost, without changing the total order of operations. For example, we 
can assume that $|s| 5^p \le \tilde c$ for any constant $\tilde c=O(1)$, and still satisfy
${\cal M}= \Theta \left ( |t| \exp(\gamma \sqrt{\log((n+2)|t|/\epsilon)}) \right)$. 

\end{proof}

Three remarks are in order. First, so far we assumed $\hbar \rightarrow 1$ so we can disregard units.
If the Hamiltonian is $H = \hbar \omega (\hat x^2 + \hat p^2)/2$, and units are considered,
then the number of exponentials is ${\cal M}= \Theta \left ( \omega |t| \exp(\gamma \sqrt{\log((n+2)\omega |t|/\epsilon)}) \right)$.
We also note that $|s| 5^p =\Theta(1)$ in Lemma~\ref{lemma:TSA2}.  
Also, if $|t|=O(1)$ and $\epsilon=O(1)$, the number of exponentials is ${\cal M}=\Theta(\exp(\tilde \gamma \sqrt{\log( n+2)}))$,
for some constant $\tilde \gamma >0$.
This implies $\lim_{n\rightarrow \infty} {\cal M}/n^\alpha =0$  and $\lim_{n\rightarrow \infty} (\log n)^\alpha/{\cal M}=0$, for all $\alpha >0$.
Finally, if $s$ and $p$ are as in Lemma~\ref{lemma:TSA2}, then $\| [(U_p(s))^k - U(t)] \sket{\psi_{m}} \| = O(\epsilon)$ for all $m \le n$.


In the discrete case, we define 
the $N \times N$ unitary matrices $U^{\rD}_p(s)$ by replacing $\hat x \rightarrow x^\rD$ and 
$\hat p \rightarrow p^\rD$ in the definition of $U_1(s)$ and $U_p(s)$ -- see Eqs.~\eqref{eq:U1} and~\eqref{eq:Upplus}. 
The next goal is to use the previous results
for continuous variables, and show that $U^\rD(t) \sket{\psi_n^\rD}$ approximates $(U^\rD_p(s))^k \sket{\psi^\rD_n}$,
within precision $O(\epsilon)$, for certain values of $n$.
Because $U^\rD_p(s)$ is unitary, it will suffice to show $\sbra{\psi_n^\rD}(U^\rD_p(s))^k \sket{\psi^\rD_n}$ is sufficiently close to
 $\sbra{\psi_n}(U_p(s))^k \sket{\psi_n}$ as these are close to 1 in absolute value.

\begin{lemma}
\label{lem:TSA10}
Let $s$ and $p$ satisfy $|s| 5^p =\Theta(1)$. Then, there exists a constant   $c'>0$ such that, for all $n \le c' N$,
\begin{align}
 | \sbra{\psi_n^\rD} U_p^\rD(s) \sket{\psi_n^\rD} - \sbra{\psi_n} U_p(s) \sket{\psi_n} |& =  O(5^p (\nu_1(N))^{|s|})  \; .
\end{align}
\end{lemma}
\begin{proof}
By definition, $U_p^\rD(s)$ is a product of $ r \le 5^p$ exponentials of $(x^\rD)^2$ and $(p^\rD)^2$:
\begin{align}
\label{eq:discreteTSA}
U_p^\rD(s) = e^{-i t_1 (x^\rD)^2}e^{-i t_2 (p^\rD)^2} \ldots e^{-i t_r (x^\rD)^2} \; .
\end{align}
The corresponding evolution times satisfy $|t_i| < |s|$. We first approximate each exponential
by truncating the Taylor series at order $l$, where $l$ will be chosen below. Since $\|(x^\rD)^2 \| = \|(p^\rD)^2 \| =O(N)$,
the subadditivity of errors implies that the overall error in approximating all the exponentials in Eq.~\eqref{eq:discreteTSA} is $O(5^p (N|s|)^l/l!)$,
assuming that $l \ge N|s|$.
We then choose, for example, $l=\lceil 2e |s| N \rceil$ and Stirling's approximation implies
\begin{align}
O(5^p (N|s|)^l/l!) = O(5^p (1/2)^{2e |s| N})  =O(5^p (\nu_1(N))^{|s|})  \; .
\end{align}
 The property $x \psi_n(x) = \sqrt{n/2} \; \psi_{n-1}(x) + \sqrt{(n+1)/2} \; \psi_{n+1}(x)$ implies 
 \begin{align}
 x^\rD \sket{\psi_n^\rD} = \sqrt{n/2} \sket{\psi_{n-1}^\rD} + \sqrt{(n+1)/2} \sket{\psi_{n+1}^\rD}
 \end{align}
 so that 
 the approximation of $e^{-i t_j  (x^\rD)^2}$, when acting on $\sket{\psi_n^\rD}$, gives
 $ \sum_{n'=0}^{n+2l} c_{n'}  \sket{\psi_{n'}^\rD}$. The coefficients $c_{n'}$ can be obtained
 from the continuous-variable case; that is,
 \begin{align}
 \left( \sum_{l'=0}^l (-i t_j \hat x^2)^{l'}/l'! \right)  \sket{\psi_{n}}= \sum_{n'=0}^{n+2l} c_{n'}  \sket{\psi_{n'}} \; .
 \end{align}
 
  If $n+2l \le cN$, for some constant $c>0$ determined in Lemma~\ref{lem:smallsupport} and Lemma~\ref{lem:D-CVposition},
 we can  approximate $\sket{\psi_{n'}^\rD}$ by $ \sket{\bar \psi_{n'}^\rD}$ within precision $\nu_1(N)$ --
 see Eq.~\eqref{eq:convapprox}.  
 Since $e^{-i t_j  (p^\rD)^2} =  F^\rD_{\rm c} e^{-i t_j  (x^\rD)^2} (F^\rD_{\rm c})^\dagger$, a truncated series
 approximation of $e^{-i t_j  (p^\rD)^2}$, when acting on $\sket{\psi_n^\rD}$, gives
 $\sum_{n'=0}^{n+2l} d_{n'}  \sket{\psi_{n'}^\rD} +\nu_1(N)$. The coefficients $d_{n'}$
 can also be obtained from the corresponding continuous-variable case:
  \begin{align}
 \left( \sum_{l'=0}^l (-i t_j \hat p^2)^{l'}/l'! \right)  \sket{\psi_{n}}= \sum_{n'=0}^{n+2l} d_{n'}  \sket{\psi_{n'}} \; .
 \end{align}
 
 Because we approximate a product of $O(5^p)$ exponentials,
 the above approximations are valid as long as $n + 5^p 2 l \le cN$. Then, the  assumptions 
 of working within the``low-energy" subspace   apply in this analysis. Equivalently, 
 we can assume that $n \le c'N$, for some constant $c'<c$. The result is
 \begin{align}
 \nonumber
 U_p^\rD(s)   \sket{\psi_n^\rD} = \sum_{n'=0}^{cN} x_{n'}  \sket{\psi_{n'}^\rD} +  
 O(5^p \nu_1(N)) + \\ + O(5^p( \nu_1(N))^{|s|}) \;.
 \end{align}
 Because $|s|=O(1)$ and $\nu_1(N) <1$, the dominant order in the approximation is $O(5^p (\nu_1(N))^{|s|})$.
 The coefficients $x_{n'}$ can be obtained from 
 \begin{align}
 \nonumber
&( \sum_{l_1=0}^l  (-i t_1 \hat x^2)^{l_1}/l_1 ! ) ( \sum_{l_2=0}^l  (-i t_2 \hat p^2)^{l_2}/l_2 ! ) \ldots \\
 & \ldots ( \sum_{l_r=0}^l  (-i t_r \hat x^2)^{l_r}/l_r ! ) \sket{\psi_n} = \sum_{n'=0}^{cN} x_{n'}  \sket{\psi_n'} \; .
 \end{align}

Since  $n+5^p (2l) \le cN$, approximating
$U_p(s)$ by truncating the Taylor series of each exponential of $\hat x^2$ and $\hat p^2$ at order $l$ implies
\begin{align}
\nonumber
U_p(s) \sket{\psi_n} & = \sum_{n'=0}^{cN} x_{n'}  \sket{\psi_n'} + O(5^p (N|s|)^l/l!) \\ &
= \sum_{n'=0}^{cN} x_{n'}  \sket{\psi_n'} + O(5^p (\nu_1(N))^{|s|}) \; .
\end{align}
 Then,
 \begin{align}
 \nonumber
 x_n & = \sbra{\psi_n}U_p(s) \sket{\psi_n}+ O(5^p (\nu_1(N))^{|s|}) \\ &
 = \sbra{\psi_n^\rD} U_p^\rD(s)   \sket{\psi_n^\rD} +  O(5^p (\nu_1(N))^{|s|})\; ,
 \end{align}
 where we used $n \le c'N \le cN$ and $|s| \le 1$, $5^p \ge 1$. The result is
 \begin{align}
 \label{eq:Sevolapprox}
 \sbra{\psi_n}U_p(s) \sket{\psi_n} = \sbra{\psi_n^\rD} U_p^\rD(s)   \sket{\psi_n^\rD} + O(5^p (\nu_1(N))^{|s|}) \; .
 \end{align}
\end{proof}

For the following results, we will assume that $n \le N'$ for some $N' \ge 1$.
We will set $s$ and $p$ as given by Lemma~\ref{lemma:TSA2}, that is 
\begin{align}
\label{eq:spoptimal2}
p & = \Theta \left(\sqrt{\log(N' |t|/\epsilon) } \right) \; ,
\end{align}
and $5^p |s|=\tilde c$, for some constant $\tilde c>0$. 

\begin{corollary}
\label{cor:TSA}
Let $p$ be as in Eq.~\eqref{eq:spoptimal2}, $N' \ge 1$, $|t| \ge 1$, and $\epsilon >0$. Then, 
there exist constants $c'>0$, $\delta >0$, and dimension $N=\lceil \exp(\delta \sqrt{\log (N' |t|/\epsilon)}) +N'/c \rceil$ such that, for all $n \le N'$,   
 \begin{align}
 | \sbra{\psi_n^\rD} U_p^\rD(s) \sket{\psi_n^\rD} - \sbra{\psi_n} U_p(s) \sket{\psi_n} |& =  O (\epsilon |s/t|) \; .
 \end{align}
\end{corollary}
\begin{proof}
Note that $n \le N' \le c'N$, so that Eq.~\eqref{eq:Sevolapprox} is valid. We then consider the term
$O(5^p (\nu_1(N))^{ |s|}))$ in Lemma~\ref{lem:TSA10}.
To make this term $O(\epsilon |s/ t|)$, it suffices to satisfy
\begin{align}
5^p e^{-\beta N |s|} = O(\epsilon |s/t|) 
 = O(\epsilon 5^{-p}/|t|) \;,
\end{align}
for some constant $\beta >0$ since $\nu_1(N)$ is exponentially small in $N$. This implies $5^{2p} e^{-\beta N |s|} = O(\epsilon/|t|)$ or, equivalently,
\begin{align}
\beta N |s| - 2p \log 5 = \Omega (\log(|t|/\epsilon)) \;.
\end{align}
Since $p^2 = \Theta (\log(N' |t|/\epsilon))$, a choice of $N$ that satisfies $\beta N |s| - 2p \log 5 = \Omega (p^2)$ also works. 
This can be satisfied if $N= O(p^2/|s|)$, which is exponentially large in $p$. Then, there exists a constant $\delta >0$ such that, if $N \le \exp(\delta \sqrt{\log(N'|t|/\epsilon)})$,
the desired error bound is obtained if, in addition, $N ' \le c ' N$.
%
\end{proof}

We can use the previous Lemmas to prove one of our main results:

\begin{lemma}
\label{lem:TSA}
Let  $p$ be as in Eq.~\eqref{eq:spoptimal2}, $N' \ge 1$, $|t| \ge 1$, and $\epsilon >0$.
 Then, there exist constants $c'>0$, $\delta >0$,   and dimension $N= \lceil \exp(\delta \sqrt{\log (N' |t|/\epsilon)}) +N'/c' \rceil$, such that
\begin{align}
\| [(U_p^\rD(s))^k - U^\rD(t) ]\sket {\psi_n^\rD} \| = O(\epsilon) \; ,
\end{align}
for all $n \le N'$, where  $k=t/s$. The unitary $(U_p^\rD(s))^k$
is a product of ${\cal M}= O(|t| 5^{2p})$ exponentials of $(x^\rD)^2$ and $(p^\rD)^2$.  
\end{lemma}

\begin{proof}
 $p$ and $s$ satisfy $5^p |s| = \tilde c$, for some constant $\tilde c >0$ --
see Eq.~\eqref{eq:spoptimal2}. The subadditivity of errors, Lemma~\ref{lemma:TSA} and Corollary~\ref{cor:TSA} imply
\begin{align}
\nonumber
| \sbra{\psi_n^\rD} U_p^\rD(s) \sket{\psi_n^\rD} - e^{-i(n+1/2)s}| & = O(\epsilon|s/t|)+ O(N' |s|^{2p+1}) \\
& = O(\epsilon |s/t|) \; ,
\end{align}
where we used Eq.~\eqref{eq:TSAerror}. This is valid if $n \le N' \le c'N$, for some constant $c'>0$,
and our choice of $N$ already satisfies such a condition.
Then, using Corollary~\ref{cor:orthogonal}, we obtain
\begin{align}
\| U_p^\rD(s) \sket{\psi_n^\rD} - e^{-i(n+1/2)s} \sket{\psi_n^\rD}\| =  O(\epsilon|s/t|) + \nu_1(N) \; .
\end{align}
 In particular, there exists a constant $\delta >0$ such that,  if $N \ge \exp(\delta \sqrt{\log (N' |t|/\epsilon)})$, then 
$\nu_1(N)$ can be made negligible with respect to $(\epsilon|s/t|)$.

Using again the subadditivity property of errors, we obtain
\begin{align}
\| (U_p^\rD(s))^k \sket{\psi_n^\rD} - e^{-i(n+1/2)t} \sket{\psi_n^\rD}\| =  O(\epsilon ) \;,
\end{align}
for the corresponding choices of $N$ and $p$.
Also, Corollary~\ref{cor:maincor} implies, for $n \le N' \le c'N$,
\begin{align}
\|  (U^\rD(t) -  e^{-i(n+1/2)t})\sket{\psi_n^\rD} \| = O(|t| \nu_1(N) N') \; .
\end{align}
 Then,  there exists a constant $\delta >0$
such that, if $N \ge \exp(\delta \sqrt{\log (N' |t|/\epsilon)})$, then $O(|t| \nu_1(N) N) = O(\epsilon)$.
The triangle inequality gives
\begin{align}
\|  [( U_p^\rD(s) )^k - U^\rD(t)] \sket{\psi_n^\rD} \| = O(\epsilon)
\end{align}
for the corresponding choices of $p$ and $N$, which determines the first  result of the Theorem.

By definition, $U^\rD_p(s)$ contains less than $5^p$ exponentials of $(x^\rD)^2$ and $(p^\rD)^2$.
The total number of exponentials in $(U^\rD_p(s))^k$ is bounded by $|t/s| 5^p = O(|t| 5^{2p})$, for our
choices of $s$ and $p$. That is, there exists a constant $\gamma >0$ such that the total number
of exponentials is ${\cal M}=O(|t| \exp (\gamma \sqrt{\log (N' |t|/\epsilon)}))$. 

\end{proof}

\section{Error bounds for the preparation of quantum states with Gaussian amplitudes}
\label{appendix:gaussians}
In this section we prove some results regarding the preparation
of quantum states with approximate Gaussian-like amplitudes that serve as a basis
for preparing  approximations of other eigenstates of the discrete QHO.
The results in this section may be of independent interest.

  \begin{lemma}
  \label{lemma:GSpreparation}
 Given $N$ and $\delta>0$, let 
 \begin{align}
 \sket{\varphi^\rD} = (1/\sqrt \kappa) \sum_{j=-N/2}^{N/2}  \exp(- j^2/(2\delta)) \ket j \; ,
 \end{align}
with $\kappa = \sum_{j=-N/2}^{N/2}  \exp(- j^2/\delta)$, and let $\sigma^2=\delta \pi/N$, $t=\sqrt{\sigma^2(2-4\sigma^2)}/2$,
 and $t'=1/(4t+4\sigma^4/t)$. Then, there exists a constant $\lambda>0$ such that
 \begin{align}
 \label{eq:gaussianerror}
 \|\sket{\psi_0^\rD} - e^{i \alpha(t)} e^{i (x^\rD)^2 t'} e^{i (p^\rD)^2 t} \sket{\varphi^\rD}\|   = O(\exp(- \lambda \delta)) \; .
 \end{align}
 The global phase $\alpha(t) \in {\bf R}$ only depends on $t$.
 \end{lemma}
 \begin{proof}
 First we show that in CVs, Eq.~\eqref{eq:gaussianerror} holds exactly if the discrete states and operators are replaced
 by their CV versions. Then, we approximate the integrals appearing in CVs by finite sums
 that appear in the actual Eq.~\eqref{eq:gaussianerror}.
 
 In CVs, the wave function of a free particle evolves according to the Schr\"odinger equation,
 with a Hamiltonian $-\hat p^2$.
 If the initial state ($t=0$) of the CV system is $\ket \varphi$ and the normalized wave function
is
\begin{align}
\varphi(x,0)= \bra x \varphi \rangle = \frac 1 { (2 \pi \sigma^2)^{1/4} }\exp(-x^2/(4\sigma^2)) \; ,
\end{align} 
then, the wave function at time $t \ge 0$ is given by
\begin{align}
\label{eq:gaussianevolved}
\varphi(x,t) =\left( \frac {2\sigma^2}\pi \right)^{1/4} \frac 1 {\sqrt{-i  2t + 2 \sigma^2} } \exp (-x^2/(-i4t + 4 \sigma^2)) \; .
\end{align}
We note that $|\varphi(x,0)|^2$ and  $|\varphi(x,t)|^2$ correspond to normal distributions
of variances $\sigma^2$ and $\sigma^2 +t^2/ \sigma^2$, respectively.    
We can rewrite $\varphi(x,t) $ as
\begin{align}
    \frac { \gamma(x,t)}  {(2 \pi (\sigma^2+t^2/\sigma^2))^{1/4}}\exp(-x^2/(4 \sigma^2+ 4 t^2/\sigma^2)) \; ,
\end{align}
where the phase factor satisfies
\begin{align}
\gamma(x,t)= e^{i \alpha(t)} e^{-i x^2/(4t+4 \sigma^4/t)} \; ,
\end{align}
 and $\alpha(t)$ solely depends on $t$, and can be computed exactly from Eq.~\eqref{eq:gaussianevolved}.
 Selecting $\sigma^2= \delta \pi/N$, for some given $N>0$, and $t$ so that $4 \sigma^2+4 t^2/\sigma^2=2$,
 we obtain
 \begin{align}
\bra x e^{-i \alpha(t)}e^{i \hat x^2 t '} e^{i \hat p^2 t}\ket{\varphi} =\langle x \sket{\psi_0}= \psi_0(x) \; ,
 \end{align}
with $t'= 1/(4t + 4 \sigma^4/t)$.
Then, in CVs, we can exactly prepare the ground state of the CV QHO from the initial state $\sket{\varphi}$
by applying the unitary sequence $e^{-i \alpha(t)}e^{i \hat x^2 t '} e^{i \hat p^2 t}$.
  We will show that a similar result is obtained, up to some
  bounded approximation errors, when we work in the discrete
  Hilbert space of dimension $N$.

 We rewrite
 \begin{align}
 \label{eq:innerproducthg}
 e^{i\alpha(t)}=\sbra{\psi_0}e^{i \hat x^2 t '} e^{i \hat p^2 t}\ket{\varphi} =  \sbra{\psi_0}e^{i \hat x^2 t '} \hat F e^{i \hat x^2 t} \hat F \ket{\varphi} \; ,
 \end{align}
 where $\hat F$ is the operator that implements the Fourier transform, and we used the property $\hat F^2 e^{-i \hat x^2 t '}  \sket{\psi_0}
 = e^{-i \hat x^2 t '}  \sket{\psi_0}$.
  We also
  define the functions
 \begin{align}
 g(x)& = \bra x \hat F \sket{\varphi}  = (2 \sigma^2/ \pi)^{1/4} \exp(-x^2 \sigma^2) \; , \\
 h(x) &=\sbra{\psi_0}e^{i \hat x^2 t '} \hat F \ket x = \frac {\exp(-x^2/(2-i4t'))} {\pi^{1/4} \sqrt{1-i2t'}}  \; , 
  \end{align}
  where we used $\psi_0(x)=e^{-x^2/2}/\pi^{1/4}$.
  Then, Eq.~\eqref{eq:innerproducthg} becomes
  \begin{align}
  \label{eq:innerproductgaussian}
    \int dx \; h(x) e^{ix^2 t} g(x) 
   = \frac{(2\sigma^2)^{1/4}} {\sqrt{\pi(1-i2t')}}\int dx \exp(-\alpha x^2) \;,
  \end{align}
  and 
  \begin{align}
  \nonumber
  \alpha & = \sigma^2 - it + 1/(2-i 4t') \\ &
  =[ \sigma^2 + 1/(2+8t'^2)] -i [t -t'/(1+4t'^2)] \; .
  \end{align}
 Next we approximate Eq.~\eqref{eq:innerproductgaussian} by a finite sum,
 assuming the discretization where $x_j = j \sqrt{2 \pi/N}$, $j = -N/2, \ldots, N/2-1$,
 and bound the errors of the approximation. Later, we relate this approximation with
 Eq.~\eqref{eq:gaussianerror}.
 
 We will show that one of the dominant sources of error in the approximation, for $N \gg 1$,
 results from bounding $j$ so that $|j| \le N/2$.
 To estimate this approximation error, we consider
 \begin{align}
 \varepsilon_1= \left| \frac{(2\sigma^2)^{1/4}} {\sqrt{\pi(1-i2t')}} \sqrt{\frac{2 \pi} {N}}\sum_{j = N/2}^\infty e^{-\alpha x_j^2} \right|\; .
 \end{align}
 In particular,
 \begin{align}
 | \sum_{j = N/2}^\infty e^{-\alpha x_j^2}| & \le  \sum_{j = N/2}^\infty e^{-\Re(\alpha)\pi j} 
  \le \frac{e^{-\Re(\alpha)\pi N/2}}{1 -e^{-\Re(\alpha)\pi} } \; ,
 \end{align}
 where $\Re (\alpha)$ is the real (and positive) part of $\alpha$,
 and we used $x_j \ge x_{N/2}=\sqrt{\pi N/2}$.
 Then, since $\Re (\alpha) \ge \sigma^2 = \delta \pi/N$,
 \begin{align}
 \label{eq:expsumapp}
 | \sum_{j = N/2}^\infty e^{-\alpha x_j^2}| & \le 2 \frac{e^{-\delta \pi^2 /2}}{(\delta \pi^2/N) } \; ,
 \end{align}
 where we used that $1-e^{-x} \ge x/2$ for $x \le 1$, which in this case requires $N \ge \delta \pi^2$.
 We also note
 \begin{align}
 \nonumber
  \left| \frac{(2\sigma^2)^{1/4} \sqrt 2} {\sqrt{(1-i2t') N}} \right |  & \le \frac{(2 \sigma^2)^{1/4}}{\sqrt{\pi t' N}} \\ \nonumber &
  \le \frac{(2 \sigma^2)^{1/4}}{\sqrt{\pi N }} \sqrt{4t} \\  &
  \le  2 (\sigma^2/(\pi N))^{1/2}    \;.
 \end{align}
 Combining this with Eq.~\eqref{eq:expsumapp}, we obtain
 \begin{align}
 \varepsilon_1 & \le e^{- \pi^2 \delta/2}/\sqrt \delta  = O(\exp(-\lambda^{(1)} \delta)) \; ,
 \end{align}
 where we replaced $\sigma^2 = \pi \delta/N$ and $\lambda^{(1)} >0$ is a constant.
 
%
 We consider now the approximation of Eq.~\eqref{eq:innerproductgaussian} 
 by the infinite sum
 \begin{align}
  \frac{(2\sigma^2)^{1/4}} {\sqrt{\pi(1-i2t')}} \sqrt{\frac{2 \pi} {N}}\sum_{j } e^{-\alpha x_j^2} \; .
 \end{align}
 Using similar tools as in previous proofs based on the Dirac comb
 of period $T=\sqrt{2\pi N}$, we obtain the error of the approximation 
 of the integral as
  \begin{align}
  \label{eq:gaussianerror2}
 \varepsilon_2= \frac{(2\sigma^2)^{1/4}} {\sqrt{2 \pi \alpha(1-i2t')}} \sum_{k \ne 0} e^{- \omega_k^2/(4 \alpha)} \; .
 \end{align}
 This error was obtained using the Fourier transform of $\exp (-\alpha x^2)$.
  The frequencies are $\omega_k = k \sqrt{2 \pi N}$.
  With our choices for $\sigma$, $t$, and $t'$, 
  we can bound
  \begin{align}
  \nonumber
  \varepsilon_2 & \le (\sigma^2 / (\pi |\alpha|))^{1/2} \sum_{k \ne 0} e^{- \omega_k^2/(4 \alpha)} \\
  \nonumber
&  \le \sum_{k \ne 0} e^{- \omega_k^2/(4 \Re( \alpha))} /\sqrt \pi \\ &
  \le \sum_{k \ne 0} e^{-N^2 k^2/(6 \delta)} /\sqrt \pi  \; ,
  \end{align}
  where we first used $|\alpha| \ge \Re (\alpha) \ge \sigma^2$ and then $\Re(\alpha) \le 3 \sigma^2$,
  which is valid under the assumption that $N \ge \delta \pi^2$ or $\sigma^2 \le 1/\pi$.
  Then, there is a constant $\lambda^{(2)}>0$ such that
 $ \varepsilon_2 = O(\exp(-\lambda^{(2)} N^2/\delta)) $.
   
%

 The next goal is to show that the finite sum
  \begin{align}
  \label{eq:gaussiandiscreteIP}
  \sqrt{\frac {2\pi} N} \sum_{j=-N/2}^{N/2-1} h(x_j) e^{i (x_j)^2 t} g(x_j) \; ,
  \end{align}
  which we showed is a good approximation to Eq.~\eqref{eq:innerproducthg}, also
  approximates the desired inner product between discrete states appearing in Eq.~\eqref{eq:gaussianerror}.
  Equation~\eqref{eq:gaussiandiscreteIP} can be realized as the inner product
  \begin{align}
    \sqrt{\frac {2\pi} N} \sbra{\tilde \xi^\rD} F^\rD_{\rm c}  e^{i (x^\rD)^2t }F^\rD_{\rm c} \sket{\tilde \varphi^\rD} \; .
  \end{align}
  with $ \sbra{\tilde \xi^\rD} F^\rD_{\rm c} \ket j = h(x_j)$ and $\bra j F^\rD_{\rm c} \sket{\tilde \varphi^\rD} = g(x_j)$.
  Using the properties of the discrete and CV Fourier transforms, the last condition implies
 $
  \sket{\tilde \varphi^\rD} = \sum_{j=-N/2}^{N/2-1} \tilde \varphi(x_j) \ket j
 $,
  with 
 $
  \tilde \varphi(x) = \sum_k \varphi(x + k T,0) $,
  i.e., the convolution of $\varphi$ with the Dirac comb.
  Equivalently, we are defining
  \begin{align}
   \sket{\tilde \varphi^\rD} = \sum_{j=-N/2}^{N/2-1} \ket j \bra{C_j} \varphi \rangle \; ,
  \end{align}
  as used in Corollary~\ref{cor:orthogonal}, Appendix~\ref{appendix:spectralproperties}.
  We note that
  \begin{align}
  \nonumber
  \tilde \varphi(x) - \varphi (x,0) &=  \sum_{k \ne 0} \varphi(x + k T,0) \\ &
  = O(N^{1/4} \exp(-\pi N/(2\sigma^2)) )
  \end{align}
  and then there is a constant $\lambda^{(3)}>0$ such that
  \begin{align}
  \varepsilon_3 & = \|  \sket{\tilde \varphi^\rD} - \sum_{j}   \varphi(x_j,0) \ket j \|
  = O(N^{5/4} \exp(-\lambda^{(3)} N^2)) \; .
  \end{align}
  We note that $\sum_{j}   \varphi(x_j,0) \ket j =(N/(2 \pi^2 \delta))^{1/4} \sqrt c \sket{\varphi^\rD}$, as defined in the statement
  of the Lemma.
  Similarly, the condition on $h(x_j)$ implies
  \begin{align}
  \sket{\tilde \xi^\rD} &= \sum_{j=-N/2}^{N/2-1} \ket j \bra{C_j} e^{-i \hat x^2 t'} \sket{\psi_0}
  =  \sum_{j=-N/2}^{N/2-1} \tilde \xi(x_j) \ket j \; ,
  \end{align}
  with 
   \begin{align}
\tilde   \xi(x)= \sum_k e^{-i (x+k T)^2 t'} e^{-(x+k T)^2/2}/\pi^{1/4} \; .
  \end{align}
  being the convolution with the Dirac comb.
  If $\xi(x) = e^{-i (x)^2 t'} e^{-(x)^2/2}/\pi^{1/4}$, then
  \begin{align}
  | \tilde \xi(x) - \xi(x) | = O(\exp (- \Omega(N))) \; .
  \end{align}
  This implies that there is a constant $\lambda ^{(4)}>0$ such that
  \begin{align}
  \varepsilon_4 & = \|  \sket{\tilde \xi^\rD} - \sum_{j}   \xi(x_j) \ket j \| 
 = O(N \exp(-\lambda^{(4)} N)) \; .
  \end{align}
  We then define  $\sket{\xi^\rD}=\sum_{j}   \xi(x_j) \ket j$ and, using the definition of $\sket{\psi_0^\rD}$,
  we obtain $\sket{\xi^\rD} =(N/(2\pi))^{1/4} e^{-i (x^\rD)^2 t'}\sket{\psi_0^\rD}$.
  
 In summary, so far we demonstrated that   Eq.~\eqref{eq:innerproducthg}, or Eq.~\eqref{eq:innerproductgaussian},
  can be approximated by
  \begin{align}
 \left(\frac{1}{\pi \delta} \right)^{1/4} \sqrt \kappa \bra{\psi_0^\rD} e^{i (x^\rD)^2 t'} F^\rD_{\rm c} e^{i(x^\rD)^2 t} F^\rD_{\rm c} \sket{\varphi^\rD} \; ,
  \end{align}
  within approximation error of order $\varepsilon_1+\varepsilon_2+\varepsilon_3+\varepsilon_4$,
  which can be obtained by using the subadditivity property of errors.
  
  There are two additional approximations that still need to be bounded.
  One is
  $ \varepsilon_5 = \| \sket{\psi_0^\rD} - (F^\rD_{\rm c})^2 \sket{\psi_0^\rD} \|$.
  Using the results in Eq.~\eqref{eq:convapprox}, we can show that $\sket{\psi_0^\rD}$
  is almost an eigenstate of $F^\rD_{\rm c}$ of eigenvalue. In particular, Eq.~\eqref{eq:convapprox} implies
   $\varepsilon_5 = O(\exp(-\beta N))$, for some constant $\beta>0$.
   Last, we seek to show that $\kappa$ approximates $\sqrt{\pi \delta}$.
   We write
  $ \sqrt{\pi \delta} = \sqrt{N/2\pi}\int dx \; e^{-N x^2/(2 \pi \delta)}$,
  which can be approximated  by the infinite sum $\sum_j e^{-j^2/\delta}$
  within approximation error
  \begin{align}
  \varepsilon_6 &=\sqrt{ \frac N {2\pi}} \sqrt{2 \pi} \sum_{k \ne 0} e^{- \omega_k^2 \pi \delta /(2N)}
  =  \sqrt{\pi \delta} \sum_{k \ne 0} e^{-\pi^2 k^2 \delta} \; ,
  \end{align}
  where $\omega_k = k \sqrt{2 \pi N}$, and $\varepsilon_6$ was obtained by computing the Fourier transform
  of a Gaussian with variance $\pi \delta/N$. Then, there is a constant $\lambda^{(6)}>0$ such that
  $\varepsilon_6 = O(\delta \exp(-\lambda^{(6)} \delta))$.
  Setting a cutoff in the infinite sum so that $|j| \le N/2$, this introduces an additional
  approximation error 
  $
  \varepsilon_7 = O(\exp(-\lambda^{(7)} N^2/\delta)) $,
  for some constant $\lambda^{(7)}>0$.

  The overall result is that we can approximate 
  Eq.~\eqref{eq:innerproducthg}, or Eq.~\eqref{eq:innerproductgaussian},
  by 
   \begin{align}
 \bra{\psi_0^\rD} (F^\rD_{\rm c})^2 e^{i (x^\rD)^2 t'} (F^\rD_{\rm c})^2 (F^\rD_{\rm c})^{-1} e^{i(x^\rD)^2 t} F^\rD_{\rm c} \sket{\varphi^\rD} \; ,
  \end{align}
  within precision $\sum_{i=1}^7 \varepsilon_i$. But since $[(F^\rD_{\rm c})^2 , (x^\rD)^2]=0$, and $F^\rD_{\rm c})^4=\one$,
  the above equation is
  exactly 
$ \bra{\psi_0^\rD} e^{i (x^\rD)^2 t'}    e^{i(p^\rD)^2 t}  \sket{\varphi^\rD}$,
  the desired quantity for the Lemma.
  In particular, this implies
  \begin{align}
  e^{i \alpha(t)} -  \bra{\psi_0^\rD} e^{i (x^\rD)^2 t'}    e^{i(p^\rD)^2 t}  \sket{\varphi^\rD}&\le \sum_{i=1}^7 \varepsilon_i\;, \\
  \| \sket{\psi_0^\rD} - e^{-i \alpha(t)} e^{i (x^\rD)^2 t'}    e^{i(p^\rD)^2 t}  \sket{\varphi^\rD} \| &\le \sum_{i=1}^7 \varepsilon_i \; .
  \end{align}
  
 \end{proof}

\section{The discrete Jaynes-Cummings Hamiltonian}
\label{appendix:DHJC}
\label{appendix:DHJC}
The discrete Jaynes-Cummings Hamiltonian $H_{\rm JC}^\rD$ was introduced in Eq.~\eqref{eq:DJCH}.
The Hilbert space dimension is $2N$, where $N$ is the dimension for $x^\rD$ and $p^\rD$.
In this Appx., we first  prove that the normalized states
 $
 \sket{\gamma^\rD_{n,\pm}} = \frac 1 {\sqrt 2} [\sket{\phi_n^\rD} \ket 0 \pm \sket{\phi_{n+1}^\rD} \ket 1 ] $,
where $n \le cN$ for some constant $c>0$, are almost eigenstates of $H_{\rm JC}^\rD$ of eigenvalues $\pm \sqrt{n+1}$.
\begin{lemma}
\label{lem:DJCH1}
There exists a constant $c>0$ such that
\begin{align}
\| H^\rD_{\rm JC} - (\pm \sqrt{n+1})   \sket{\gamma^\rD_{n,\pm}} \| = \nu_1(N) \; ,
\end{align}
for all $n \le N' \le cN$, where $\nu_1(N) =\exp (-\Omega( N))$.
\end{lemma}

\begin{proof}
First, we note
$H_{\rm JC}^\rD \sket{\psi_n^\rD} \ket 0 = (x^\rD - i p^\rD) \sket{\psi_n^\rD} \ket 1$.
As in Eq.~\eqref{eq:DPosAction}, $x^\rD$ satisfies, for all $n \le N$,
 \begin{align}
  \label{eq:Xaction}
  x^\rD \sket{\psi^\rD_n}   = \sqrt{(n+1)/2} \sket{\psi_{n+1}^\rD}+ \sqrt{n/2} \sket{\psi_{n-1}^\rD} \; .
  \end{align}
Lemma~\ref{lem:D-CVmomentum} implies
\begin{align}
 \label{eq:Paction}
\|  p^\rD \sket{\psi^\rD_n} - i \sqrt{(n+1)/2} \sket{\psi_{n+1}^\rD}+ i \sqrt{n/2} \sket{\psi_{n-1}^\rD} \| = 
\nu_1(N)  \; 
\end{align}
for all $n \le N' \le cN$, where $1>c>0$, and  $\nu_1(N) = \exp(-\Omega(N))$.
Then,
\begin{align}
\| H_{\rm JC}^\rD \sket{\psi_n^\rD} \ket 0 - \sqrt{n+1}  \sket{\psi_{n+1}^\rD} \ket 1 \| & =\nu_1(N)    \\
\| H_{\rm JC}^\rD \sket{\psi_{n+1}^\rD} \ket 1 - \sqrt{n+1}  \sket{\psi_{n}^\rD} \ket 0 \| &= \nu_1(N)   \; ,
\end{align}
where the last equation follows by doing a similar analysis.
Since $\| H_{\rm JC}^\rD \|= O(N^{1/2})$, we can replace  $\sket{\psi_n^\rD}$ by the actual eigenstates
$\sket{\phi_n^\rD}$ using Eq.~\eqref{eq:approxeigenstate}, and the order of the approximation error is also 
exponentially small in $N$.
It follows that
\begin{align}
\| H^\rD_{\rm JC} - (\pm \sqrt{n+1})   \sket{\gamma^\rD_{n,\pm}} \| = \nu_1(N) \; .
\end{align}

\end{proof}

Next, we seek efficient decompositions of the evolution operator induced by $H_{\rm JC}^\rD$.
For simplicity, we use an asymmetric first order Trotter-Suzuki approximation as this already
provides the desired scaling with $N$, however, one may use higher-order approximations
to improve upon the scaling on the precision parameter.

\begin{lemma}
\label{lem:DJCevol}
Let 
\begin{align}
W (s)= e^{-i(x^\rD \otimes \sigma_x) s/\sqrt 2} e^{i(p^\rD \otimes \sigma_y) s/\sqrt 2}  \; .
\end{align}
Then, there exists a constant $c>$ such that, given $n \le N' \le cN$ and
$s= O(\epsilon/\sqrt{n+1})$,
\begin{align}
\| (W(s) - e^{-i H_{\rm JC}^\rD s} ) \sket{\phi_n^\rD} \ket 0 \| =O(\epsilon^2) \; . 
\end{align}

\end{lemma}

\begin{proof}
We first define 
$\varepsilon_1(s) = W (s)e^{i H_{\rm JC}^\rD s} - \one$,
and, since $\varepsilon_1(s)=\int_0^s ds' \partial_{s'} \varepsilon_1(s')$, we obtain
\begin{align}
\nonumber
\varepsilon_1 (s)& = \int_0^{s}  ds' \; W(s')  \sum_{k \ge 1} \left ( \frac {i s' }{\sqrt 2} \right)^k \frac 1 {k!} \times \\ & \times
[(p^\rD \otimes \sigma_y),[(p^\rD \otimes \sigma_y),[...,(x^\rD \otimes \sigma_y)] \ldots ]]
e^{i H_{\rm JC}^\rD s'} \; .
\end{align}
Because $\| x^\rD\| = \|p^\rD \| = O(\sqrt N)$, we can cutoff the sum in $k$ at $K=O(\sqrt N)$
with an approximation error that is negligible (i.e., exponentially small in $\sqrt N$).
Since $K \le cN$, for some constant $c>0$, we then use Eqs.~\eqref{eq:Xaction} and~\eqref{eq:Paction}  to show
\begin{align}
\nonumber
& \| (p^\rD)^{k_1} (x^\rD)^{k_2} ... (p^\rD)^{k_{l-1}} (x^\rD)^{k_l} \sket{\phi_n^\rD} \| = \\ & =
O(\sqrt{(n+1)\ldots (n+k)}) \;,
\end{align}
where $\sum_i k_i =k \le K$. Since $\sket{\gamma^\rD_{n,\pm}}$ are approximate eigenstates of $H_{\rm JC}^\rD$ when $n \le N' \le cN$,
and these are combinations
of $\sket{\phi_n^\rD} \ket 0$ and $\sket{\phi_{n+1}^\rD} \ket 1$, we obtain
\begin{align}
\nonumber
& \| \varepsilon_1(s) \sket{\gamma^\rD_{n,\pm}} \|  = \\ \nonumber
& = O\left( \int_0^s ds \sum_{k = 1}^K ((s' \sqrt 2)^k /k!) \sqrt{(n+2)\ldots (n+k+1)} \right)  \\
&=
O\left(  \sum_{k = 1}^K (s^{k+1} (\sqrt 2)^k /(k+1)!) \sqrt{(n+2)\ldots (n+k+1)} \right)\; ,
\end{align}
where we assumed that other contributions to the error that are exponentially small in $N$ or $\sqrt N$
are negligible in the asymptotic limit. We use the bound on the binomial coefficient
$\begin{pmatrix} m \\ k \end{pmatrix} < m^k/k!$,
to obtain
\begin{align}
\| \varepsilon_1(s) \sket{\gamma^\rD_{n,\pm}} \| = O\left(  \sum_{k = 1}^K (s^{k+1} (\sqrt 2)^k /(k)!)   (n+k+1)^{k/2} \right)\; .
\end{align}
If $ s \le \epsilon/\sqrt{2(n+1)}$, Stirling's approximation implies
\begin{align}
\| \varepsilon_1(s) \sket{\gamma^\rD_{n,\pm}} \| = O\left(  \sum_{k = 1}^K (\epsilon^{k+1} e^k)/k^{k/2} \right)\; ,
\end{align}
and then
$\| \varepsilon_1(s) \sket{\gamma^\rD_{n,\pm}} \| = O\left( \epsilon^2 \right)$.
\end{proof}


\begin{thebibliography}{000}

\bibitem{feynman_simulating_1982}
R.~P. Feynman (1982),
{\em Simulating physics with computers},
 International Journal of Theoretical Physics, Vol. 21(6), pp. 467--488.
 
 
 \bibitem{lloyd_universal_1996}
Seth Lloyd (1996),
{\em Universal quantum simulators},
Science, Vol. 273, pp. 1073--1078.

\bibitem{BMK2010}
Vivien M~Kendon Katherine L~Brown, William J~Munro (2010),
{\em Using quantum computers for quantum simulation},
 Entropy,  Vol. 12, pp. 2268--2307.
 
 \bibitem{ortiz:qc2001a}
G.~Ortiz, J.~E. Gubernatis, E.~Knill, and R.~Laflamme (2001),
{\em Quantum algorithms for fermionic simulations},
 Phys. Rev. A, Vol. 64, p. 022319.
 
 \bibitem{somma_physics_2003}
Rolando Somma, Gerardo Ortiz, Emanuel Knill, and James Gubernatis (2003),
{\em Quantum simulations of physics problems},
 Int. J. of Quant. Inf., Vol. 1, p. 189.
 
 \bibitem{aspuru-guzik:qc2005a}
A.~Aspuru-Guzik, A.~D. Dutoi, P.~J. Love, and M.~Head-Gordon (2005),
{\em Simulated quantum computation of molecular energies},
 Science, Vol. 309, pp. 1704--1707.
 
 \bibitem{kassal_chem_2011}
I.~Kassal, J.~D. Whitfield, A.~Perdomo-Ortiz, M.-H. Yung, and A.~Aspuru-Guzik (2011),
{\em Simulating chemistry using quantum computers},
Annual Review of Physical Chemistry, Vol. 62, p. 185.


\bibitem{wecker_chem_2013}
Dave Wecker, Bela Bauer, Bryan~K. Clark, Matthew~B. Hastings, and Matthias
  Troyer (2013),
{\em Can quantum chemistry be performed on a small quantum computer?},
 arXiv:1312.1695.
 
 \bibitem{PHW2014}
David Poulin, M.~B. Hastings, Dave Wecker, Nathan Wiebe, Andrew~C. Doherty, and
  Matthias Troyer (2014),
{\em The Trotter step size required for accurate quantum simulation of
  quantum chemistry},
 quant-ph/1406.4920.
 
  \bibitem{DJRZ2006}
Josep Diaz, Klaus Jansen, Jose~D.P. Rolim, and Uri Zwick, editors (2006).
\newblock {\em Approximation, Randomization, and Combinatorial Optimization.
  Algorithms and Techniques.}
\newblock Springer-Berlag.

\bibitem{GACZ2000}
L.âJ. Garay, J.âR. Anglin, J.âI. Cirac, and P.~Zoller (2000),
{\em Sonic analog of gravitational black holes in Bose-Einstein
  condensates},
 Phys. Rev. Lett.,  Vol. 85, p. 4643.
 
 \bibitem{Fri2008}
A.~Friedenauer, H.~Schmitz, J.T. Glueckert, D.~Porras, and T.~Schaetz (2008),
{\em Simulating a quantum magnet with trapped ions},
 Nature Phys., Vol. 4, p. 757.
 


\bibitem{aharonov_adiabatic_2003}
D.~Aharonov and A.~Ta-Shma (2003),
{\em Adiabatic quantum state generation and statistical zero knowledge},
  Proceedings of the 35th ACM Symposium on Theory of Computing,
  pp. 20--29.

\bibitem{suzuki_90}
M.~Suzuki (1990),
{\em Fractal decomposition of exponential operators fractal decomposition
  of exponential operators with applications to many-body theories and Monte
  Carlo simulations},
 { Phys. Lett. A}, Vol. 146, p. 319.


\bibitem{berry_efficient_2007}
Dominic Berry, Graeme Ahokas, Richard Cleve, and Barry Sanders (2007), {\em
Efficient quantum algorithms for simulating sparse Hamiltonians}, 
  Comm. Math. Phys., Vol.  270, p. 359.
  
  \bibitem{childs_thesis}
A.M Childs (2004)
\newblock {\em Quantum information processing in continuous time}.
\newblock PhD thesis, Massachusetts Institute of Technology, Cambridge, MA, US.

\bibitem{childs_efficient_2010}
A.~Childs and R.~Kothari (2011),
{\em Simulating sparse Hamiltonians with star decompositions}, 
 Theory of Quantum Computation, Communication, and Cryptography,
  p. 94.
  
  \bibitem{wiebe_product_2010}
Nathan Wiebe, Dominic Berry, Peter Hoyer, and Barry~C. Sanders (2010),
{\em Higher-order decompositions of ordered operator exponentials},
 { J. Phys. A: Math. Theor.}, Vol. 43, p. 065203.
 
 \bibitem{trotter_1959}
H.F. Trotter. (1959),
{\em On the product of semigroups of operators},
 { Proc. Am. Math. Phys.}, Vol. 545.

 
 \bibitem{berry_sparse}
D.W. Berry, A.M. Childs, R.~Cleve, R.~Kothari, and R.D. Somma (2014),
{\em Exponential improvement in precision for simulating sparse
  Hamiltonians},
 Proceedings of the 2014 ACM Symposium on Theory of Computing,
  pp. 283--292.

\bibitem{BCC_Taylor_2015}
D.W. Berry, A.M. Childs, R.~Cleve, R.~Kothari, and R.D. Somma (2015),
{\em Simulating Hamiltonian dynamics with a truncated Taylor series}, 
  Phys. Rev. Lett., Vol. 114, p. 090502.
  
  \bibitem{MM2014}
John~M. Martinis and A.~Megrant (2014),
{\em UCSB final report for the CSQ program: Review of decoherence and
  materials physics for superconducting qubits},
 Technical report, University of California Santa Barbara.

\bibitem{PW2009}
David Poulin and Pawel Wocjan (2009),
{\em Preparing ground states of quantum many-body systems on a quantum
  computer},
 Phys. Rev. Lett., Vol. 102, p. 130503.

\bibitem{schwarz_peps_2011}
Martin Schwarz, Kristan Temme, and Frank Verstraete (2012),
{\em Preparing projected entangled pair states on a quantum computer},
 Phys. Rev. Lett., Vol. 108, p. 110502.

\bibitem{somma_gap_2013}
Rolando~D. Somma and Sergio Boixo (2013),
{\em Spectral gap amplification},
 SIAM J. Comp, Vol. 42, pp. 593--610.
 
 
\bibitem{farhi_quantum_2000}
Edward Farhi, Jeffrey Goldstone, Sam Gutmann, and Michael Sipser (2000),
{\em Quantum computation by adiabatic evolution},
 quant-ph/0001106.


 
 \bibitem{kitaev_quantum_1995}
A.~Yu Kitaev (1995),
{\em Quantum measurements and the abelian stabilizer problem},
 arxiv:quant-ph/9511026.

\bibitem{kitaev_computation_2002}
A.~Yu Kitaev, A.H. Shen, and M.N. Vyalyi (2002),
\newblock {\em Classical and Quantum Computation}.
\newblock American Mathematical Society.

\bibitem{papageorgiou_CV_2013}
A.~Papageorgiou, I.~Petras, J.F. Traub, and C.~Zhang (2013),
{\em A fast algorithm for approximating the ground state energy on a
  quantum computer},
 Mathematics of Computation, Vol. 82, pp. 2293--2304.
 
 \bibitem{jordan_QFT_2012}
Stephen~P. Jordan, Keith S.~M. Lee, and John Preskill (2012),
{\em Quantum algorithms for quantum field theories},
 Science, Vol. 336, pp. 1130--1133.
 
%
 \bibitem{knill_expectation_2007}
Emanuel Knill, Gerardo Ortiz, and Rolando Somma (2007),
{\em Optimal quantum measurements of expectation values of observables},
 Phys. Rev. A, Vol. 75, p. 012328.


\bibitem{Som_15}
R.D. Somma (2015),
{\em A Trotter Suzuki Approximation for Lie Groups with Applications to Hamiltonian Simulation},
arXiv:1512.03416.



\bibitem{OZA01}
Haldun~M. Ozaktas, Zeev Zalevsky, and M.~Alper Kutay (2001),
\newblock {\em The Fractional Fourier Transform: with Applications in Optics
  and Signal Processing}.
\newblock John Wiley and Sons, London, UK.


\bibitem{AB95}
Arthur Beiser (1995),
\newblock {\em Concepts of Modern Physics},
\newblock McGraw-Hill, 5th edition.



\bibitem{cleve_query_2009}
R.~Cleve, D.~Gottesman, M.~Mosca, R.D. Somma, and D.L. Yonge-Mallo (2009),
{\em Efficient discrete-time simulations of continuous-time quantum query
  models},
  Proceedings of the 41st Annual IEEE Symp. on Theory of
  Computing, pp. 409--416.
  
  \bibitem{WBHS11}
Nathan Wiebe, Dominic W. Berry, Peter Høyer, and Barry C Sanders (2011),
{\em Simulating quantum dynamics on a quantum computer},
 J. Math. A: Math. Theor., Vol. 44, p. 445308.
  

  

\bibitem{nielsen_quantum_2000}
M.~A. Nielsen and I.~L. Chuang (2000),
\newblock {\em Quantum Computation and Quantum Information},
\newblock Cambridge University Press, Cambridge, UK.




\bibitem{Huyghebaert_1990}
J.~Huyghebaert and H.~De Raedt (1990),
{\em Product formula methods for time-dependent Schr\"odinger problems},
 J. Phys. A: Math. Gen., Vol. 23, p. 5777.



\bibitem{suzuki_qmc_1998}
M.~Suzuki (1998),
\newblock {\em Quantum Monte Carlo Methods in Condensed Matter Physics}.
\newblock World Scientific, Singapore.




\bibitem{somma_physics_2002}
Rolando Somma, Gerardo Ortiz, James Gubernatis, Emanuel Knill, and Raymond
  Laflamme (2002),
{\em Simulating physical phenomena by quantum networks},
 Phys. Rev. A, Vol. 65, p. 042323.



\bibitem{KW2008}
A.~Kitaev and W.A. Webb (2008),
{\em Wavefunction preparation and resampling using a quantum computer},
 arXiv:0801.0342.
 
 \bibitem{GR2002}
 L. Grover and T. Rudolph (2002),
 {\em Creating superpositions that correspond to efficiently integrable probability distributions},
 quant-ph/0208112.


\bibitem{Alm94}
L.B. Almeida (1994),
{\em The fractional fourier transform and time-frequency representations},
  IEEE Trans. Sig. Proc., Vol. 42, p. 3084.




\bibitem{LZ03}
H.~Liu and M.~Zhu (2004),
{\em Applying fractional Fourier transform to radar imaging of moving
  targets},
 in IEEE Int. Symp. on Geoscience and Remote Sensing.





\bibitem{Jou04}
I.~I. Jouny (2003),
{\em Radar backscatter analysis using fractional Fourier transform},
in IEEE Symp. on Antennas and Propagation Society.






\bibitem{TA84}
T.~R. Taha and M.~J. Ablowitz (1984),
{\em Analytical and numerical aspects of certain nonlinear evolution
  equations. ii. numerical, nonlinear Schr{\"o}dinger equation},
 { J. Comput. Phys.}, Vol. 55, p. 203.




\bibitem{yung_chem_2012}
M.-H. Yung, J.~D. Whitfield, S.~Boixo, D.~G. Tempel, and A.~Aspuru-Guzik (2012),
{\em Introduction to quantum algorithms for physics and chemistry},
 arXiv:1203.1331.




\bibitem{IntTable}
I.S. Gradshteyn and I.M Ryzhik (2007),
\newblock {\em Table of Integrals, Series, and Products}.
\newblock Academic Press, MA, USA.



\end{thebibliography}
\end{document}